%% file: 00-main-lics.tex
\documentclass{lmcs} 
\pdfoutput=1
\usepackage[utf8]{inputenc}

\usepackage{lastpage}
\lmcsdoi{22}{3}{4}
\lmcsheading{}{\pageref{LastPage}}{}{}%
{Oct.~24,~2024}{Jul.~21,~2026}{}


\usepackage{enumitem}
\usepackage{amssymb}
\usepackage{thm-restate}
\usepackage{stmaryrd}
\usepackage[all]{xy}
\usepackage[backref, hypertexnames=false]{hyperref}

\usepackage{comment}
\specialcomment{proofdetail}{\color{blue}}{\color{black}}
\excludecomment{proofdetail}

\newcommand{\typd}[2] {
#1\hspace{-.75mm}:\hspace{-.75mm}#2}

\newcommand{\lbd}[3] {
\lambda{#1}\hspace{-.75mm}:\hspace{-.75mm}{#2}\hspace{.35mm}.\hspace{.35mm}{#3}
}

\newcommand{\Lbd}[3] {
\Lambda{#1}\hspace{-.75mm}:\hspace{-.75mm}{#2}\hspace{.35mm}.\hspace{.35mm}{#3}
}

\newcommand{\cbd}[3] {
[{#1}\hspace{-.75mm}:\hspace{-.75mm}{#2}]\hspace{.35mm}.\hspace{.35mm}{#3}
}

\newcommand{\param}{\theta}

\newcommand{\db}[1]{\llbracket#1\rrbracket}

\newcommand{\mlt}[1]{\boldsymbol\lambda_{#1}}
\newcommand{\mltup}{\boldsymbol\lambda_{\param}}
\newcommand{\mltpure}{\mlt{\omega}}
\newcommand{\lt}{\boldsymbol\lambda}

\newcommand{\clt}[1]{CL_{#1}}
\newcommand{\cltup}{CL_{\param}}
\newcommand{\cltpure}{CL}

\newcommand{\reduce}[1]{\rightarrow_{#1}}
\newcommand{\reduces}[2]{\rightarrow_{#1}^{#2}}
\newcommand{\reducesbw}[2]{\leftarrow_{#1}^{#2}}
\newcommand{\equal}[1]{=_{#1}}
\newcommand{\equals}[2]{=_{#1}^{#2}}

\newcommand{\transreduce}[1]{\twoheadrightarrow_{#1}}
\newcommand{\transreduces}[2]{\twoheadrightarrow_{#1}^{#2}}

\newcommand{\Reduces}[2]{\Rightarrow_{#1}^{#2}}

\newcommand{\ctI}[2]{{\mathsf I}_{#1}#2}

\newcommand{\ctK}[4]{{\mathsf K}_{#1, #2}#3#4}
\newcommand{\ctS}[6]{{\mathsf S}_{#1, #2, #3}#4#5#6}
\newcommand{\ctC}[6]{{\mathsf C}_{#1, #2, #3}#4#5#6}
\newcommand{\ctD}[6]{{\mathsf D}^{#1}_{#2, #3, #4}#5#6}
\newcommand{\ctW}[4]{{\mathsf W}_{#1, #2}#3#4}
\newcommand{\ctB}[6]{{\mathsf B}_{#1, #2, #3}#4#5#6}

\begin{document}

\title[Simply-typed constant-domain modal lambda calculus~I]{Simply-typed constant-domain modal lambda calculus~I: distanced beta reduction\texorpdfstring{\\}{} and combinatory logic}

\thanks{Thanks to Thomas Ede Zimmermann, Dana Scott, and the reviewers for comments and feedback.}

\author[S.~Walsh]{Sean Walsh\lmcsorcid{0009-0007-5460-778X}}

\address{}

\address{UCLA Department of Philosophy \\ 390 Portola Plaza \\ 300 Dodd Hall \\ Box 951451 \\ Los Angeles, CA 90095-1451}

\email{walsh@ucla.edu}  



\begin{abstract}
A system $\mltup$ is developed that combines modal logic and simply-typed lambda calculus, and that generalizes the system studied by Montague and Gallin. Whereas Montague and Gallin worked with Church's simple theory of types, the system $\mltup$ is developed in the typed base~theory most commonly used today, namely the simply-typed lambda calculus. Further, the system $\mltup$ is controlled by a parameter $\param$ which allows more options for state types and state variables than is present in Montague and Gallin. A main goal of the paper is to establish some basic metatheory of $\mltup$: (i) an Andrews-like characterization of its models in terms of combinatory logic is given, and this combinatory logic involves a $\mathsf{BCKW}$-like basis rather than an $\mathsf{SKI}$-like basis and (ii) semantic conservation and expressibility results relating $\mltup$ to the maximal system $\mltpure$ are proven. Similar results are proven for the relation between $\mltpure$ and~$\lt$, the corresponding ordinary simply-typed lambda calculus. This answers a question of Zimmermann in the semantics of the simply typed setting. In a companion paper this is extended to Church's simple theory of types. We further develop a partial correspondence between a pure combinatory logic centered on the $\mathsf{BCKW}$-like basis and the weak deductive system for $\mltpure$ wherein $\beta$-reduction is not allowed under a lambda abstract, and we use this to show partial deductive conservation between the maximal system $\mltpure$ and the intermediary systems $\mltup$.
\end{abstract}

\maketitle




\input{01-intro/intro}


\input{02-systems/00-intro}

\input{02-systems/01-types}

\input{02-systems/02-terms}

\input{02-systems/03-nonmodal}

\input{02-systems/04-conventionsreduction}

\input{02-systems/05-alpha}

\input{02-systems/06-betaeta}

\input{02-systems/07-CRandSN}


\input{03-soundness/00-intro}

\input{03-soundness/01-semantics}

\input{03-soundness/02-soundness-validity-reductions}

\input{03-soundness/03-completeness-open-term}

\input{03-soundness/04-automorphisms}


\input{04-combinatory/00-intro}

\input{04-combinatory/01-typed-combinators}

\input{04-combinatory/02-BCKW}

\input{04-combinatory/03-combinatorial-models}

\input{04-combinatory/04-conservation}

\input{04-combinatory/05-deductive}


\input{05-ty2/00-intro}

\input{05-ty2/01-conserve}

\input{05-ty2/02-termredeux}


\input{06-pure-combinatory/00-intro}

\input{06-pure-combinatory/01-typed-combinators-CL}

\input{06-pure-combinatory/02-CR}

\input{06-pure-combinatory/03-simulate-abstraction}

\input{06-pure-combinatory/04-lambda-calculus-to-CL}

\input{06-pure-combinatory/05-CL-to-lambda-calculus}

\input{06-pure-combinatory/06-completeconserve}


\bibliographystyle{alphaurl}
\bibliography{bibliography.bib}
\appendix

\end{document}

%% file: 01-intro/intro.tex
\section{Introduction}\label{sec:intro}

Two of the great achievements of modern logic are modal logic and typed lambda calculus. At the advent of formal semantics in linguistics, Montague developed a system that integrated the two.\footnote{Montague's work \cite{Montague1974-hm} is discussed at length in standard semantics textbooks, such as \cite{Dowty1981-gt}, \cite{Gamut1991-gc}, and \cite{Chierchia2000-zn}. Montague's work was made well-known in part through the work of Partee; see \cite{Partee1997-ou} for some of the history. The theory is sometimes divided into the intensional theory of types and Montague grammar (e.g. \cite[Chapters 5-6]{Gamut1991-gc}). This paper focuses on the intensional theory of types, as did Chapters 1-2 of Gallin's book \cite{Gallin1975-tp}. In recent decades, textbook treatments of semantics focus foremost on Montague grammar in extensional contexts. That is the topic of \cite{Heim1998-bs}, with its anticipated sequel \cite{Von-Fintel2011-zs} being devoted to intensional matters.} However, by contemporary lights, Montague's theory is both too strong and too weak. It is too strong in that he worked only with Church's simple theory of types, replete with the resources of quantification and identity.\footnote{Church \cite{Church1940-tj} took quantification as primitive and defined identity; Henkin \cite{Henkin1963-km} did it the other way around, and Henkin's approach is in e.g. Andrews' book \cite[Chapter 5]{Andrews2013-na}.} But modern typed lambda calculi work with a weaker base system, and have many different extensions besides Church's simple theory of types.\footnote{See \cite{Barendregt2013-eb} for an authoritative modern treatment of the base system and its extension to intersection types and recursive types. See \cite{Barendregt1992-gl}, \cite{Nederpelt2014-yk} for lengthy treatises on dependently-typed lambda calculus, including Coquand and Huet's Calculus of Constructions \cite{Coquand1985-yo}, \cite{Coquand1988-ez}, which was an important predecessor to the Lean proof verification system (cf. \cite{Ebner2017-za}); a distinct but related branch of dependently-typed lambda calculus is Martin-L\"of Type Theory (cf. \cite{Martin-Lof1984-uj}). Simply-typed lambda calculus is the internal logic of Cartesian closed categories, and see \cite{Lambek1988-hb} for the generalization to topoi. See \cite{Winskel1994-ha}, \cite{Harper2016-th} for systematic contemporary treatments of the denotational and operational semantics for programming languages, developed initially by Scott, Strachey, and Plotkin (\cite{Scott1971-jz}, \cite{Scott1993-bw}, \cite{Plotkin1977-pm}, \cite{Plotkin2004-gi}).} Further, Montague's theory is too weak in that it does not have many of the features of modern modal logics, such as two-dimensionality and actuality operators and other devices for referring to many distinct states, and binding many distinct variables of state type, within one and the same expression.\footnote{For two-dimensional semantics and actuality operators, see \cite{Davies1980-ym}, \cite{Nimtz2017-qu}. For hybrid logics and multiple state variables, see \cite{Areces2006-ma}, \cite[Part I]{Cresswell1990-bx}. In recent joint work with K\"opping \cite{Kopping2020-xz}, Zimmermann developed an extension of Montague's original system, in the standard semantics, which contains two-dimensionality and multiple state types.} (Following common usage in modal logic, ``state'' is a term of art which, depending on application, covers worlds, times, machine-configurations, etc.) A chief aim of this paper is to remedy this deficit, and to begin the development of a thoroughly modern version of Montague's simply-typed modal lambda calculus. This should be of interest wherever modal logic and simply-typed lambda calculus and related systems are used, be in intensional semantics, in higher-order metaphysics, or in program verification.\footnote{For higher-order metaphysics, see \cite{Williamson2013-lw}, \cite{Fritz2024-wo}. Outside of the operational semantics, two other important paradigms of program verification are propositional modal logic (e.g. \cite{Clarke2018-ni}), and dependently-typed lambda calculus (e.g. \cite{Bertot2013-vc}).} 

Another goal of this paper and its companion is to answer one of the outstanding questions about the metatheory of Montague's original system.\footnote{The companion paper is \cite{Walsh2024-mltII}.} Zimmermann showed in 1989 that Montague's simply-typed modal lambda calculus was expressively rich in the \emph{standard semantics}, in that its ostensibly more limited vocabulary could express anything expressible in the usual simply-typed lambda calculus formed with an additional atomic type for states.\footnote{\cite{Zimmermann1989-gs}.} In other words, in the standard semantics, Zimmermann showed that by using simply-typed lambda calculus with an atomic type for possible worlds one cannot really say or assert anything above and beyond what one can say or assert in the object-language of Montague's modal logic itself, i.e. with ordinary statements of possibility and contingency.\footnote{\cite[p. 75]{Zimmermann1989-gs}. This result is discussed extensively by Zimmermann in \cite{Zimmermann2020-an}; it is formally stated as the second theorem on \cite[p. 31]{Zimmermann2020-an}. In work with K\"opping, Zimmermann extended his result to settings with more than one kind of state type, cf. discussion of their equation~(10) in \cite[pp. 171, 175-176]{Kopping2020-xz}.} But the standard semantics is both incomplete and highly set-theoretically entangled since its validity relation is not recursively enumerable; hence any result about it may well just be a result about the ambient set theory and need not match up with what one can express with primitive rules for the system. Zimmermann asked in 1989 whether his result would generalise to the \emph{Henkin semantics}.\footnote{\cite[\S{4.2} pp. 75-76]{Zimmermann1989-gs}. For other discussions in the Montagovian tradition friendly to the Henkin semantics, see \cite[\S{3.1.1} p. 316]{Partee1977-vr}, \cite[p. 98]{Janssen1983-rm}. The distinction between the Henkin semantics and the standard semantics comes up in all discussions of higher-order logic and related systems, see e.g. \cite{Shapiro1991-oe}, \cite{Button2018-ux}.} In this and the companion paper, I resolve Zimmermann's question, largely in the affirmative.\footnote{The caveat ``largely'' is due to the presence of description axioms and constant symbols for description operators. If both are included in the right way, the answer is affirmative. The general situation is more complicated. See \cite{Walsh2024-mltII} for more details.} 

The division between the two papers is as follows: in this present paper I focus on the simply-typed lambda calculus, and in the companion paper I extend the results to Church's simple theory of types (again, the key difference is that the latter includes identity and quantification).

The simply-typed modal lambda calculus is designated with $\mltup$, and there are as many of these systems as there are choices of atomic types and choices of signatures. The atomic \emph{state} types of $\mltup$ are not allowed to be the codomain of a functional type; and $\mltup$ is further controlled by the parameter~$\param$ which dictates how many variables the state types have (see \S\S\ref{subsec:typesandterms}-\ref{subsec:terms} for formal definition). Due to the lack of variables, the usual proofs of completeness for $\beta\eta$-equality using open term models are not in general available (cf. Example~\ref{exa:failureopenterm1}). One special case where it is available is when parameter $\param$ is set to countably infinitely many variables for each state type (the maximal setting), in which case $\mltup$ is written as $\mltpure$. The first part of the resolution to Zimmermann's question in the semantics of the simply-typed setting is the following pair of theorems (proven in \S\ref{sec:conservation}):

\begin{restatable}[Semantic conservation of $\mltpure$ over $\mltup$]{thm}{thmconservation}\label{thmconservation} 
Every model of $\mltup$ is also a model of $\mltpure$. Hence for terms $M,N$ of $\mltup$, one has $\mltup\models M=N$ iff $\mltpure\models M=N$.
\end{restatable}
\begin{restatable}[Semantic expressibility of $\mltpure$ in $\mltup$]{thm}{express}\label{thm:express}
Suppose that term $\typd{N}{A}$ of $\mltpure$ is such that its free variables and constants are those of~$\mltup$. Then there is a term $\typd{M}{A}$ of $\mltup$ with the same free variables and constants such that $\mltpure\models M=N$.
\end{restatable}
\noindent In these theorems, the semantic validity relation $\models M=N$ just means sameness of denotation of $M,N$ for all models of the relevant system and for all variable assignments. Since $\mltpure$ is maximal, these theorems also imply semantic conservation and expressibility results between nested intermediary systems of~$\mltup$.

The ordinary non-modal simply-typed lambda calculus is designated as $\lt$, and there are as many of these systems as there are choices of atomic types and choices of signatures. Each system of simply-typed modal lambda calculus $\mltpure$ is equipped with a choice of atomic types, and this then induces a system of $\lt$. The types of $\mltpure$ are a small subset of the types of $\lt$: for instance if $A$ is a state type, then $\mltpure$ does not have a type $A\rightarrow A$, but $\lt$ has this type (see \S\ref{sec:nonmodal} for formal definition). The second part of the resolution to Zimmermann's question in the simply-typed setting is the following pair of theorems (proven in \S\ref{sec:conserlt}):

\begin{restatable}[Deductive conservation of $\lt$ over $\mltpure$]{thm}{thmconservationlt}\label{thmconservationlt}
For types $A$ of $\mltpure$ and terms $\typd{M,N}{A}$ of $\mltpure$, one has $\mltpure\vdash_{\beta\eta} M=N$ iff $\lt\vdash_{\beta\eta} M=N$.
\end{restatable}

\begin{restatable}[Deductive expressibility of  $\lt$ in $\mltpure$]{thm}{expresslt}\label{thm:expresslt}
Suppose that $A$ is a type of $\mltpure$ and that term $\typd{N}{A}$ of $\lt$ is such that its free variables and constants have types in~$\mltpure$. Then there is a term $\typd{M}{A}$ of $\mltpure$ with the same free variables and constants such that $\lt\vdash_{\beta\eta} M=N$.
\end{restatable}
\noindent In these theorems, the deductive relation $\vdash_{\beta\eta} M=N$ just means $\beta\eta$-equality of $M,N$ in the relevant system. By the Completeness Theorem for $\lt$ and the Completeness Theorem for $\mltpure$  (cf. Theorem~\ref{thmcompletelambda}), we can also equivalently express conservation and expressibility in terms of model-theoretic validity. Putting these four theorems together, one also has a semantic conservation result for~$\lt$ over~$\mltup$ and a semantic expressibility result for~$\lt$ in~$\mltup$. Theorems~\ref{thmconservationlt}-\ref{thm:expresslt} are proven in \S\ref{sec:ty2}, and the proofs are comparatively short since one can make use of all the known tools of the ordinary simply-typed lambda calculus, like Church-Rosser and strong normalization. 

The proofs of Theorems~\ref{thmconservation}-\ref{thm:express} go through a characterization of Henkin models of $\mltup$ in terms of combinatory logic. This characterization, proven in \S\ref{sec:combinatorial-models}, generalizes Andrews' similar characterization for $\lt$. Like Andrews' result, this characterization gives ``the way out of Henkin mysterious conditions that all $\lambda$-terms must have a denotation'':\footnote{This apt praise for Andrews is from Dowek \cite[p. 255]{Dowek2009-mf}. Andrews' original result is \cite[Proposition 4, Theorem 1 pp. 390-391]{Andrews1972-xk}; see \cite[Proposition 3.1.19(iii) p. 101]{Barendregt2013-eb}.}
\begin{restatable}[Combinatorial characterisation of models]{thm}{corcreatemodel}\label{corcreatemodel}
If $\mathcal{M}$ is a frame, then $\mathcal{M}$ is a model of $\mltup$ iff the denotations of all the $\mathsf{BCDKW}$-combinatorial terms of $\mltup$ are well-defined.
\end{restatable}
\noindent The combinatory logic used to prove this result deploys a $\mathsf{BCKW}$-basis, to which we add a combinator $\mathsf{D}$, rather than an $\mathsf{SKI}$-basis. See Question~\ref{qu:whatabouts} for whether there is any deeper necessity in this.

A chief difficulty in proving the theorem is that one has to work with a more complicated notion of $\beta$-reduction. Let us denote by $\beta_0$ the usual beta reduction $\big(\lbd{v}{A}{L}\big)N\reduce{\beta_0} L[v:=N]$, subject to the usual constraints. Intuitively, it is an expression of how an input $N$ gets processed according to a rule $\lbd{v}{A}{L}$: namely one systematically replaces free instances of $v$ in $L$ by $N$ to form the term $L[v:=N]$. The more general version which seems necessary is \emph{a distanced version of beta reduction}
$\big( \lbd{\vec{x}}{\vec{B}}{\lbd{v}{A}{L}}\big) \vec{M} N \reduce{\beta} \big( \lbd{\vec{x}}{\vec{B}}{L[v:=N]} \big) \vec{M}$ (which is subjected to natural constraints, cf. Definition~\ref{defn:betareduction}). This is ``distanced'' in that the input $N$ is separated from the lambda abstract $\lbd{v}{A}{L}$ by the intermediary terms~$\vec{M}$. The usual beta reduction $\beta_0$ is just ``distance zero'' $\beta$-reduction.\footnote{While the reductions $\beta_0$ and $\beta$ are different, we do not know if the equalities are different (cf. Question~\ref{qu:betazerobetarelation} and Proposition~\ref{prop:oneway}).}

The main theorems of this paper can be used to produce a new rule which when added to~$\beta\eta$-reduction allows~$\mltup$ to be deductively complete. Namely, for terms $M,N:A$ of~$\mltup$, say that term~$M$ reduces to term~$N$ by \emph{the $\infty$-rule} iff~$M$ is $\beta\eta$-equal to~$N$ in~$\mltpure$ (cf. Definition~\ref{defn:inftyrule}). Then the Completeness Theorem~\ref{thmcompletelambda} for~$\mltpure$ together with the semantic conservation result (Theorem~\ref{thmconservation}) implies the following, which is proven in \S\ref{sec:infty}:
\begin{restatable}[Completeness of $\mltup$ with $\infty$-rule]{thm}{thmcompletemltup}\label{thmcompletemltup}
For terms $\typd{M,N}{A}$ of $\mltup$, one has $\mltup\vdash_{\beta\eta\infty} M=N$ iff $\mltup\models M=N$.

\end{restatable}

\noindent Further one can show that when~$\param$ has at least two variables for each state type, then~$\mltup$ has an open term model formed with respect to~$\beta\eta\infty$-equality (cf. \S~\ref{sec:opentermmodelsredeux}). The natural open question is whether~$\mltup$ has a Completeness Theorem for~$\beta\eta$-equality, or equivalently whether~$\mltpure$ is conservative over~$\mltup$ for~$\beta\eta$-equality (cf. Question~\ref{q:dedcompcon}). By the methods used to establish~Theorem~\ref{corcreatemodel}, these two questions are equivalent to the question of whether~$\mltup$ and $\mltpure$ can prove all of the same $\beta\eta$-equalities made up from applications, variables in $\mltup$, constants, and $\mathsf{BCDKW}$-combinators in $\mltup$.

We can establish partial conservation of $\mltpure$ over $\mltup$ for $\beta$-equality and $\beta\eta$-equality in a large class of cases. Recall that in the lambda calculus, reduction relations are traditionally taken to be preserved under lambda abstracts: that is, if $M$ reduces to $N$, then $\lbd{x}{A}{M}$ reduces to $\lbd{x}{A}{N}$; and similarly, equality relations are traditionally taken to be preserved under lambda abstracts (this is called the $\xi$-rule). In so-called weak lambda calculus,\footnote{See \c{C}a{\u{g}}man-Hindley \cite{Cagman1998-pi} for the untyped setting, and Sestini \cite{Sestini2019-cc} for the typed setting.} one restricts this rule,\footnote{See the outset of \S\ref{sec:purecombinatory} for examples.} and a $w$-prefix is used to indicate that weak reduction and equalities are at issue in the lambda calculus. The weak lambda calculus is of interest because evaluation in many programming languages does not happen under lambda abstracts,\footnote{See \cite{Sestoft2001-gf}, \cite{Biernacka2022-hr} for surveys and taxonomies of different evaluation strategies, many of which are weak. See \cite{Dal-Lago2008-pb}, \cite{Forster2020-sl} for studies of the space and time complexity of systems centered around weak reduction.} and because in philosophy and linguistics it is commonly thought that lambda abstraction over states produces intensional contexts that are not necessarily closed under entailment relations.\footnote{\cite[\S{6.4}]{Dowty1981-gt}, \cite[pp. 73ff]{Gamut1991-gc}, \cite[\S{2.4}]{Chierchia2000-zn}.} Using an intensional version of combinatory logic, we show in \S\ref{sec:pureclcompleteconserve} that:
\begin{restatable}[Conservation of weak $\beta$-equalities in $\mltpure$ over $\beta\eta$-equalities in $\mltup$]{thm}{thmconservemltpureovermltup}\label{thm:conservemltpureovermltup}
Suppose that $\typd{M,N}{A}$ are terms of $\mltup$ in a signature with no constants of state type.

If $\mltpure\vdash_{w\beta} M=N$ then $\mltup\vdash_{\beta\eta} M=N$.
\end{restatable}
\noindent Note that this result concerns weak $\beta$-equality in $\mltpure$, but ordinary $\beta\eta$-equality in $\mltup$. It is not obvious whether this result continues to hold when one adopts just $\beta$ in $\mltup$ and whether it continues to hold when one adopts weak $\beta$-equality on both sides; likewise it is not obvious whether it holds when constants of state type are included in the signature.\footnote{The obstacles are in Proposition~\ref{prop:verydiff}(\ref{prop:verydiff:2}) for constants of state type, and in Proposition~\ref{prop:littlehelper} for weakness and $\eta$.}

As this discussion suggests, combinatory logic is important for the study of the metatheory of the simply-typed modal lambda calculus. The basic idea of combinatory logic is that it axiomatizes core lambda terms and their behavior under beta reduction.\footnote{Combinatory logic was initially developed by Curry \cite{Curry1958-yi}, \cite{Curry1972-tn}. For a modern treatment, see \cite{Bimbo2011-me}. For a side-by-side development of combinatory logic and lambda calculus, see \cite[Chapters 1-9]{Hindley2008-vw}.} Given its close relation to the lambda calculus, today combinatory logic is most used in theoretical computer science, and is for instance the paradigmatic example of a term rewriting system.\footnote{\cite[p. 64]{Terese2003-wa}.} I hope this use of combinatory logic helps to make more of a case for its deployment in philosophy and linguistics.\footnote{Combinatory logic has one prominent advocate in linguistics, namely Steedman \cite{Steedman1996-us}, \cite{Baldridge2011-zz}, \cite{Steedman2018-pc}. According to Partee's history in \cite[\S{3}]{Partee2012-on}, it was also used in early still-unpublished papers of Terence Parsons (of which I have been unable to find copies). Finally, combinatory logic has a rich history in philosophy, and was the original site of Curry's formalism \cite{Curry1951-pq} and his paradox \cite{Curry1942-vm}, \cite[p. 832]{Seldin2009-kc}, \cite{sep-curry-paradox}.} Combinatory logic is centrally used in the proof of Theorem~\ref{corcreatemodel} and Theorem~\ref{thm:conservemltpureovermltup}. The latter goes through Church-Rosser for combinatory intensional logic (\S\ref{subsec:clcr}), which we establish using Takahashi's method of complete developments.\footnote{\cite{Takahashi1995-kp}.} By contrast, minor extensions of Friedman-Warren's examples show that $\mltup$ does not have the Church-Rosser property for $\mltup$ non-maximal.\footnote{\cite{Friedman1980-pp}. Cf. Example~\ref{ex:friedmanwarren} in \S\ref{subsec:cr}.}

Outside of Zimmermann, the prior work which this paper is most indebted to is that of Gallin's book, which was his dissertation started under Montague and finished under Scott.\footnote{\cite{Gallin1975-tp}.} Chapters 1-2 of Gallin's book were concerned with Montague's original theory, in the setting of Church's simple theory of types. This present paper can then be thought of as the pursuit of Gallin's project where Church's simple theory of types is replaced with the simply-typed lambda calculus and where more options are given for the number of state types and the variables allotted to them. 

As this dependence on Montague makes clear, a primary reason to be interested in $\mltup$ is that one takes its finite number of state variables, and its restricted means of referring to states, to be accurate of some parts of the languages which we use.\footnote{I.e. Bach and Partee write: ``Suppose you had a language that was just like a predicate calculus except that the number of distinct variables was fixed at some small number. Suppose further that most of the time the elements acting as binders of the variables (quantifiers, abstraction operators, and the like) were either invisible in the language or not syntactically distinct from the singular terms and variables themselves. What you would have would be something like English'' (\cite{Bach1980-tj}, \cite[122]{Bach2008-aq}).} From this perspective, a chief import of this paper is that semantically one can freely use the more familiar system $\lt$ in studying the ostensibly more limited system $\mltup$. Further, this is compatible with adopting a more instrumentalist attitude towards $\lt$, and this attitude would be comfortable for one who thought that modal discourse was primarily wrought with statements of possibility and contingency of $\mltup$ rather than with more elaborate modes of references to states available in~$\lt$. Finally, from the Montagovian perspective, the system $\mltup$ is interesting because it makes available the interpretation of terms of type $A\rightarrow B$ as procedures which take one from items of type $A$ to items of type $B$, in keeping with the traditional interpretation of the simply-typed lambda calculus as a way of representing algorithms which meet specifications.\footnote{This thus gives a new perspective on the old idea of intensions as procedures, due to Dummett \cite[93, cf. pp. 96, 102, 179~ff]{Dummett1981aa} and Tich\'y \cite{Tichy1969-kv}. But there may be more to procedures that one would want to express than can be expressed in $\mltup$. See in particular the discussion of comonads at the close of \S\ref{sec:automorphisms}.}

Of course, formal systems do not wear their interpretations on their sleeves. While $\mltup$ was primarily designed to be a system which integrates modal logic with the simply-typed lambda calculus, it can also be used to describe a system which integrates finite-variable logics with the simply-typed lambda calculus. Finite-variable logics have been extensively studied in the context of first-order predicate logic due to their coupling of modest expressive power with desirable computational properties, at least when compared with full first-order predicate logic.\footnote{\cite{Pratt-Hartmann2003-tx}, \cite{Pratt-Hartmann2023-mn}.} On the interpretation of $\mltup$ on which it is a type-theoretic generalization of finite-variable logics rather than of modal logics, the terms of state types of $\mltup$ refer not to worlds but rather to individuals.\footnote{See Example~\ref{eqn:higherorderlogicatomics} and Example~\ref{ex:finvar}.} However, in contrast to the setting of first-order predicate logic, the main theorems of this paper point to an expressive parity between the finite-variable systems and their infinite-variable counterparts, as far as their usual semantics goes.

As a last word of introduction, let me say something to the reader coming from modal logic. One may be disappointed in the following pages by the lack of accessibility relations, the lack of options for Barcan, the lack of bisimulations, and the apparent lack of necessity operators. As for accessibility relations $R$, they can be added since one can type them as $\typd{R}{A\rightarrow A\rightarrow B}$, where $A$ is state type and $B$ is a type for truth-values. The Montagovian tradition put them in the metatheory,\footnote{E.g. \cite[p. 158]{Dowty1981-gt}.} since they restricted to a single variable for each state type and hence $Ruv$ would be ill-formed; but part of the goal in this paper is to extend Montague's system to settings with more variables for the state types. As for Barcan, this is just for lack of space, and another sequel to this paper focuses on simply-typed variable-domain modal lambda calculus. As regards bisimulations, these are a device for showing inexpressibility, and the resolution of Zimmermann's question (including Zimmermann's own resolution in the standard semantics) shows that simply-typed modal lambda calculus is maximally expressive.\footnote{Cf. \cite[fn 40 p. 34]{Zimmermann2020-an}.} Finally, regarding necessity operators: these are given by the assertion that ``$\lbd{u}{A}{P}$ is the function which always outputs true,'' where $A$ is a state type and $\typd{P}{B}$ has a type $B$ reserved for truth-values.\footnote{Cf. \cite[p. 16]{Gallin1975-tp}.} In general, the lack of bound state variables in $\Box P$ is simulated in $\mltup$ by making the set of state variables very small; this simulation is not new and is of course just the idea behind the standard translation of modal propositional logic into non-modal first-order predicate logic.

This paper is organized as follows. In \S\ref{sec:systems} the system $\mltup$ is formally defined, including its types~(\S\ref{subsec:typesandterms}), terms~(\S\ref{subsec:terms}), its relation to $\lt$~(\S\ref{sec:nonmodal}), the conventions on reduction notions~(\S\ref{subsec:conventionsredu}), its alpha conversion notion defined in terms of permutations~(\S\ref{sec:alpha}), its distinctive distanced beta version~(\S\ref{subsec:beta}), and it finally it is noted that $\mltup$ has strong normalization but does not have Church-Rosser~(\S\ref{subsec:cr}). In \S\ref{sec:soundness}, the Henkin semantics is described~(\S\ref{sec:semantics}), it is shown that the Soundness Theorem~\ref{thm:soundness} holds for these semantics and $\beta\eta$-equality~(\S\ref{subsec:soundness}), it is shown that the corresponding Completeness Theorem~\ref{thmcompletelambda} holds for $\mltpure$ since it has an open term model but it is noted in Example~\ref{exa:failureopenterm1} that the minimal system of $\mltup$ does not have an open term model~(\S\ref{sec:completeness-open-term}), and the section closes by showing how to use automorphisms to establish some fundamental inexpressibility results for $\mltup$~(\S\ref{sec:automorphisms}). In \S\ref{sec:combo} combinatory terms in $\mltup$ are introduced~(\S\ref{sec:typedcomboterms}), the expanded $\mathsf{BCDKW}$-combinatorial terms of $\mltup$ are defined and studied~(\S\ref{sec:combo:bcdkw}), and Theorem~\ref{corcreatemodel} is proven~(\S\ref{sec:combinatorial-models}), from which Theorems~\ref{thmconservation}, \ref{thm:express} are established~(\S\ref{sec:conservation}), and finally Theorem~\ref{thmcompletemltup} is proven~(\S\ref{sec:infty}).  In \S\ref{sec:ty2}, Theorems~\ref{thmconservationlt}, \ref{thm:expresslt} are proven~(\S\ref{sec:conserlt}), and open terms models for non-minimal systems equipped with the $\infty$-rule are shown to exist ~(\S\ref{sec:opentermmodelsredeux}). In \S\ref{sec:purecombinatory}  a pure intensional version of combinatory logic is developed and it is used to prove Theorem~\ref{thm:conservemltpureovermltup}. 

Finally, before beginning, we note for ease of reference that the numbered open questions are Question~\ref{qu:betazerobetarelation}, Question~\ref{qu:whatabouts}, Question~\ref{q:dedcompcon}, and Question~\ref{qu:howaboutpotentialcounter}.

%% file: 02-systems/00-intro.tex
\section{The systems of simply-typed modal lambda calculi}\label{sec:systems}

The systems of simply-typed modal lambda calculi are formed by altering the usual simply-typed lambda calculus in two ways: by a restriction on type formation and by a restriction on the number of variables. We begin with types.

%% file: 02-systems/01-types.tex
\subsection{Types}\label{subsec:typesandterms}

\begin{defi}[Types]\label{defn:type}
The \emph{atomic types} are made up of two disjoint sets, the \emph{state types}, which may be empty, and the \emph{basic entity types}, which must be non-empty.

The \emph{regular types} are defined as follows:
\begin{enumerate}[leftmargin=*]
    \item \label{type:definition:1} Each basic entity type $A$ is a regular type.
    \item \label{type:definition:2}  If $A$ is a regular type and $B$ is a regular type, then $(A\rightarrow B)$ is a regular type.
    \item \label{type:definition:3}  If $A$ is a state type and $B$ is a regular type, then $(A\rightarrow B)$ is regular type. 
\end{enumerate}
A \emph{type} is a state type or a regular type.  
\end{defi}

We associate arrows to the right in the usual way, so that $A\rightarrow B\rightarrow C$ is $A\rightarrow (B\rightarrow C)$. And we drop outermost parentheses.

Any choice of state types and basic entity types gives a choice of atomic types. A traditional choice is the following:

\begin{exa}[Montagovian atomics]\label{eqn:montagovianatomics}
Montague took as his atomics a single state type $S$ for worlds, and two basic entity types $E,T$, where $E$ is for individuals and $T$ is for truth-values. In the Montagovian tradition, $S\rightarrow T$ is the type of propositions, and $S\rightarrow E$ is the type of intensions of individuals.\footnote{In semantics in linguistics, one writes the type $S\rightarrow T$ as $st$, and the type $S\rightarrow E$ as $se$.} But $T\rightarrow S$ and $E\rightarrow S$ are not types since $S$ is a state type. Montague worked in Church's simple theory of types, where the truth-values in $T$ are made to be just $0$ and $1$. In weaker systems, one might also take $T$ to be a non-atomic type $A\rightarrow A\rightarrow A$ where $A$ is a regular type, which is a common way of representing Booleans (e.g. \cite[p. 39]{Barendregt2013-eb}). In temporal extensions of Montague's system, one would add another state type for times. In two-dimensional extensions of Montague's system, one would add another state type for epistemically possible worlds.
\end{exa}

For another example, consider:
\begin{exa}[Higher-order logic]\label{eqn:higherorderlogicatomics}
In contrast to the previous example, suppose that the atomic types are the type $E$, reserved for individuals, but which is now taken to be of \emph{state type}, and the type $T$, reserved for truth-values, which is still taken to be a basic entity type. Then the types are in one-one correspondence with the types of relational type theory, which is also known as higher-order logic.\footnote{\cite[73]{Orey1959-zk}, \cite[68]{Gallin1975-tp}. For a notional variant of Definition~\ref{defn:type} specialized to Example~\ref{eqn:higherorderlogicatomics}, consider: ``Our type system will be [\ldots] defined to be the smallest which includes the letter `$E$' (the `type of individuals') and `$T$' (the `type of propositions'), and is such that whenever $A$ and $B$ are in it and $B$ is distinct from $E$, [then] $A\rightarrow B$ (the type of operations that makes type-$B$ things out of type-$A$ things) is in it'' (\cite[112]{Fritz2024-wo}; choice of variables changed to match ours). The type $T$ is described as propositions rather than truth-values because higher-order metaphysics often models it as a Boolean algebra; this is so that the modal operators can then be viewed as operators on this algebra.}
For example, the type of binary relations between individuals is given by the type $E\rightarrow E\rightarrow T$. And the type of binary relations between unary relations and individuals is given by $(E\rightarrow T)\rightarrow E\rightarrow T$.\footnote{That is, we implicitly ``undo'' the currying and think of the type $(E\rightarrow T)\rightarrow E\rightarrow T$ as the type $((E\rightarrow T)\times E)\rightarrow T$.} 
\end{exa}

\begin{proofdetail}

Regarding the previous example, here is a proof that the types from that example are in one-one correspondence with the types of relational type theory or higher-order logic.

Formally, recall that the types of higher-order logic are given by the following:
\begin{enumerate}[label=(\alph*), ref = \alph*, align=left, leftmargin=*]
    \item\label{eqn:HOL1} $E$ is a type of higher-order logic.
    \item\label{eqn:HOL2} If $n\geq 0$ and $A_0, \ldots, A_{n-1}$ are types of higher-order logic, then $(A_0, \ldots, A_{n-1})$ is a type of higher-order logic.
\end{enumerate}
In (\ref{eqn:HOL2}), when $n=0$ there is understood to be exactly one length zero string, namely $\emptyset$, which is hence a type of higher-order logic.

One then defines a map $A\mapsto A^{\dag}$ from the types of higher-order logic into the types formed from state type $E$ and basic entity type $T$ according to Definition~\ref{defn:type} as follows:
\begin{enumerate}[label=(\roman*), ref = \roman*, align=left, leftmargin=*]
    \item\label{eqn:dag1} One sets $E^{\dag}$ equal to $E$
    \item\label{eqn:dag2} One sets $\emptyset^{\dag}$ equal to $T$.
    \item\label{eqn:dag3} If $n\geq 1$ and $A_0, \ldots, A_{n-1}$ are types of higher-order logic, then one sets $(A_0, \ldots, A_{n-1})^{\dag}$ equal to $A_0^{\dag}\rightarrow \cdots \rightarrow A_{n-1}^{\dag}\rightarrow T$
\end{enumerate}

This map is clearly injective, and we claim it is surjective. By (\ref{eqn:dag1}), it suffices to show that every regular type in the sense of Definition~\ref{defn:type} is in its range. By (\ref{eqn:dag2}), we have that $T$ is in its range. Suppose that $A\rightarrow B$ is a regular type, so that $B$ is a regular type, and by induction hypothesis $B$ is in the range, say $B$ is $D^{\dag}$ for some type $D$ in the sense of (\ref{eqn:HOL1})-(\ref{eqn:HOL2}). If (\ref{eqn:HOL1}) were used and $D$ was $E$, then by (\ref{eqn:dag1}) we would have that $B$ is $E$, and then $B$ would not be regular. Hence (\ref{eqn:HOL2}) must be used. Then $D$ would be $(A_1, \ldots, A_n)$, and then by (\ref{eqn:dag3}) one would have that $B$ would be $A_0^{\dag}\rightarrow \cdots \rightarrow A_{n-1}^{\dag}\rightarrow T$. There are then two cases to consider.
\begin{itemize}[leftmargin=*]
    \item If $A$ is a state type, then $A$ is $E$, and hence by (\ref{eqn:dag1}) one has that $A\rightarrow B$ is $E^{\dag}\rightarrow A_0^{\dag}\rightarrow \cdots \rightarrow A_{n-1}^{\dag}\rightarrow T$, which is equal by (\ref{eqn:dag3}) to $(E, A_0, \ldots, A_n)^{\dag}$.
    \item If $A$ is a regular type, then by induction hypothesis $A$ is $C^{\dag}$ for some type $C$ in the sense of (\ref{eqn:HOL1})-(\ref{eqn:HOL2}). Then $A\rightarrow B$ is $C^{\dag}\rightarrow A_0^{\dag}\rightarrow \cdots \rightarrow A_{n-1}^{\dag}\rightarrow T$, which is equal by (\ref{eqn:dag3}) to $(C, A_0, \ldots, A_n)^{\dag}$. 
\end{itemize}

\end{proofdetail}

%% file: 02-systems/02-terms.tex
\subsection{Terms}\label{subsec:terms}

The restrictions on variables are enforced by the following:

\begin{defi}[Parameter, which controls variables of state types]\label{defn:parameter}
The parameter $\param$ is a function which sends each state type to an element of $\{1,2,\ldots, \omega\}$ (where, recall $\omega$ is the least infinite cardinal).

We extend to $\param$ to all types by setting $\param(A)=\omega$ for all regular types~$A$.

We define \emph{the set of variables of type $A$} to be $\{\typd{v_i}{A} : 0\leq i<\param({A})\}$.
\end{defi}
Of course, we quickly move to writing $u,v,w,\ldots$ etc. instead of the more formal $v_0, v_1, v_2, \ldots$ etc. But when doing so we must be careful not to exceed the number set by the parameter.

There is a natural partial order on parameters given by $\param \leq \param^{\prime}$ iff for all state types $A$ one has $\param(A)\leq \param^{\prime}(A)$. Hence, if $\param \leq \param^{\prime}$, then all of the variables of $\param$ are also variables of $\param^{\prime}$.

Relative to these restrictions on types and variables, we define the terms in the usual way, where we assume that a collection of typed constants, called a \emph{signature}, has been specified in advance:
\begin{defi}[Terms]\label{defn:term}
Let $\param$ be a parameter and let $\mathcal{D}$ be a signature. Then the terms $\typd{M}{A}$ of $\mltup$ is signature~$\mathcal{D}$ are defined as follows:
\begin{enumerate}[leftmargin=*]
    \item \label{defn:term:1} \emph{Variables}: the variables $\typd{v_i}{A}$ for $i<\param(A)$ are terms of $\mltup$.
    \item \label{defn:term:2} \emph{Constants}: the constants $\typd{c}{A}$ from $\mathcal{D}$ are terms of $\mltup$.
    \item \label{defn:term:3} \emph{Applications}: if $\typd{M}{A\rightarrow B}$ and $\typd{N}{A}$ are terms of $\mltup$ then the application $\typd{(MN)}{B}$ is a term of $\mltup$.
    \item \label{defn:term:4} \emph{Lambda abstracts}: if $C$ is a regular type and $\typd{L}{C}$ is a term of $\mltup$ and further $0\leq i<\param(A)$, then the lambda abstract $\typd{(\lbd{v_i}{A}{L})}{\;A\rightarrow C}$ is a term of $\mltup$.
\end{enumerate}
\end{defi}

Formally, the signature $\mathcal{D}$ ought to be displayed in the definition of $\mltup$ since terms of $\mltup$ depend on the signature $\mathcal{D}$. But in this paper we are not switching often between different signatures, and so we omit it.

We write $\mlt{\kappa}$ for $\mltup$ where $\param(A)=\kappa$ for all state types $A$. Montague and Zimmermann studied $\mlt{1}$.\footnote{The system $\mlt{1}$ is Zimmermann's IL$^{\ast}$ from \cite[p. 67]{Zimmermann1989-gs} when the constants are restricted to be of type $A\rightarrow B$, where $A$ is a state type. See \cite[\S{4.1} p. 75]{Zimmermann1989-gs}, \cite[pp. 17-19]{Zimmermann2020-an} for discussion of the relation of this to Montague's choice of object-language. As Zimmermann says (\cite[p. 67]{Zimmermann1989-gs}, \cite[p. 19]{Zimmermann2020-an}), the idea is to find in type theory the image of Montague's system under the standard translation (he calls it Gallin's translation after \cite[pp. 61 ff]{Gallin1975-tp}).} The maximal system is $\mltpure$, which has countably many variables for each state type.

For application, we associate to the left in the usual way, so that $PQR$ is $(PQ)R$. And we drop outermost parentheses.

For nested lambda abstraction, we use vector notation and abbreviate the term $\lbd{v_1}{A_1}{\cdots\lbd{v_n}{A_n}{M}}$ by $\lbd{\vec{v}}{\vec{A}}{M}$; and we refer to $n$ as the \emph{length} of $\vec{v}$. Further, if $\typd{M}{C}$ then we abbreviate the type $A_1\rightarrow \cdots \rightarrow A_n\rightarrow C$ of the nested lambda abstract $\lbd{\vec{v}}{\vec{A}}{M}$ as $\vec{A}\rightarrow C$.\footnote{Occasionally, when $A_1, \ldots, A_n$ are all identical, we write $\lbd{\vec{v}}{A}{M}$ instead of $\lbd{\vec{v}}{\vec{A}}{M}$; and we refer to the type of this term as $A^n\rightarrow C$ where $\typd{M}{C}$. Since we use this notation sparingly, we remind the reader when we use it.\label{fn:vecsame}} We similarly use vector notation in writing $\vec{N}$ for $N_1 \cdots N_n$ when these are of the appropriate type; and we similarly write $\big(\lbd{\vec{v}}{\vec{A}}{M}\big) \vec{N}$ for $\big(\lbd{v_1}{A_1}{\cdots\lbd{v_n}{A_n}{M}}\big) N_1 \cdots N_n$ when $\typd{N_1}{A_1}, \ldots, \typd{N_n}{A_n}$.

Officially, everything is Church-typed rather than Curry-typed, but we minimize the display of types on the terms to maintain readability. However, we always display the type on the bound variables since the restrictions on the number of variables is so central to the system.

Here is a simple but important proposition:

\begin{prop}[Terms of state type]\label{prop:state-free}
The only terms of $\mltup$ of state type are the variables and constants.
\end{prop}
\begin{proof}
Suppose $B$ is a state type. A term of type $B$ cannot be an application $MN$ since then we would have $\typd{M}{A\rightarrow B}$ and $\typd{N}{A}$, but $A\rightarrow B$ is not a type since $B$ is a state type. Also, a term of type $B$ cannot be a lambda abstract since lambda abstracts always have functional type and $B$ is a state type and so atomic. Hence, the only remaining options for terms are constants and variables.
\end{proof}
\noindent This proposition would fail if one tried to restrict the variables without introducing restrictions on types, since e.g. one could introduce terms of state type by mapping from a functional type into the state type.

To illustrate the usefulness of constants, consider:
\begin{exa}[Actuality operators: named worlds vs. diagonals]\label{exa:namedworlds}
The simplest version of the actuality operator is $\lbd{p}{A\rightarrow B}{pc}$, which has type $(A\rightarrow B)\rightarrow B$, where $A$ is a state type, $B$ is a type reserved for truth-values, and $\typd{c}{A}$ is a constant. If $c$ is the actual world, then the actuality operator just takes a proposition $p$ and evaluates it at the actual world.

There is also a distinct actuality operator $\lbd{p}{A\rightarrow A\rightarrow B}{\lbd{v}{A}{pvv}}$. If $B$ is again a type reserved for truth-values, then this term intuitively takes the proposition $p$, which takes two state arguments, and when given a single state argument $v$ returns the diagonal $pvv$. As a lambda term, this actuality operator is the Warbler of combinatory logic (Definition~\ref{defn:typedcombo}, cf. Definition~\ref{defn:typedcombo2} for pure combinatory logic). In modal logic itself, this actuality operator is widespread in two-dimensional logics (cf. \cite{Davies1980-ym}, \cite{Nimtz2017-qu}).
\end{exa}

%% file: 02-systems/03-nonmodal.tex
\subsection{Non-modal simply-typed lambda calculus}\label{sec:nonmodal}

We introduce some notation, prefigured in \S\ref{sec:intro}, for the ordinary non-modal simply-typed lambda calculus:
\begin{defi}[Ordinary simply-typed lambda calculus]
The ordinary simply-typed lambda calculus $\lt$ is the simply-typed modal lambda calculus $\mltpure$ without any state types.
\end{defi}
If there are no state types, then the definition of type and term (Definitions~\ref{defn:type}, \ref{defn:term}) just results in the ordinary simply-typed lambda calculus (\cite[Part I]{Barendregt2013-eb}, \cite[Chapters 10, 12]{Hindley2008-vw}). As usual, there are as many systems of $\lt$ as there are choices of atomic basic entity types. In the proofs of Theorems~\ref{thmconservationlt}-\ref{thm:expresslt} in \S\ref{sec:ty2}, we adopt the following convention:
\begin{defi}[Convention on a pair of $\lt$ and $\mltpure$]\label{defn:nonmodal}
When discussing the relation between a specific pair $\lt$ and $\mltpure$ (e.g. as regards conservativity or expressibility), we assume that the basic entity types of $\lt$ are the union of the state types and basic entity types of $\mltpure$.
\end{defi}
We do not need this convention until \S\ref{sec:ty2}, and remind the reader of it there. But we illustrate with two examples.

\begin{exa}[Montagovian atomics revisited]\label{ex:montre}
Recall from Example~\ref{eqn:montagovianatomics} that Montague's atomics were state type $S$, for worlds, and basic entity types $E,T$, for individuals and truth-values, respectively. While Montague himself worked in $\mlt{1}$, we could consider working in $\mltpure$ as well.

Following the convention in Definition~\ref{defn:nonmodal}, the associated $\lt$ has no state types but has basic entity types $S,E,T$. It is just the ordinary simply-typed lambda calculus with the three atomic types $S,E,T$. 

For instance, $\lt$ has types $T\rightarrow S$ and $E\rightarrow S$, but $\mltpure$ does not.

Further, $\lt$ has terms of type $\mltpure$ which are not terms of $\mltpure$, such as $Uv$, where $\typd{U}{E\rightarrow S}$ and $\typd{v}{E}$ are variables. 

\end{exa}

\begin{exa}[Higher-order logic revisited]\label{ex:finvar}
Recall from Example~\ref{eqn:higherorderlogicatomics} that higher-order logic (also known as relational type theory) has state type $E$ for individuals and basic entity type $T$ for truth-values (or propositions).

Relative to this choice of atomics, the systems $\mltup$ with $\param(E)<\omega$ are higher-order versions of the extensively studied finite-variable first-order logics.\footnote{\cite{Pratt-Hartmann2003-tx}, \cite{Pratt-Hartmann2023-mn}.} But in $\mltup$ the finite variable restrictions are imposed only on the state types, which in this example are reserved for the individuals. And $\mltpure$ is just a notational variant of higher-order logic.

Following the convention in Definition~\ref{defn:nonmodal}, the associated $\lt$ has no state types but has basic entity types $E,T$. It is just the ordinary simply-typed lambda calculus with the two atomic types $E,T$. That is, in this example, $\lt$ is the functional type theory associated to a traditional relational type theory $\mltpure$.

For instance, $\lt$ has types $T\rightarrow E$ and $E\rightarrow E$, but $\mltpure$ does not.

Further, $\lt$ has terms of type $\mltpure$ which are not terms of $\mltpure$, such as $Uv$, where $\typd{U}{E\rightarrow E}$ and $\typd{v}{E}$ are variables. 

\end{exa}

%% file: 02-systems/04-conventionsreduction.tex
\subsection{Conventions on reduction notions}\label{subsec:conventionsredu}

For a family of binary relations $R_A$ of terms of $\mltup$ of type $A$, we define $\reduces{R_A}{\param}$ to be its compatible closure, i.e. the smallest binary relation on terms of type $A$ which includes $R_A$ which is \emph{compatible}: 
\begin{enumerate}[label=(\roman*), ref=\roman*]
    \item\label{eqn:compatible:1} It is closed under well-formed application on both sides.
    \item\label{eqn:compatible:2} It is closed under lambda abstraction.
\end{enumerate}

Regarding (\ref{eqn:compatible:1}), this means: if $P\reduces{R_A}{\param}Q$ then $MP\reduces{R_B}{\param}MQ$ for all terms $\typd{M}{A\rightarrow B}$ of $\mltup$; and likewise if $M\reduces{R_{A\rightarrow B}}{\param}N$ then $MP\reduces{R_B}{\param}NP$ for all terms $\typd{P}{A}$ of $\mltup$. Note that ``closed under well-formed application on both sides'' does \emph{not} mean: if $P\reduces{R_A}{\param}Q$ and $M\reduces{R_{A\rightarrow B}}{\param}N$, then $MP\reduces{R_{A\rightarrow B}}{\param}NQ$. This would be a parallel reduction notion (cf. \S\ref{subsec:clcr}), whereas the intended idea is a single $R$-reduction happening somewhere inside the term.

Regarding (\ref{eqn:compatible:2}), this means: if $P\reduces{R_A}{\param}Q$ then $\lbd{x}{A}{P} \reduces{R_A}{\param} \lbd{x}{A}{Q}$ for every variable $\typd{x}{A}$ of $\mltup$.\footnote{In the case of equalities, this closure condition is traditionally known as the $\xi$-rule; in the case of reductions, it usually is not given a specific name (cf. \cite[pp. 23, 50]{Barendregt1981-ai}, \cite[pp. 13, 30]{Hankin2004-lz}).\label{fn:xirule}}

We define the family $\transreduces{R_A}{\param}$ as the reflexive transitive closure of the family $\reduces{R_A}{\param}$, and we define the family $\equals{R_A}{\param}$ as the smallest equivalence relation containing $\reduces{R_A}{\param}$. In all this, we are just following the standard treatment in \cite[p. 50]{Barendregt1981-ai}, adapted to the typed context. By the same argument as \cite[p. 52 Lemma 3.1.6]{Barendregt1981-ai} one has that both the family  $\transreduces{R_A}{\param}$ and the family $\equals{R_A}{\param}$ are compatible.

Similarly, for two families of binary relations $R_A,S_A$ of terms of $\mltup$ of type $A$, we define~$\reduces{R_A S_A}{\param}$ to be  $\reduces{R_A\cup S_A}{\param}$. And likewise for three families etc.

When the type $A$ is clear from context, we just drop it from the subscripts of the definitions in the previous paragraph; since it is almost always clear from context, we almost always drop it. Further, in what follows, to slightly compress discussion, we often use $\reduce{R}$ to introduce a binary relation $R$ directly, allowing ourselves to skip the extra step of first declaring $R$ and then its compatible closure.

If $R_A$ is any family of binary relations on terms of $\mltup$ of type $A$, then for terms $\typd{M,N}{A}$ of $\mltup$, we define $\mltup\vdash_R M=N$ iff $M\equals{R_A}{\param}N$. Note that there is no identity in the object language of $\mltup$, and rather this is a meta-theoretically defined notion.

Finally, systems which keep (\ref{eqn:compatible:1}) but which drop (\ref{eqn:compatible:2}) are known as weak systems, and we designate the associated reducibility relation with a `$w$' prefix, and we write the associated reduction relation as $\reduces{wR}{\param}$. We do not use weak systems until \S\ref{sec:purecombinatory}, and we remind the reader of the convention when we come to it.

%% file: 02-systems/05-alpha.tex
\subsection{Alpha conversion}\label{sec:alpha}

It is standard in lambda calculus to identify $\alpha$-equivalent terms, that is, terms which are the same up to renaming of bound variables, and to view oneself as formally working with equivalence classes of $\alpha$-equivalent terms (\cite[p. 26, pp. 577 ff]{Barendregt1981-ai}, \cite[p. 277 ff]{Hindley2008-vw}). We proceed similarly with $\mltup$.

But different choices of $\param$ result in different equivalence classes. For, if $\param\leq \param^{\prime}$, then the $\alpha$-equivalence classes of $\mltup$ are finer than the $\alpha$-equivalence classes of $\mlt{\param^{\prime}}$. Here is an example where the equivalence classes are maximally fine:
\begin{exa}[Example of fineness of $\alpha$-equivalence classes]
In $\mlt{1}$, if $A$ is a state type and $B$ is a basic entity type and $\typd{v}{A}$ and $\typd{u}{B}$ are variables, then $\lbd{v}{A}{u}$ is the only term in its $\alpha$-equivalence class. This is because formally $\typd{v}{A}$ is $\typd{v_0}{A}$ and there are no other variables of state type $A$ in $\mlt{1}$.
\end{exa}

In many treatments of the lambda calculus, it is useful to present a step-by-step reduction notion which slowly rewrites a term into an $\alpha$-equivalent (\cite[p. 26]{Barendregt1981-ai}, \cite[p. 278]{Hindley2008-vw}). But in $\mlt{2}$, implementing this usual procedure would require moving to $\mlt{n}$ for $n>2$, since if one is in $\mlt{2}$ and working with state type $A$, if one needs to change $\lbd{v}{A}{\lbd{u}{A}{cuv}}$ into $\lbd{u}{A}{\lbd{v}{A}{cvu}}$, then one will have to appeal to other variables to do the transition inductively. Rather than set up this procedure, in the few places where we need a formal definition of $\alpha$-equivalence (cf. Proposition~\ref{prop:commpermtran}), we define $\alpha$-equivalence in terms of permutations. While not common, this approach has been independently developed by Gabbay and Pitts.\footnote{\cite{Gabbay2002-lw}, \cite{Pitts2006-ak}, \cite[Chapter 4]{Pitts2013-bb}.}

\begin{defi}[$\alpha$-equivalence]\label{defn:alphaeq}
Let parameter $\param$ be fixed.

Suppose that $\pi$ is a type-preserving permutation of the variables of $\mltup$. Then we extend to a type-preserving permutation from terms $M$ of $\mltup$ to terms $M^{\pi}$ of $\mltup$ by further setting $c^{\pi}$ to be $c$ for constants $c$; by setting $(M_0 M)^{\pi}$ to be $M_0^{\pi} M_1^{\pi}$; and by setting $(\lbd{v}{A}{M})^{\pi}$ to be $\lbd{u}{A}{M^{\pi}}$, where $\pi(\typd{v}{A})=\typd{u}{A}$.

We define $\alpha_{A}$ to be the binary relation on ordered pairs of terms $\mltup$ of type $A$ given by: the ordered pair $(M,N)$ stands in the $\alpha_A$ relation iff $N$ is $M^{\pi}$ for some type-preserving permutation $\pi$ of the variables of $\mltup$ which is the identity on the free variables of $M$. Then, we define $\equals{\alpha_A}{\param}$ to be the smallest equivalence relation containing the compatible closure~$\reduces{\alpha_A}{\param}$ of~$\alpha_A$ (cf. \S\ref{subsec:conventionsredu}).

Finally, we say that two terms $\typd{M,N}{A}$ in $\mltup$ are \emph{$\alpha$-equivalent} if $M\equals{\alpha_A}{\param} N$.

\end{defi}

Here are two brief illustrations:
\begin{exa}
In  $\mlt{2}$, suppose $A$ is a state type and $B$ is a basic entity type and $\typd{d}{A\rightarrow B}$ is a constant and $\typd{c}{(A\rightarrow B)\rightarrow (A\rightarrow B) \rightarrow B}$ is a constant and $\typd{x,y}{A}$ are distinct variables. Using the permutation which transposes $x,y$ we obtain that $\lbd{y}{A}{dy}$ is $\alpha$-equivalent to $\lbd{x}{A}{dx}$. Then the closure of $\alpha$-equivalence under application gives that $c(\lbd{x}{A}{dx})(\lbd{y}{A}{dy})$ is $\alpha$-equivalent to $c(\lbd{x}{A}{dx})(\lbd{x}{A}{dx})$.
\end{exa}

\begin{exa}
In  $\mlt{2}$, suppose $A$ is a state type and $B$ is a basic entity type and $\typd{d}{A\rightarrow A \rightarrow B}$ is a constant and $\typd{x,y}{A}$ are distinct variables. Then $\lbd{x}{A}{dxy}$ is the only formula in its $\alpha$-equivalence class. For if a permutation of $\typd{x,y}{A}$ is the identity on the free variables of $\lbd{x}{A}{dxy}$ or its subterm $dxy$, then it is the identity permutation.
\end{exa}

%% file: 02-systems/06-betaeta.tex
\subsection{Beta and eta reduction}\label{subsec:beta}

The following definition is the most important definition in the paper, and developing this specific generalization of $\beta$-reduction proved instrumental to proving the main results of this paper. But the definition of this generalization is admittedly a bit baroque at first glance. For more motivation, we would point the reader forward towards the subsequent Remark~\ref{rem:transposition} for a helpful heuristic. Likewise, we would point the reader towards \S\S\ref{sec:typedcomboterms}-\ref{sec:combinatorial-models} which culminates in the proof of Theorem~\ref{corcreatemodel}, where this generalization seems the almost inevitable way to extend Andrews' proof of the combinatorial characterization of models to $\mltup$.\footnote{While our notion came about in trying to extend Andrews' proof, it bears some resemblance to Accattoli's distant B-reduction \cite[\S{4.1} p. 13]{Accattoli2019-sf}. This similarity merits further investigation, but one can say the following: while both allow ``abstraction and the argument to interact at a distance'' (\cite[\S{4.1} p. 13]{Accattoli2019-sf}), our distanced $\beta$-reduction takes place in the simply-typed non-weak setting with no explicit substitutions, and Accattoli's takes place in the untyped weak setting with explicit substitutions. That said, in \S\ref{sec:purecombinatory}, we too move to the weak setting.}

\begin{defi}[Definition of $\beta$- and $\eta$-reduction in $\mltup$]\label{defn:betareduction}
We say $\big( \lbd{\vec{x}}{\vec{B}}{\lbd{v}{A}{L}}\big) \vec{M} N \reduces{\beta}{\param} \big( \lbd{\vec{x}}{\vec{B}}{L[v:=N]} \big) \vec{M}$ if each of the following conditions holds
\begin{enumerate}[leftmargin=*]
    \item\label{defn:betareduction:1} $\typd{N}{A}$ is free for $\typd{v}{A}$ in $\typd{L}{C}$;
    \item\label{defn:betareduction:3} the variables in $\typd{\vec{x}}{\vec{B}}$ are not free in $\typd{N}{A}$;
    \item \label{defn:betareduction:2} the variables in $\typd{\vec{x}}{\vec{B}}, \typd{v}{A}$ are pairwise distinct.
\end{enumerate}

The \emph{distance} of the $\beta$-reduction is the length of vector $\typd{\vec{x}}{\vec{B}}$.

We use $\beta_0$ for distance zero $\beta$-reduction.

We say that an instance of $\beta$-reduction is \emph{regular} if
\begin{enumerate}[leftmargin=*]\setcounter{enumi}{3}
    \item\label{defn:betareduction:5} the only variable, if any, of state type in the tuple $\typd{\vec{x}}{\vec{B}}\equiv \typd{x_1}{B_1}, \ldots, \typd{x_n}{B_n}$ is the first one $\typd{x_1}{B_1}$.
\end{enumerate}

We use $\beta_r$ for regular $\beta$-reduction.

Finally, we define $\lbd{x}{A}{Mx} \reduces{\eta}{\param} M$ when $\typd{x}{A}$ not free in $\typd{M}{A\rightarrow B}$.
\end{defi}

Now we turn to various notational remarks about $\beta$:

\begin{rem}[Notation related to $\beta$]
The notion of ``free for'' in (\ref{defn:betareduction:1}) is standard: a term $\typd{N}{A}$ is \emph{free for} $\typd{v}{A}$ in $L$ if all free occurrences of $\typd{v}{A}$ in $L$ do not occur in a subterm $P$ of a subterm $\lbd{u}{C}{P}$ of $L$ where $\typd{u}{C}$ is free in $N$. 

The vector $\vec{M}$ in $\beta$-reduction has type $\typd{\vec{M}}{\vec{B}}$, that is, the same type as that of the variables  $\typd{\vec{x}}{\vec{B}}$.

Note that we are using $L[x:=N]$ for substitution, instead of $[N/x]L$. We prefer $L[x:=N]$ since it matches nicely with the familiar notation for variable assignments (cf. \S\ref{sec:semantics}).

We adopt the convention of taking note of the distance in a $\beta$-reduction in the accompanying text whenever it is non-zero. This helps one to easily see where the more distinctive instances of our generalised $\beta$-reduction are being used in the proofs: in particular, one can just search the document for the word ``distance.'' Unless a result concerns what can be done with $\beta_0$ reduction alone, we tend not to mark $\beta_0$ explicitly but just refer to the instance as $\beta$; and similarly for $\beta_r$.
\end{rem}

We take note of the following elementary proposition, which follows directly from the definitions.
\begin{prop}
$P\reduces{\beta_0}{\param} Q$ implies  $P\reduces{\beta_r}{\param} Q$, which in turn implies  $P\reduces{\beta}{\param} Q$. 
\end{prop}

One can render distanced $\beta$-reduction more familiar by defining a notion of transposition reduction:\footnote{Thanks to an anonymous referee for emphasizing this.}
\begin{defi}[Transposition reduction $\tau$]\label{defn:tau}
Suppose that variables $\typd{\vec{x}}{\vec{B}},\typd{\vec{v}}{\vec{A}}$ are pairwise disjoint. Then one defines the \emph{transposition reduction} $\tau$ by $\big( \lbd{\vec{x}}{\vec{B}}{\lbd{\vec{v}}{\vec{A}}{L}}\big) \vec{M} \vec{N} \reduces{\tau}{\param}\big( \lbd{\vec{v}}{\vec{A}}{\lbd{\vec{x}}{\vec{B}}{L}}\big)  \vec{N} \vec{M}$.
\end{defi}

\begin{rem}[$\beta$-reduction as the composition of transposition and $\beta_0$-reduction]\label{rem:transposition}
 If the conditions (\ref{defn:betareduction:1})-(\ref{defn:betareduction:2}) of Definition~\ref{defn:betareduction} are satisfied, then one has that $\beta$ is the composition $\beta_0\circ \tau$, in the sense of the following commutative diagram:
\begin{equation}\label{eqn:taureduction}
\xymatrix{
\big( \lbd{\vec{x}}{\vec{B}}{\lbd{v}{A}{L}}\big) \vec{M} N \ar[r]_{\tau} \ar@{-->}@/^2pc/[rr]_{\beta} & \big( \lbd{v}{A}{\lbd{\vec{x}}{\vec{B}}{L}}\big)  N \vec{M} \ar[r]_{\beta_0} & \big( \lbd{\vec{x}}{\vec{B}}{L[v:=N]} \big) \vec{M}
}
\end{equation}

\noindent Expressed in words, the idea is: to do a $\beta$-reduction, first check that the conditions (\ref{defn:betareduction:1})-(\ref{defn:betareduction:2}) of Definition~\ref{defn:betareduction} are satisfied, and second transpose, and third do a $\beta_0$-reduction. 

Taking note of this composition is certainly a useful heuristic, which no doubt leaves one better positioned to work, at the outset, with $\beta$-reduction of non-zero distance than by trying to primitively memorize  conditions (\ref{defn:betareduction:1})-(\ref{defn:betareduction:2}) of Definition~\ref{defn:betareduction}. For instance, on the basis of the composition, one sees quickly that conditions (\ref{defn:betareduction:1})-(\ref{defn:betareduction:2}) of Definition~\ref{defn:betareduction} ensure that $\typd{N}{A}$ is free for $\typd{v}{A}$ in $\lbd{\vec{x}}{\vec{B}}{L}$, which permits the $\beta_0$-reduction in~(\ref{eqn:taureduction}).

While the composition is a useful conceptualization of $\reduces{\beta}{\param}$, it would not be appropriate to work in $\reduces{\beta_0 \tau}{\param}$ i.e. the compatible closure of the binary relation which is the union of $\beta_0$-reduction and $\tau$-reduction. For, $\tau$-reduction and hence $\beta_0\tau$-reduction is trivially not strongly normalizing since one has loops 
\begin{equation*}
(\lbd{x}{B}{\lbd{v}{A}{M}})PQ\reduces{\tau}{\param} (\lbd{v}{A}{\lbd{x}{B}{M}})QP \reduces{\tau}{\param} (\lbd{x}{B}{\lbd{v}{A}{M}})PQ\reduces{\tau}{\param}\cdots
\end{equation*}
By contrast, one can show that $\beta$-reduction is strongly normalizing in $\mltup$ (cf. \S\ref{subsec:cr}).
\end{rem}

Here is an example of a $\beta$-reduction of distance $1$:
\begin{exa}[\emph{De re} vs. \emph{de dicto} example of $\beta$-reduction]
Consider the proposition that the baliff thinks the defendant is a chef. The proposition has two traditional readings:

\noindent \emph{De re}: The proposition$_{u}$ that the baliff$_u$ thinks$_v$ that the defendant$_u$ is a chef$_v$.

\hspace{5mm} $\lbd{u}{S}{\big( bu\big(\lbd{v}{S}{cv(du)}\big)}$

\noindent \emph{De dicto}: The proposition$_u$ that the baliff$_u$ thinks$_v$ that the defendant$_v$ is a chef$_v$.

\hspace{5mm} $\lbd{u}{S}{\big( bu\big(\lbd{v}{S}{cv(dv)}\big)}$

\noindent These are terms of $\mlt{2}$ of type $S\rightarrow T$, where the types $S,E,T$ are as in Example~\ref{eqn:montagovianatomics}. Further, suppose that the constants have the types $\typd{b}{S\rightarrow (S\rightarrow T)\rightarrow T}, \typd{c}{S\rightarrow E\rightarrow T}, \typd{d}{S\rightarrow E}$.

In the Montagovian tradition, one obtains the \emph{de re} reading from the \emph{de dicto} reading by replacing defendent$_v$ (i.e. $dv$) with a fresh variable (say $\typd{x}{E}$), lambda abstracting over that variable, and then applying the resulting lambda abstract to the value of the defendent$_u$ (i.e. $du$).\footnote{\cite[pp. 206-207]{Dowty1981-gt}, \cite[p. 184]{Gamut1991-gc}. I am using the \emph{de re} vs. \emph{de dicto} example to illustrate distanced $\beta$-reduction. See \cite{Keshet2010-lz} for a recent discussion of the empirical adequacy of the Montagovian perspective on the \emph{de re} and \emph{de dicto}.} This results in  $\big( \lbd{x}{E}{\lbd{u}{S}{bu\big(\lbd{v}{S}{cvx} \big)}}\big) (du)$. While this term has the right type, namely $S\rightarrow T$, it has a free state variable $\typd{u}{S}$, whereas both the \emph{de re} and \emph{de dicto} reading are closed; further, it will not $\beta$-reduce to the \emph{de re} reading since $du$ is not free for $x$ in $\lbd{u}{S}{bu\big(\lbd{v}{S}{cvx} \big)}$. To get the requisite generality and closed term, one should additionally apply the state variable $\typd{u}{S}$ to get a term of type $T$ and then lambda abstract over $\typd{u}{S}$ one more time. If one does so, then one can derive the \emph{de re} reading, where the first step is a $\beta$-reduction of distance~1 happening under the $\lambda$-abstract $\lbd{u}{S}{}\ldots$:
\begin{align}
   \lbd{u}{S}{\bigg( \big( \lbd{x}{E}{\lbd{u}{S}{bu\big(\lbd{v}{S}{cvx} \big)}}\big) (du) u\bigg)}\reduces{\beta}{2}\;\; & \lbd{u}{S}{\bigg( \big( \lbd{x}{E}{bu\big(\lbd{v}{S}{cvx} \big)}\big) (du) \bigg)} \notag \\
 \reduces{\beta}{2}\;\;   &\lbd{u}{S}{\bigg( bu\big(\lbd{v}{S}{cv(du)} \big)  \bigg)} \notag 
\end{align}
Note that it is \emph{not} possible to use a $\beta$-reduction of distance zero on the first line since $du$ is not free for $\typd{x}{E}$ in $\lbd{u}{S}{bu\big(\lbd{v}{S}{cvx} \big)}$. (Finally, note again that the 2 superscript on $\reduces{\beta}{2}$ indicates that we are working in $\mlt{2}$).
\end{exa}

While this example shows how $\beta$-\emph{reductions} of distance $>0$ are different than $\beta$-reductions of distance zero, the following proposition shows that these differences are not present when we restrict to regular $\beta$-\emph{equality}:
\begin{prop}\label{prop:oneway}
Suppose that $\typd{M,N}{A}$ are terms of $\mltup$. 

If $M\transreduces{\beta_r}{\param} N$ then $M\equals{\beta_0}{\param} N$. Hence: $M\equals{\beta_r}{\param} N$ iff  $M\equals{\beta_0}{\param} N$.
\end{prop}
\begin{proof}
Suppose that $\big( \lbd{\vec{x}}{\vec{B}}{\lbd{v}{A}{L}}\big) \vec{M} N \reduces{\beta_r}{\param} \big( \lbd{\vec{x}}{\vec{B}}{L[v:=N]} \big) \vec{M}$. We must show $
\big( \lbd{\vec{x}}{\vec{B}}{\lbd{v}{A}{L}}\big) \vec{M} N \equals{\beta_0}{\param} \big( \lbd{\vec{x}}{B}{L[v:=N]} \big) \vec{M}$. Let $\typd{\vec{x}}{\vec{B}}\equiv \typd{x_1}{B_1}, \ldots, \typd{x_{\ell}}{B_{\ell}}$. By Definition~\ref{defn:betareduction}~(\ref{defn:betareduction:5}), we have that $B_i$ is regular for $1< i\leq \ell$. Hence by $\alpha$-conversion we may assume that $\typd{x_i}{B_i}$ for $1< i\leq \ell$ does not appear free in $\typd{\vec{M}}{\vec{B}}$. Then we have:
\begin{align}
 \big(\lbd{\vec{x}}{\vec{B}}{\lbd{v}{A}{L}}\big) \vec{M} N & \twoheadleftarrow_{\beta_0}^{\param} \bigg(\lbd{\vec{x}}{\vec{B}}{\bigg(\big(\lbd{\vec{x}}{\vec{B}}{\lbd{v}{A}{L}}\big) \vec{x} N }\bigg)\bigg) \vec{M}\label{eqn:betabetzerofact}  \\
 & \transreduces{\beta_0}{\param} \big(\lbd{\vec{x}}{\vec{B}}{\big(({\lbd{v}{A}{L}}) N }\big)\big) \vec{M}\notag   \reduces{\beta_0}{\param} \big( \lbd{\vec{x}}{B}{L[v:=N]} \big) \vec{M}\notag 
\end{align}

In the first step, we do as many $\beta_0$-reductions as the length $\ell$ of $\typd{\vec{x}}{\vec{B}}$. This has the stated effect because Definition~\ref{defn:betareduction}~(\ref{defn:betareduction:3}) says that $\typd{\vec{x}}{\vec{B}}$ is not free in $\typd{N}{A}$. Further, for each $1\leq i\leq \ell$ one has that $\typd{x_i}{B_i}$ does not occur free in $\typd{M_1}{B_1}, \ldots, \typd{M_{i-1}}{B_{i-1}}$. For $i=1$ this is vacuously true, and for $1<i\leq \ell$ this is by our previous $\alpha$-conversions. Hence, for each $1\leq i\leq \ell$, $\typd{M_i}{B_i}$ is free for $\typd{x_i}{B_i}$ in the term $\big(\lbd{x_{i+1}}{B_{i+1}}{\cdots \lbd{x_{\ell}}{B_{\ell}}{(\lbd{\vec{x}}{\vec{B}}{\lbd{v}{A}{L})}}}\big) M_1 \cdots M_{i-1} x_i \cdots x_{\ell} N$. 

Finally, regarding the first step in (\ref{eqn:betabetzerofact}), note that by
Definition~\ref{defn:betareduction}(\ref{defn:betareduction:2}), we have the pairwise distinctness of the variables in the vector $\typd{\vec{x}}{\vec{B}}$. This ensures that the iterated $\beta$-reductions have the effect displayed in the first line. For, this condition prohibits $\typd{\vec{x}}{\vec{B}}$ e.g. being $\typd{x_1}{B_1}, \typd{x_1}{B_1}$; if we had this then the left-hand side of (\ref{eqn:betabetzerofact}) would read $\big(\lbd{x_1}{B_1}{\lbd{x_1}{B_1}{\lbd{v}{A}{L}}}\big) M_2 M_2 N$ instead of $\big(\lbd{x_1}{B_1}{\lbd{x_1}{B_1}{\lbd{v}{A}{L}}}\big) M_1 M_2 N$. 

 The second $\beta$-reduction in (\ref{eqn:betabetzerofact}) follows from the variable $x_i$ being free for itself, and from the pairwise distinctness of the variables in the vector $\typd{\vec{x}}{\vec{B}}$. The third $\beta$-reduction follows from Definition~\ref{defn:betareduction}(\ref{defn:betareduction:1}), namely, $\typd{N}{A}$ being free for $\typd{v}{A}$ in $L$.
\end{proof}

We add however that we do not know the answer to the following question:
\begin{qu}\label{qu:betazerobetarelation}
For $\typd{M,N}{A}$ are terms of $\mltup$, is it the case that $M\equals{\beta}{\param} N$ iff $M\equals{\beta_0}{\param} N$? What if $\eta$ is added?
\end{qu}

The previous proposition and examples concern what one can do with $\beta,\beta_0,\beta_r$-reductions in $\mltup$. Here is a simple example which shows how the limited number of variables in $\mltup$ can prevent even a $\beta_0$-reduction:
\begin{exa}[An example of when lack of variables prevents $\beta_0$-reduction]\label{ex:thelbock}
Suppose $B$ is a state type, $C$ is a regular type, and $\typd{v_0}{B}$ and $\typd{v}{B\rightarrow B\rightarrow C}$ are variables. If $0<j<\param(B)$ then one has the following $\beta_0$-reductions:
\begin{align}
& {\bigg( \lbd{V}{(B\rightarrow C)\rightarrow C}{\lbd{v_j}{B}{V(vv_j)}}\bigg) \big( \lbd{U}{B\rightarrow C}{Uv_0} \big) } \notag\\
\reduces{\beta_0}{\param}\;\; & \lbd{v_j}{B}{\big( \lbd{U}{B\rightarrow C}{Uv_0} \big) (vv_j)} \reduces{\beta_0}{\param} \lbd{v_j}{B}{vv_jv_0}  \label{eqn:thelbock}
\end{align}
However if $\param(B)=1$ and $j=0$, then one does not have the first step of this $\beta$-reduction, since $\lbd{U}{B\rightarrow C}{Uv_0}$ is not free for $\typd{V}{(B\rightarrow C)\rightarrow C}$ in the term $\lbd{v_0}{B}{(V(vv_0))}$. Further, one cannot use $\alpha$-conversion to change $\lbd{v_0}{B}{(V(vv_0))}$ since $\param(B)=1$. 
\end{exa}

For sake of future reference, we extend the following example ever so slightly by putting it under lambda abstracts. The naturalness of this example will emerge in \S\ref{sec:typedcomboterms} when we recognize its initial term as the Cardinal $\ctC{B}{B}{C}{}{}{}$ of combinatory logic. Despite its humble appearance, it plays a leading role in the proof of Theorem~\ref{thmconservation}, proven in \S\ref{sec:conservation}.

\begin{exa}[A naturally occurring example of a $\beta_0$-reduction under lambda abstracts]\label{ex:thelbockxi}
Suppose $B$ is a state type, $C$ is a regular type. If $0<j<\param(B)$ then one can put the previous example under two $\lambda$-abstracts as follows:
\begin{align}
& \lbd{v}{B\rightarrow B\rightarrow C}{\lbd{v_0}{B}{{\bigg( \lbd{V}{(B\rightarrow C)\rightarrow C}{\lbd{v_j}{B}{V(vv_j)}}\bigg) \big( \lbd{U}{B\rightarrow C}{Uv_0} \big) }}} \notag \\
\reduces{\beta_0}{\param}\;\; & \lbd{v}{B\rightarrow B\rightarrow C}{\lbd{v_0}{B}{\lbd{v_j}{B}{\big( \lbd{U}{B\rightarrow C}{Uv_0} \big) (vv_j)}}}\notag \\
\reduces{\beta_0}{\param}\;\; &  \lbd{v}{B\rightarrow B\rightarrow C}{\lbd{v_0}{B}{\lbd{v_j}{B}{vv_jv_0}}} \label{eqn:thelbockxi} 
\end{align}
As with the previous example, this is unavailable if $\param(B)=1$.
\end{exa}

Since the reduction is happening under lambda abstracts, this reduction is unavailable in the weak setting (cf. end of \S\ref{subsec:conventionsredu} and outset of \S\ref{sec:purecombinatory}).

%% file: 02-systems/07-CRandSN.tex
\subsection{Church-Rosser and Strong Normalization}\label{subsec:cr}

Recall that $R$ satisfies \emph{Church-Rosser} if whenever $M\transreduces{R}{} N_0$ and $M\transreduces{R}{} N_1$, then there is $L$ such that $N_0\transreduces{R}{} L$ and $N_1\transreduces{R}{} L$. Further, $\reduces{R}{}$ is \emph{strongly normalizing} if for each $M$ there is no infinite sequence $\{M_i: i\geq 0\}$ with $M_0\equiv M$ and $M_i\reduces{R}{} M_{i+1}$ for all $i\geq 0$. An $R$-\emph{normal form} is an $N$ such that there is no $L$ with $N\reduces{R}{} L$. Hence, strong normalization implies that for every $M$ there is normal form $N$ such that $M\transreduces{R}{} N$.

Since $\beta_0$ and $\eta$ reductions preserve $\mltpure$, one has that Church-Rosser for $\beta_0$-reduction and $\beta_0\eta$-reduction in the ordinary simply typed lambda calculus $\lt$ implies that $\beta_0$ and $\beta_0\eta_0$ in $\mltpure$ satisfy Church-Rosser.\footnote{\cite[Proposition 1.1.9 p. 7]{Barendregt2013-eb}.} However, Church-Rosser for $\beta$-reduction and $\beta\eta$-reduction fails for $\mltup$ with $\param$ non-maximal. This was noted for  $\mlt{1}$ by Friedman and Warren,\footnote{\cite[p. 323]{Friedman1980-pp}.} and their example generalizes.

\begin{exa}[Church-Rosser fails for $\mltup$ with $\param$ non-maximal]\label{ex:friedmanwarren}
We work in $\mltup$ where for some state type $A$, one has that $\param(A)=n$ for some natural number $n\geq 1$. Let $B$ be a regular type. Suppose that $\typd{u_1}{B\rightarrow A\rightarrow B}, \typd{u_2}{A\rightarrow B}, \typd{u_3}{(A\rightarrow B)\rightarrow (A\rightarrow B)\rightarrow B}$ are variables. Further, suppose $\typd{u_4}{A\rightarrow \cdots \rightarrow A\rightarrow B\rightarrow B}$ is a variable, where there are $n-1$ many $A$'s (and in the case $n=1$, there are no $A$'s and $\typd{u_4}{B\rightarrow B}$). Consider the following $\lambda$-term of type $B$:
\begin{equation}
P\equiv \bigg[\lbd{z_1}{B}{\bigg(\big[\lbd{z_2}{B}{\big(u_3(\lbd{v_0}{A}{u_4 v_1 \cdots v_{n-1} z_2})(u_1 z_1)\big)}\big]}z_1 \bigg)\bigg](u_2 v_0)
\end{equation}

By contracting $P$'s outer redex (the one starting with $\lbd{z_1}{B}{\ldots}$) we get:
\begin{equation}\label{eqn:friedmanwarren:Q1}
P\reduces{\beta}{\param} Q_1\equiv \big[\lbd{z_2}{B}{\big(u_3(\lbd{v_0}{A}{u_4 v_1 \cdots v_{n-1} z_2})(u_1 (u_2 v_0))\big)}\big](u_2 v_0)
\end{equation}
Then $Q_1$ cannot be $\beta$-reduced, since $\typd{u_2v_0}{B}$ is not free for $\typd{z_2}{B}$ in the term $u_3(\lbd{v_0}{A}{u_4 v_1 \cdots v_{n-1} z_2})(u_1 (u_2 v_0)$; and this term is the only term in its $\alpha$-equivalence class since we are in $\mltup$ and the other $n-1$ many variables of type $A$ appear freely under the scope of the $\lbd{v_0}{A}{\ldots}$. Then $Q_1$ is a $\beta$ normal form in $\mltup$. Further, $Q_1$ is not of the right form to be $\eta$-reduced. Hence $Q_1$ is in $\beta\eta$ normal form in $\mltup$. 

By contacting $P$'s inner redex (the one starting with $\lbd{z_2}{B}{\ldots}$) we get:
\begin{equation}\label{eqn:friedmanwarren:Q2}
P\reduces{\beta}{\param} Q_2 \equiv \bigg[\lbd{z_1}{B}{\bigg(\big(u_3(\lbd{v_0}{A}{u_4 v_1 \cdots v_{n-1} z_1})(u_1 z_1)\big)}\bigg) \bigg](u_2 v_0)
\end{equation}
Then $Q_2$ cannot be $\beta$-reduced, since $\typd{u_2v_0}{B}$ is not free for $\typd{z_1}{B}$ in the term $\big(u_3(\lbd{v_0}{A}{u_4 v_1 \cdots v_{n-1} z_1})(u_1 z_1)\big)$; and this term is the only term in its $\alpha$-equivalence class since we are in $\mltup$ and the other $n-1$ many variables of type $A$ appear freely under the scope of the $\lbd{v_0}{A}{\ldots}$. Then $Q_2$ is a $\beta$-normal form in $\mltup$. Further, $Q_2$ is not of the right form to be $\eta$-reduced. Hence $Q_2$ is in $\beta\eta$ normal form in $\mltup$. 
\end{exa}

Since $\beta_0$ and $\eta$ reductions preserve $\mltpure$, one further has $\beta_0$-reduction and $\beta_0\eta$-reduction in the ordinary simply typed lambda calculus $\lt$ being strongly normalizing entails that $\beta_0$ and $\beta_0\eta$ reduction in $\mltpure$ is strongly normalizing.\footnote{\cite[p. 64]{Barendregt2013-eb}. See \S\ref{sec:nonmodal} for the relation between $\mltup$ and $\lt$.} Using this one can show:
\begin{prop}\label{prop:SNholdsmltup}
In $\mltup$, both $\beta$ and $\beta\eta$ reduction satisfies strong normalization.
\end{prop}
\noindent Before giving the proof, it is worth noting that Example~\ref{ex:friedmanwarren} shows that $\beta$- and $\beta\eta$-normal forms need not be unique in $\mltup$ for $\param$ non-maximal.
\begin{proof}
Suppose that one had an infinite sequence $\{P_i: i\geq 0\}$ in $\mltup$ such that $P_i\reduces{\beta\eta}{\param} P_{i+1}$ for all $i\geq 0$. We show that this violates strong normalization in $\mltpure$. For the rest of the proof, we view $P_i$ as terms in $\mltpure$. 

By $\alpha$-conversion if necessary, one has that the two sides of a $\beta$-reduction are such that they $\beta_0$-transitively reduce to a common term. For, if $\big( \lbd{\vec{x}}{\vec{B}}{\lbd{v}{A}{L}}\big) \vec{M} N \reduces{\beta}{\param} \big( \lbd{\vec{x}}{\vec{B}}{L[v:=N]} \big) \vec{M}$, then one has that both $\big( \lbd{\vec{x}}{\vec{B}}{\lbd{v}{A}{L}}\big) \vec{M} N \transreduces{\beta_0}{\omega} L[\vec{x}:=\vec{M}, v:=N]$ and $\big( \lbd{\vec{x}}{\vec{B}}{L[v:=N]} \big) \vec{M} \transreduces{\beta_0}{\omega} L[\vec{x}:=\vec{M}, v:=N]$. 

Applying this to $P_0, P_1, P_2$, we have terms $Q_0, Q_1$ in $\mltpure$ satisfying the following in $\mltpure$:
\begin{equation*}
\xymatrix{
P_0  \ar@{->>}[dr]_{\beta_0\eta}  \ar[rr]_{\beta\eta} & & \ar@{->>}[dl]^{\beta_0\eta}  P_1  \ar@{->>}[dr]_{\beta_0\eta} \ar[rr]_{\beta\eta} & & P_2\ar@{->>}[dl]^{\beta_0\eta} \\
& Q_0 & & Q_1 & 
}
\end{equation*}
By Church-Rosser for $\beta_0\eta$ in $\mltpure$, one can extend this chart downwards one level. Continuing in this way, one produces an infinite descending $\beta_0\eta$-reduction sequence in~$\mltpure$.
\end{proof}

%% file: 03-soundness/00-intro.tex
\section{Soundness, and some completeness and open term models}\label{sec:soundness}

%% file: 03-soundness/01-semantics.tex
\subsection{Semantics}\label{sec:semantics}

As usual, the semantics is defined in terms of frames and variable assignments; and using these one can give the inductive definition of denotation. 

\begin{defi}[Frame]\label{defn:frame}
A \emph{frame} $\mathcal{M}$ of $\mltup$ is a sequence of non-empty sets $\mathcal{M}(A)$ for each type $A$ of $\mltup$ such that for all types $A\rightarrow B$ of $\mltup$ one has that $\mathcal{M}(A\rightarrow B)$ is a subset of $\{F:\mathcal{M}(A)\rightarrow \mathcal{M}(B)\}$.

A frame is \emph{standard} if  $\mathcal{M}(A\rightarrow B)=\{F:\mathcal{M}(A)\rightarrow \mathcal{M}(B)\}$.

A \emph{decorated frame} is a frame $\mathcal{M}$ together with an assignment of each constant $\typd{c}{C}$ in the signature to an element $c_{\mathcal{M}}$ in $\mathcal{M}(C)$.
\end{defi}

Variable assignments are defined in the usual way, but they only have to assign the variables dictated by the parameter $\param$.

We use $\rho$ for variable assignments, and we use $\rho[v:=x]$ for the $v$-variant of $\rho$ which assigns $v$ to element $x$.

\begin{defi}[Model and denotation]\label{defn:model}
A \emph{model} $\mathcal{M}$ of $\mltup$ is a decorated frame of $\mltup$ such that for any variable assignment $\rho$ and any term $\typd{M}{A}$ of $\mltup$ one has that the inductively defined denotation $\db{M}_{\mathcal{M},\rho}$ is an element of $\mathcal{M}(A)$:
\begin{enumerate}[leftmargin=*]
    \item \label{defn:model:1} $\db{v_i}_{\mathcal{M},\rho} =\rho(\typd{v_i}{A})$ for variables $\typd{v_i}{A}$ of $\mltup$
    \item \label{defn:model:2} $\db{c}_{\mathcal{M},\rho}=c_{\mathcal{M}}$ for constants $c$ in the signature
    \item \label{defn:model:3} $\db{MN}_{\mathcal{M},\rho}=\db{M}_{\mathcal{M},\rho}\db{N}_{\mathcal{M},\rho}$ for terms $\typd{M}{B\rightarrow C}$ and $\typd{N}{B}$ of $\mltup$
    \item \label{defn:model:4} $\db{\lbd{v_i}{A}{L}}_{\mathcal{M},\rho}=\Lbd{x}{\mathcal{M}(A)}{\db{L}_{\mathcal{M},\rho[v_i:=x]}}$ for terms $\typd{L}{C}$ of $\mltup$ of regular type $C$
\end{enumerate}
\end{defi}

On the right-hand side of (\ref{defn:model:4}), the expression $\Lbd{x}{\mathcal{M}(A)}{\db{M}_{\mathcal{M},\rho[v_i:=x]}}$ means the metatheoretically defined function $F:\mathcal{M}(A)\rightarrow \mathcal{M}(B)$ given by $F(x)=\db{M}_{\mathcal{M},\rho[v_i:=x]}$ for $x$ in $\mathcal{M}(A)$. That is, we are using $\Lambda$ (capital lambda) for the metatheoretically defined lambda abstraction.

A model is \emph{standard} if the underlying frame is standard. A synonym for model is \emph{Henkin model} or sometimes \emph{generalized model}.

The only way in which a decorated frame can fail to be a model is if the metatheoretically defined function in (\ref{defn:model:4}) fails to be an element of $\mathcal{M}(A\rightarrow B)$, since this may be a small subset of the set of functions $\{F:\mathcal{M}(A)\rightarrow \mathcal{M}(B)\}$. Since we often have to argue by induction on complexity of term that a decorated frame is a model, we introduce the following definition: if $\mathcal{M}$ is a decorated frame of $\mltup$ and $\typd{M}{A}$ is a term of $\mltup$, then \emph{the denotation of} $\typd{M}{A}$ \emph{is well-defined} in $\mathcal{M}$ if for all subterms $\typd{N}{B}$ of $\typd{M}{A}$ one has that $\db{N}_{\mathcal{M},\rho}$ as defined in (\ref{defn:model:1})-(\ref{defn:model:4}) are in $\mathcal{M}(B)$, for all variable assignments~$\rho$.

We write $\mathcal{M}\models M=N$ if $\db{M}_{\mathcal{M},\rho}=\db{N}_{\mathcal{M},\rho}$ for all variable assignments $\rho$. We write $\mltup\models M=N$ if for all models $\mathcal{M}$ of $\mltup$ one has that $\mathcal{M}\models M=N$. (Note that there is no identity in the object language of $\mltup$, and rather this is a meta-theoretically defined notion).

Here is an elementary but useful proposition, whose proof we omit:
\begin{prop}[The semantic effect of nested lambda abstracts]\label{prop:helpervec}
For any model $\mathcal{M}$ and variable assignment $\rho$ relative to $\mathcal{M}$ and vector of terms $\typd{\vec{M}}{\vec{A}}$ of length $n$ and vector of variables $\typd{\vec{v}}{\vec{A}}$ of length $n$, define a sequence of variable assignments $\rho_0, \rho_1, \ldots, \rho_n$ by $\rho_0=\rho$ and $\rho_{i+1} = \rho_i[v_i:=\db{M_i}_{\mathcal{M},\rho}]$ for $0\leq i<n$. Then for all terms $\typd{L}{C}$, one has $\db{\big( \lbd{\vec{v}}{\vec{A}}{L}\big)\vec{M}}_{\mathcal{M},\rho}=\db{L}_{\mathcal{M},\rho_n}$.
\end{prop}
\begin{proofdetail}
\begin{proof}
The proof is by induction on $n$. For $n=1$ we have $\db{\big(\lbd{v_0}{A}{L}\big)M_0}_{\mathcal{M},\rho}=\db{\lbd{v_0}{A}{L}}_{\mathcal{M},\rho} \db{M_0}_{\mathcal{M},\rho} = \db{L}_{\mathcal{M},\rho_1}$. In this, the first identity follows from the semantics for application, and the second identity follows from the semantics for lambda abstraction and the definition of $\rho_1 = \rho [v_0:=\db{M}_{\mathcal{M},\rho}]$. Suppose it holds for $n$; we show it holds for $n+1$:
\begin{align*}
& \db{\big(\lbd{\vec{v}}{\vec{A}}{\lbd{v_{n}}{A_{n}}{L}}\big) \vec{M} M_{n}}_{\mathcal{M},\rho} =\db{\big(\lbd{\vec{v}}{\vec{A}}{\lbd{v_{n}}{A_{n}}{L}}\big) \vec{M}}_{\mathcal{M},\rho} \db{M_{n}}_{\mathcal{M},\rho} \\
& = \db{\lbd{v_{n}}{A}{L}}_{\mathcal{M},\rho_n}  \db{M_{n}}_{\mathcal{M},\rho} = \db{L}_{\mathcal{M},\rho_{n+1}}
\end{align*}
In this, the first identity follows from semantics for application; the second identity follows from induction hypothesis; and the third identity follows from the semantics for lambda abstraction and the definition $\rho_{n+1}=\rho_n[v_{n+1}:=\db{M_{n+1}}_{\mathcal{M},\rho}]$.
\end{proof}
\end{proofdetail}

%% file: 03-soundness/02-soundness-validity-reductions.tex
\subsection{The validity of the reductions and soundness}\label{subsec:soundness}

In this subsection we prove that $\beta\eta$-reductions are valid on $\mltup$-structures, in the strong form that the well-definedness of the denotation of the redex implies the well-definedness of the denotation of the contractum. 

\begin{prop}\label{prop:prebeta1}
Suppose that $\mathcal{M}$ is a decorated frame of $\mltup$. 
\begin{enumerate}[leftmargin=*]
\item\label{prop:prebeta1:1} For all terms $\typd{L}{B}$ of $\mltup$ and all variables $\typd{v}{A}$ of $\mltup$ and all terms $\typd{N}{A}$ of $\mltup$, if the denotations of $\typd{L}{B},\typd{N}{A}$ are well-defined in $\mathcal{M}$ and $\typd{N}{A}$ is free for $\typd{v}{A}$ in $\typd{L}{B}$, then the denotation of $L[v:=N]$ is well-defined and for all variable assignments $\rho$ one has $\db{L[v:=N]}_{\mathcal{M},\rho} = \db{L}_{\mathcal{M},\rho[v:=\db{N}_{\mathcal{M},\rho}]}$.
\item \label{prop:prebeta2} For all terms $\big( \lbd{\vec{x}}{\vec{B}}{\lbd{v}{A}{L}}\big) \vec{M} N$ of $\mltup$ whose denotation is well-defined in $\mathcal{M}$, if $\typd{N}{A}$ is free for $\typd{v}{A}$ in $\typd{L}{C}$ and if $\typd{\vec{x}}{\vec{B}}$ is not free in $\typd{N}{A}$, then the denotation of $\big(\lbd{\vec{x}}{B}{L[v:=N]} \big) \vec{M}$ is well-defined in $\mathcal{M}$ and $\db{\big( \lbd{\vec{x}}{\vec{B}}{\lbd{v}{A}{L}}\big) \vec{M} N}_{\mathcal{M},\rho}=\db{\lbd{\vec{x}}{B}{L[v:=N]} \big) \vec{M}}_{\mathcal{M},\rho}$ for all variable assignments $\rho$.
\end{enumerate}
\end{prop}
\begin{proof}
The usual inductive proof of (\ref{prop:prebeta1:1}) works for $\mltup$ (cf. \cite[Lemma 3.1.13 p. 98]{Barendregt2013-eb}). For (\ref{prop:prebeta2}), suppose that the vector $\typd{\vec{x}}{\vec{B}}$ has length $n$ and is $\typd{x_1}{B_1}, \ldots, \typd{x_n}{B_n}$. Let $\rho$ be a variable assignment relative to $\mathcal{M}$. Define variable assignment $\rho_0 = \rho$ and $\rho_{i+1} = \rho_i[x_i:=\db{M_i}_{\mathcal{M},\rho}] $ for $0\leq i<n$. Then one has the following:
\begin{align*}
&  \hspace{5mm}\db{\big( \lbd{\vec{x}}{\vec{B}}{\lbd{v}{A}{L}}\big)\; \vec{M}\; N}_{\mathcal{M}, \rho} & & \\
& = \db{\big( \lbd{\vec{x}}{\vec{B}}{\lbd{v}{A}{L}}\big)\; \vec{M}}_{\mathcal{M},\rho} \;\db{N}_{\mathcal{M}, \rho} & \mbox{by semantics of app.} & \\
& = \db{\lbd{v}{A}{L}}_{\mathcal{M}, \rho_n} \;\db{N}_{\mathcal{M},\rho} & \mbox{by Proposition~\ref{prop:helpervec}}&  \\
& = \db{\lbd{v}{A}{L}}_{\mathcal{M}, \rho_n}\; \db{N}_{\mathcal{M},\rho_n} & \mbox{since $\typd{\vec{x}}{\vec{B}}$ not free in $\typd{N}{A}$} \\
& = \db{L}_{\mathcal{M}, \rho_n[v:= \db{N}_{\mathcal{M},\rho_n}]} & \mbox{by semantics of $\lambda$-abs}& \\
& = \db{L[v:=N]}_{\mathcal{M},\rho_n} & \mbox{by (\ref{prop:prebeta1:1})} \\
& =\db{\big( \lbd{\vec{x}}{\vec{B}}{L[v:=N]}\big)\; \vec{M}}_{\mathcal{M}, \rho}  & \mbox{by Proposition~\ref{prop:helpervec}} & \qedhere
\end{align*}
\end{proof}

\begin{proofdetail}

In more detail, for (\ref{prop:prebeta1:1}), we argue by induction on the complexity of $\typd{L}{B}$ as follows.

For the variables, there are two cases:

Suppose $L$ is $\typd{v}{A}$. Then $L[v:=N]$ is $N$, and hence the denotation of $L[v:=N]$ is well-defined. Further, one has that both $\db{L[v:=N]}_{\mathcal{M},\rho}, \db{L}_{\mathcal{M},\rho[v:=\db{N}_{\mathcal{M},\rho}]}$ are equal to $\db{N}_{\mathcal{M},\rho}$.

Suppose $L$ is $\typd{u}{B}$, where this is a distinct variable from $\typd{v}{A}$. Then $L[v:=N]$ is $u$, and hence the denotation of $L[v:=N]$ is well-defined.  Further, one has that both $\db{L[v:=N]}_{\mathcal{M},\rho}, \db{L}_{\mathcal{M},\rho[v:=\db{N}_{\mathcal{M},\rho}]}$ are equal to $\rho(u)$.

Suppose that $L$ is a constant $c$. Then $L[v:=N]$ is $c$ as well, and hence the denotation of $L[v:=N]$ is well-defined.  Further, one has that both $\db{L[v:=N]}_{\mathcal{M},\rho}, \db{L}_{\mathcal{M},\rho[v:=\db{N}_{\mathcal{M},\rho}]}$ are equal to $c_{\mathcal{M}}$.

Suppose that $L$ is the application $PQ$. Then $L[v:=N]$ is $(P[v:=N])(Q[v:=N])$, and $N$ is free for $v$ in both $P,Q$.  By induction hypothesis, one has that the denotation of both $P[v:=N]$ and $Q[v:=N]$ are well-defined and $\db{P[v:=N]}_{\mathcal{M},\rho} = \db{P}_{\mathcal{M},\rho[v:=\db{N}_{\mathcal{M},\rho}]}$ and $\db{Q[v:=N]}_{\mathcal{M},\rho} = \db{Q}_{\mathcal{M},\rho[v:=\db{N}_{\mathcal{M},\rho}]}$. Then the denotation of $L[v:=N]$ is well-defined, namely $\db{L[v:=N]}_{\mathcal{M},\rho} =\db{P[v:=N]}_{\mathcal{M},\rho}(\db{Q[v:=N]}_{\mathcal{M},\rho})$. By induction hypothesis, this is $\db{P}_{\mathcal{M},\rho[v:=\db{N}_{\mathcal{M},\rho}]}(\db{Q}_{\mathcal{M},\rho[v:=\db{N}_{\mathcal{M},\rho}]}) = \db{L}_{\mathcal{M},\rho[v:=\db{N}_{\mathcal{M},\rho}]}$.

Suppose that $L$ is the lambda abstract $\lbd{v}{A}{Q}$. Then $v$ does not appear free in $L$, and so $L[v:=N]$ is just $L$, and so the denotation of $L[v:=N]$ is well-defined. Further, one has that both $\db{L[v:=N]}_{\mathcal{M},\rho}, \db{L}_{\mathcal{M},\rho[v:=\db{N}_{\mathcal{M},\rho}]}$ are equal to $\db{L}_{\mathcal{M},\rho}$. 

 Suppose that $L$ is the lambda abstract $\lbd{u}{C}{Q}$, where $\typd{Q}{D}$, and where $\typd{u}{C}$ is distinct from $\typd{v}{A}$. Then $L[v:=N]$ is $\lbd{u}{C}{(Q[v:=N])}$, and $N$ is free for $v$ in $Q$. By induction hypothesis, the denotation of $Q[v:=N]$ is well-defined and $\db{Q[v:=N]}_{\mathcal{M},\rho} = \db{Q}_{\mathcal{M},\rho[v:=\db{N}_{\mathcal{M},\rho}]}$ for all variable assignments $\rho$. Let $\rho$ be a variable assignment, and let variable assignment $\sigma$  be $\rho[v:=\db{N}_{\mathcal{M},\rho}]$. Since the denotation of $L$ is well-defined, one has that $\db{L}_{\mathcal{M},\sigma}$ is an element of $\mathcal{M}(C\rightarrow D)$. It suffices to show that $\db{L[v:=N]}_{\mathcal{M},\rho} = \db{L}_{\mathcal{M},\sigma}$. We are done if $\typd{v}{A}$ does not appear free in $L$, since in this case $L[v:=N]$ is just $L$ and $\db{L}_{\mathcal{M},\sigma} = \db{L}_{\mathcal{M},\rho}$. Hence, suppose that $\typd{v}{A}$ does appear free in $L$, and hence free in $Q$.
Since $N$ is free for $v$ in $L$, which is $\lbd{u}{C}{Q}$, one has that $\typd{u}{C}$ is not free in $N$. Then for all $z$ in $\mathcal{M}(C)$ we have:
\begin{align*}
& \hspace{5mm} \db{L[v:=N]}_{\mathcal{M},\rho}(z) & \\
& = \db{\lbd{u}{C}{(Q[v:=N])}}_{\mathcal{M},\rho}(z) & \mbox{by defn. of $L[v:=N]$} \\
& = \db{Q[v:=N]}_{\mathcal{M},\rho[u:=z]}  & \mbox{by semantics of $\lambda$-abs.} \\
& = \db{Q}_{\mathcal{M},\rho[u:=z][v:=\db{N}_{\mathcal{M},\rho[u:=z]}]}  & \mbox{by IH applied to $Q$ and  $\rho[u:=z]$}\\
& = \db{Q}_{\mathcal{M},\rho[u:=z] [v:=\db{N}_{\mathcal{M},\rho}]} & \mbox{by $\typd{u}{C}$ not free in $N$} \\
& = \db{Q}_{\mathcal{M},\rho[v:=\db{N}_{\mathcal{M},\rho}][u:=z] } & \mbox{by $\typd{u}{C},\typd{v}{A}$ distinct} \\
& = \db{Q}_{\mathcal{M},\sigma[u:=z]}  & \mbox{by defn. of $\sigma$} \\
& = \db{\lbd{u}{C}{Q}}_{\mathcal{M},\sigma}(z) &\mbox{by semantics of $\lambda$-abs.} \\
& = \db{L}_{\mathcal{M},\sigma}(z) & \mbox{by defn of $L$} 
\end{align*}

\end{proofdetail}

\begin{prop}\label{prop:reducebetadelta}
Suppose that $\mathcal{M}$ is a decorated frame of $\mltup$.  For all terms $M,N$ of $\mltup$, if $N\transreduce{\beta\tau\eta} M$ and the denotation of $N$ is well-defined in $\mathcal{M}$, then the denotation of $M$ is well-defined in $\mathcal{M}$, and for all variable assignments $\rho$, one has that $\db{N}_{\mathcal{M},\rho}=\db{M}_{\mathcal{M},\rho}$.
\end{prop}
\begin{proof}
The base case for $\beta$ is Proposition~\ref{prop:prebeta1}(\ref{prop:prebeta2}); the base case for $\tau$ (cf. Definition~\ref{defn:tau}) is trivial; the base case for $\eta$ is standard. The inductive steps are trivial.
\end{proof}

We will use the previous propositions in subsequent sections. For the moment, we note the following direct consequence:
\begin{thm}[Soundness Theorem for $\mltup$]\label{thm:soundness}
If $M,N$ are terms of $\mltup$ with $\mltup\vdash _{\beta\eta\tau} M=N$ then $\mltup\models M=N$.
\end{thm}

%% file: 03-soundness/03-completeness-open-term.tex
\subsection{Some completeness via open term models}\label{sec:completeness-open-term}

Now we turn to completeness for~$\mltpure$, the maximal theory, which we can prove directly using traditional arguments. We discuss, as we proceed, where the traditional arguments break down for non-maximal parameters.

The following modifies the traditional construction to $\mltup$:\footnote{\cite[Definition 3.2.9 p. 109]{Barendregt2013-eb}.}
\begin{defi}\label{defn:openterm}
The \emph{open term applicative structure} $\mathcal{O}_{\param}$ for $\mltup$ in a signature is defined so that $\mathcal{O}_{\param}(A)$ is the set of equivalence classes $[M]$ of the set of terms $\typd{M}{A}$ of $\mltup$ under the equivalence relation of $\equals{\beta\eta}{\param}$. The application operation is given pointwise $[M][N]=[MN]$. The interpretation of constants $\typd{c}{A}$ is given by $c_{\mathcal{O}_{\param}}=[c]$.

In the case of $\mltpure$, we refer to the open term applicative structure as $\mathcal{O}_{\omega}$, and in the case of $\mlt{n}$, we refer to the open term applicative structure as $\mathcal{O}_n$.
\end{defi}

\noindent In this subsection, for the sake of simplicity we discuss $\mathcal{O}_n$ for $n\geq 1$, rather than treating separately the cases of $\mathcal{O}_{\param}$ for general parameters $\param$.

In the case of $\mathcal{O}_{\omega}$, we can turn the open term applicative structure into a frame in the sense of Definition~\ref{defn:frame} by noting that the elements $[M]$ of  $\mathcal{O}_{\omega}(A\rightarrow B)$ are in one-one correspondence with a subclass of functions $F: \mathcal{O}_{\omega}(A)\rightarrow \mathcal{O}_{\omega}(B)$:\footnote{This is the argument of \cite[ Proposition 3.2.10(i)]{Barendregt2013-eb} adapted to $\mltpure$.}
\begin{prop}\label{prop:opentermstructureprop}
Every element $[M]$ of $\mathcal{O}_{\omega}(A\rightarrow B)$ determines a function in $\{F: \mathcal{O}_{\omega}(A)\rightarrow \mathcal{O}_{\omega}(B)\}$ by setting $F_{[M]}[N]=[MN]$. Moreover, the map $[M]\mapsto F_{[M]}$ is injective. Hence the open term applicative structure $\mathcal{O}_{\omega}$ determines a frame.
\end{prop}
\begin{proof}
The function is well-defined since if $M,M^{\prime}$ and $N,N^{\prime}$ are $\beta\eta$-equivalent and of the appropriate type, then so are $MN, M^{\prime}N^{\prime}$. To see that the function is injective, suppose $\typd{M}{A\rightarrow B},\typd{N}{A\rightarrow B}$ with $[M][L]=[N][L]$ for all terms $\typd{L}{A}$ of $\mltpure$. Then we can choose a variable $\typd{x}{A}$ which does not occur freely in $M,N$, and from $Mx=_{\beta\eta} Nx$ we can infer that $\lbd{x}{A}{Mx}=_{\beta\eta} \lbd{x}{A}{Nx}$, and then by $\eta$ to $M\equal{\beta\eta} N$ and then to $[M]=[N]$.
\end{proof}

However, the previous proposition is not true for $\mathcal{O}_1$, and hence the open term applicative structure~$\mathcal{O}_1$ does not determine a frame:
\begin{exa}\label{exa:failureopenterm1} 
In $\mlt{1}$ suppose that $A$ is a state type and $B$ a regular type and $\typd{U}{A\rightarrow A \rightarrow B}$ is a variable. Let $M$ be $\lbd{v_0}{A}{Uv_0v_0}$ and let $N$ be $Uv_0$.\footnote{Note that $M$ is just the result of applying the Warbler to variable $U$ (cf. Definition~\ref{defn:typedcombo}).} Then $Mv_0\equals{\beta\eta}{1} Uv_0v_0$, while $Nv_0$ is just $Uv_0v_0$. Hence $Mv_0\equals{\beta\eta}{1} Nv_0$ but it is not the case that $M\equals{\beta\eta}{1} N$, as one can see by considering standard models.
\end{exa}

\noindent It is interesting to ask why one cannot easily produce similar examples for $n>1$. The results in \S\ref{sec:opentermmodelsredeux} are partially explanatory of this.

For $\mltpure$ one can show that the frame determined by $\mathcal{O}_{\omega}$ is a model, using the traditional argument.\footnote{ \cite[Proposition 3.2.10 p. 110(ii)]{Barendregt2013-eb}.}

\begin{prop}\label{prop:opentermmodel}
The frame determined by the open term applicative structure $\mathcal{O}_{\omega}$ is a model. 
\end{prop}

\begin{proof}
We show that by induction on a term  $\typd{M}{C}$ of $\mltpure$ with free variables from the vector $\typd{\vec{u}}{\vec{A}}$ of length $\ell$ that if $\rho$ is a variable assignment with $\rho(\typd{u_i}{A_i})=[P_i]$ for each $ i< \ell$ and if $P_i$ is free for $u_i$ in $M$ for each $i< \ell$, then $\db{M}_{\mathcal{O}_{\omega}, \rho} = [M[\vec{u}:=\vec{P}]]$.

Suppose $\typd{M}{C}$ is a variable $\typd{u_1}{A_1}$. Then $\db{M}_{\mathcal{O}_{\omega},\rho}=\rho(\typd{u_1}{A_1})=[P_1] = [M[u_1:=P_1]]$. 

Suppose $\typd{M}{C}$ is a constant $\typd{c}{C}$. Then $\db{M}_{\mathcal{O}_{\omega},\rho}=c_{\mathcal{O}_{\omega}}=[c]=[M]$.

Suppose it holds for $\typd{M}{A\rightarrow B}$ and $\typd{N}{A}$. Then one has $\db{MN}_{\mathcal{O}_{\omega},\rho}=\db{M}_{\mathcal{O}_{\omega}, \rho}\db{N}_{\mathcal{O}_{\omega}, \rho} = [M[\vec{u}:=\vec{P}]] [N[\vec{u}:=\vec{P}]]$, which is equal to $[(MN)[\vec{u}:=\vec{P}]]$.

Suppose it holds for $\typd{M}{C}$. We show it holds for $\lbd{v}{A}{M}$. Since $\lbd{v}{A}{M}$ has free variables $\typd{\vec{u}}{\vec{A}}$, these are distinct from $\typd{v}{A}$. We want to show that $\db{\lbd{v}{A}{M}}_{\mathcal{O}_{\omega}, \rho} = [(\lbd{v}{A}{M})[\vec{u}:=\vec{P}]
]$. Since $\typd{\vec{u}}{\vec{A}}$ is distinct from $\typd{v}{A}$, we have that $(\lbd{v}{A}{M})[\vec{u}:=\vec{P}]$ is $\lbd{v}{A}{M[\vec{u}:=\vec{P}]}$, and hence we want to show that  $\db{\lbd{v}{A}{M}}_{\mathcal{O}_{\omega}, \rho} = [\lbd{v}{A}{M}[\vec{u}:=\vec{P}]
]$. Since $\mathcal{O}_{\omega}$ satisfies extensionality, it suffices to work with elements. Hence, suppose that $\typd{Q}{A}$ is a term of $\mltpure$.  We must show the identity $\db{\lbd{v}{A}{M}}_{\mathcal{O}_{\omega}, \rho} [Q] = [\lbd{v}{A}{M[\vec{u}:=\vec{P}]}][Q]$. 

Since we are working in $\mltpure$, by working with an $\alpha$-equivalent of $M[\vec{u}:=\vec{P}]$ if need be, we may assume that $Q$ is free for $v$ in $M[\vec{u}:=\vec{P}]$. Since $\typd{\vec{u}}{\vec{A}}$ is distinct from $\typd{v}{A}$, we have that $Q$ is free for $v$ in $M$ as well. 

Again, we must show that $\db{\lbd{v}{A}{M}}_{\mathcal{O}_{\omega}, \rho} [Q] = [\lbd{v}{A}{M[\vec{u}:=\vec{P}]}][Q]$. By the semantics for lambda abstracts we can rewrite the left-hand side as $\db{M}_{\mathcal{O}_{\omega}, \rho[v:=[Q]]}$, and by $Q$ being free for $v$ in $M[\vec{u}:=\vec{P}]$ and $\beta$-conversion 
we can rewrite the right-hand side as $[M[\vec{u}:=\vec{P}][v:=Q]]$. Since $P_i$ is free for $u_i$ in $\lbd{v}{A}{M}$, we have that $v$ is not free in $P_i$, for each $i< \ell$. By this and the fact that $\typd{\vec{u}}{\vec{A}}$ is distinct from $\typd{v}{A}$, we can further simplify the right-hand side so that the identity reads $\db{M}_{\mathcal{O}_{\omega}, \rho[v:=[Q]]} = [M[\vec{u}:=\vec{P}, v:=Q]]$. But this identity then follows from the induction hypothesis for $M$ and $\rho[v:=[Q]]$.
\end{proof}

Note that the proof of Proposition~\ref{prop:opentermmodel} breaks down for non-maximal parameters in the second-to-last paragraph.

In the proof of Proposition~\ref{prop:opentermmodel}, we have included in the induction hypothesis the qualification ``$P_i$ is free for $u_i$ in $M$ for each $i< \ell$.'' This is often not included in the proof for the ordinary simply typed lambda calculus because there one can assume a convention to the effect that $\alpha$-conversion has been applied to avoid variable capture.\footnote{See the appeal to the variable convention in the proof of \cite[Proposition 3.2.10(ii) p. 110]{Barendregt2013-eb}.} Since we cannot institute this convention for non-maximal parameters, out of consistency we have included the qualification for the maximal parameter. The below example shows what can happen if one fails to include the qualification in a setting, such as that of this paper, where the aforementioned convention is not in force:

\begin{exa}
Suppose that $A,B$ are types with $B$ regular. Suppose that $\typd{x,y}{A}$ are distinct variables, and that $\typd{U}{A\rightarrow A\rightarrow B}$ is a variable. Let $\typd{M}{A\rightarrow B}$ be $\lbd{x}{A}{Uxy}$. Let $\rho$ be a variable assignment relative to $\mathcal{O}_{\omega}$ with $\rho(y)=[x]$ and $\rho(U)=[U]$. Then $\db{M}_{\mathcal{O}_{\omega}, \rho}=\Lambda\, z:\mathcal{O}_{\omega}(A) . \, [U]z[x]$. This has the same input-output behavior as $[\lbd{y}{A}{Uyx}]$ and hence they are identical. But $[M[y:=x]]$ is $[\lbd{x}{A} Uxx]$, which is not identical to $[\lbd{y}{A}{Uyx}]$.
\end{exa}

Finally, for $\mltpure$, we can directly show using the traditional argument:\footnote{\cite[Theorem 3.2.12]{Barendregt2013-eb}.}
\begin{restatable}[Completeness Theorem for $\mltpure$]{thm}{thmcompletelambda}\label{thmcompletelambda}
For terms $M,N$ of $\mltpure$, we have $\mltpure\vdash_{\beta\eta} M=N$ iff $\mltpure\models M=N$.
\end{restatable}

\begin{proof}
Given Soundness (Theorem~\ref{thm:soundness}), only the backwards direction needs argument. Suppose that $\mltpure\models M=N$. Since the frame determined by $\mathcal{O}_{\omega}$ is a model, one has that $\db{M}_{\mathcal{O}_{\omega},\rho}=\db{N}_{\mathcal{O}_{\omega},\rho}$ for all variable assignments $\rho$. Enumerate the free variables of $M,N$ in a vector $\typd{\vec{u}}{\vec{A}}$ of length $\ell$. Let $\rho$ be the variable assignment which assigns $\rho(\typd{u_i}{A_i})=[u_i]$ for each $i<\ell$. Then by the result proven in Proposition~\ref{prop:opentermmodel} one has $[M]=[M[\vec{u}:=\vec{u}]]=\db{M}_{\mathcal{O}_{\omega},\rho}=\db{N}_{\mathcal{O}_{\omega},\rho}=[N[\vec{u}:=\vec{u}]]=[N]$. Then $M, N$ are $\beta\eta$-equivalent.
\end{proof}

%% file: 03-soundness/04-automorphisms.tex
\subsection{Automorphisms and inexpressibility}\label{sec:automorphisms}

The main theorems of this paper concern the expressive power of $\mltup$ \emph{vis-\`a-vis} $\mltpure$, and in turn of $\mltpure$ \emph{vis-\`a-vis} $\lt$ (cf. \S\ref{sec:nonmodal} for the relation between $\mltup$ and $\lt$). However, there are also some evident inexpressibility results pertaining of $\mltup$ and $\lt$ which are simple and important to state. As in many other areas of logic, a basic tool to show inexpressibility is automorphisms. After defining them for $\mltup$ and noting that automorphisms preserve denotations, we show that the actuality operator from Example~\ref{exa:namedworlds}, which uses constants, cannot be expressed without constants. Then we show that the comultiplication operators from some elementary comonads have type in $\mltup$ and are expressible in $\lt$ with constants of $\lt$ but are not expressible in $\mltup$, which does not have these constants.

Given a frame $\mathcal{M}$ of $\mltup$, an \emph{internal permutation} $\pi$ of $\mathcal{M}$ is given by a family of permutations $\pi_{A}: \mathcal{M}(A)\rightarrow \mathcal{M}(A)$ of the atomic types $A$ such that (i) if $A$ is a basic entity type then $\pi_A$ is in  $\mathcal{M}(A\rightarrow A)$, and such that (ii) if $A\rightarrow B$ is a regular type, then the permutation $\pi_{A\rightarrow B}:\mathcal{M}((A\rightarrow B)\rightarrow (A\rightarrow B))$ defined by $\pi_{A\rightarrow B}(f) = \pi_{B} \circ f \circ \pi_{A}^{-1}$ is an element of $\mathcal{M}((A\rightarrow B)\rightarrow (A\rightarrow B))$. If $\mathcal{M}$ is in addition a decorated frame, then an \emph{internal automorphism} $\pi$ of $\mathcal{M}$ is an internal permutation such that $\pi_C(c_{\mathcal{M}})=c_{\mathcal{M}}$ for all constants $\typd{c}{C}$ of the signature. Even when $\mathcal{M}$ is a model, the requirement of membership of $\pi_{A\rightarrow B}$ in $\mathcal{M}((A\rightarrow B)\rightarrow (A\rightarrow B))$ is non-trivial when $A$ is a state type since $\pi_A$ is not an element of any $\mathcal{M}(C)$ for any type $C$ of $\mltup$, due to $A\rightarrow A$ not being a type of $\mltup$.

The following proposition is the version of the familiar ``automorphisms preserve truth-values'' result which is available in this setting; its proof is standard and so we omit it:

\begin{prop}\label{prop:auto:preserve}
If $\pi$ is an internal automorphism of a model $\mathcal{M}$ of $\mltup$, then for all types~$A$ of $\mltup$ and all terms $\typd{M}{A}$ and all variable assignments $\rho$ relative to $\mathcal{M}$, one has that $\pi_{A}(\db{M}_{\mathcal{M},\rho}) = \db{M}_{\mathcal{M},\pi\circ \rho}$.
\end{prop}
\noindent In this, the variable assignment $\pi\circ \rho$ is defined by $(\pi\circ \rho)(\typd{v}{A}) = \pi_A(\rho(\typd{v}{A}))$. 

\begin{proofdetail}
\begin{proof}
The proof is by induction on complexity:
\begin{itemize}
\item For variable $\typd{v}{A}$, one has $\pi_{A}(\db{v}_{\mathcal{M},\rho}) = \pi_{A}(\rho(\typd{v}{A})) = (\pi\circ \rho)(\typd{v}{A}) =  \db{v}_{\mathcal{M},\pi\circ \rho}$.
\item For constant $\typd{c}{C}$ one has $\pi_{C}(\db{c}_{\mathcal{M},\rho}) = \pi_C(c_{\mathcal{M}}) =c_{\mathcal{M}} = \db{c}_{\mathcal{M},\pi\circ \rho}$.
\item For application $\typd{MN}{A\rightarrow B}$ where $\typd{M}{A\rightarrow B}$ and $\typd{N}{A}$, we have $\pi_{B}(\db{MN}_{\mathcal{M},\rho}) = \pi_{B}(\db{M}_{\mathcal{M},\rho}(\db{N}_{\mathcal{M},\rho})) = \pi_{B}(\db{M}_{\mathcal{M},\rho}(\pi_{A}^{-1}(\pi_{A}(\db{N}_{\mathcal{M},\rho})))) = (\pi_{A\rightarrow B}(\db{M}_{\mathcal{M},\rho})) (\pi_{A} (\db{N}_{\mathcal{M},\rho})) =  (\db{M}_{\mathcal{M}, \pi\circ \rho})( \db{N}_{\mathcal{M}, \pi\circ \rho}) =  \db{MN}_{\mathcal{M},\pi\circ \rho}$, where the penultimate identity follows from the induction hypothesis.
\item For lambda abstraction $\typd{\lbd{v}{A}{N}}{A\rightarrow B}$ where $\typd{v}{A}$ is a variable and $\typd{N}{B}$, suppose that $x$ is in $\mathcal{M}(A)$. Then we have $(\pi_{A\rightarrow B}(\db{\lbd{v}{A}{N}}_{\mathcal{M},\rho}))(x) = \pi_{B} (\db{\lbd{v}{A}{N}}_{\mathcal{M},\rho}(\pi_{A}^{-1}(x))) = \pi_{B}(\db{N}_{\mathcal{M},\rho[v:=\pi_{A}^{-1}(x)]}) =  \db{N}_{\mathcal{M},\pi\circ (\rho[v:=\pi_{A}^{-1}(x)])} = \db{N}_{\mathcal{M},(\pi\circ \rho)[v:=x])} = \db{\lbd{v}{A}{N}}_{\mathcal{M},\pi\circ \rho}(x)$, where the antepenultimate identity follows from the induction hypothesis, and where the penultimate identity follows from the fact that variable assignments $\pi\circ (\rho[v:=\pi_{A}^{-1}(x)])$ and $(\pi\circ \rho)[v:=x]$ both assign $v$ to $x$ and are both equal to $\pi\circ \rho$ on all other variables.
\end{itemize}
\end{proof}
\end{proofdetail}

As a simple application, recall from Example~\ref{exa:namedworlds} the actuality operator $\lbd{f}{A\rightarrow B}{fc}$ where $\typd{c}{A}$ is a constant of state type $A$. The traditional discussion in propositional and predicate modal logics is whether this actuality operator is an expressive enrichment (cf. \cite{Hazen1978-ic}, \cite{Hazen2013-ob}). The question in $\mltup$ is then whether constant symbols of state type are an expressive enrichment. The following simple application of Proposition~\ref{prop:auto:preserve} shows that they are:
\begin{prop}[Non-expressibility of actuality operator in the empty signature]\label{exa:actualityisnot}
Suppose that $\mathcal{M}$ is a standard model in the empty signature. Suppose $A$ is a state type such that $\mathcal{M}(A)$ has at least two elements. Suppose that $B$ is a regular type formed from basic entity types and arrow such that $\mathcal{M}(B)$ has at least two elements. 

For each $w$ in $\mathcal{M}(A)$ let $F_w$ in $\mathcal{M}((A\rightarrow B)\rightarrow B)$ be defined by $F_w(f)=f(w)$. Then for each $w$ in $\mathcal{M}(A)$ there is no closed expression $\typd{M}{A\rightarrow B}$ of $\mltup$ in the empty signature such that $\db{M}_{\mathcal{M}}=F_w$. 
\end{prop}
\begin{proof}
To see this, suppose for reductio that $\db{M}_{\mathcal{M}}=F_w$. Since $\mathcal{M}(A)$ has at least two elements, choose $v$ in $\mathcal{M}(A)$ which is distinct from $w$. Let $\pi_A$ be a permutation which transposes $w,v$. Since $\mathcal{M}$ is standard, extend $\pi$ to an internal automorphism of $\mathcal{M}$ by setting $\pi_C$ to be the identity for any basic entity type $C$ and by setting, for any functional type $C\rightarrow D$ and any $f$ in $\mathcal{M}(C\rightarrow D)$, the identity $\pi_{C\rightarrow D}(f) = \pi_{D} \circ f \circ \pi_{C}^{-1}$. Note two things:
\begin{enumerate}[leftmargin=*]
    \item\label{exa:standout:1} For any functional type $C\rightarrow D$ of $\mltup$ and any $g$ in $\mathcal{M}(C\rightarrow D)$ one has  $\pi_{C\rightarrow D}^{-1}(g) = \pi_{D}^{-1} \circ g \circ \pi_{C}$.
    \item\label{exa:standout:2} For a regular type $C$ of $\mltup$ formed from basic entity types and arrow, one has $\pi_C$ is the identity.
\end{enumerate}

For any $G$ in $\mathcal{M}((A\rightarrow B)\rightarrow B)$ we have that $\pi_{(A\rightarrow B)\rightarrow B} G=\pi_B \circ G\circ \pi_{A\rightarrow B}^{-1}=G\circ \pi_{A\rightarrow B}^{-1}$. Further, for any $f$ in $\mathcal{M}(A\rightarrow B)$ one has that 
\begin{equation}
(\pi_{(A\rightarrow B)\rightarrow B} G)f = G(\pi_{A\rightarrow B}^{-1}(f)) = G(\pi_B^{-1}\circ f\circ \pi_A) = G(f\circ \pi_A)
\end{equation}
By the reductio hypothesis and Proposition~\ref{prop:auto:preserve} one has that 
\begin{equation}
\pi_{(A\rightarrow B)\rightarrow B} \db{M}_{\mathcal{M}}=\pi_{(A\rightarrow B)\rightarrow B}\db{M}_{\mathcal{M},\rho}=\db{M}_{\mathcal{M},\pi\circ \rho}=\db{M}_{\mathcal{M}}=F_w
\end{equation}
where the first and second-to-last identities follow since $\typd{M}{(A\rightarrow B)\rightarrow B}$ is closed. Then for all $f$ in $\mathcal{M}(A\rightarrow B)$ one has $f(w)=F_w(f)=(\pi_{(A\rightarrow B)\rightarrow B} \db{M}_{\mathcal{M}})f = \db{M}_{\mathcal{M}}(f\circ \pi_A) = F_w(f\circ \pi_A) = (f\circ \pi_A)(w) = f(\pi_A(w))=f(v)$.

But since $\mathcal{M}$ is standard and $\mathcal{M}(B)$ has at least two elements, we can choose $f:A\rightarrow B$ such that $f(w)\neq f(v)$, a contradiction.
\end{proof}

The previous example shows that the ability of $\mltup$ to have constants of state type is an expressive enrichment. But a state type $A$, viewed not as a type but as a first-order structure, may have natural options not only for constant symbols $\typd{d}{A}$ but also for constant symbols $c:A^k\rightarrow A$ for $k\geq 1$,\footnote{We use the notation from footnote~\ref{fn:vecsame}.} which are naturally interpreted as functions from the $k$-th Cartesian power of $A$ to $A$. But of course, if $A$ is a state type then $A^k\rightarrow A$ is a type of $\lt$ but not of $\mltup$ (cf. \S\ref{sec:nonmodal} for the relation between $\mltup$ and $\lt$). While $c:A^k\rightarrow A$ may not be a constant of $\mltup$, it can be used to form natural terms of $\lt$ which have type of $\mltup$, and the question is whether these are expressible in $\mltup$. The following proposition shows that they need not be. (Incidentally, this proposition shows the need for the inclusion of the qualification ``\ldots its free variables and constants are those of~$\mltup$'' in the statement of Theorem~\ref{thm:express}).

\begin{prop}[Non-expressibility of the comultiplication map in the empty signature]\label{prop:comonadinexpressible}
Suppose that $A$ is a state type and $B$ a basic entity type of $\mltup$. Suppose that $\typd{c}{A\rightarrow A\rightarrow A}$ is a constant of $\lt$. Let $\typd{P}{(A\rightarrow B)\rightarrow (A\rightarrow A\rightarrow B)}$ be the closed term $\lbd{f}{A\rightarrow B}{\lbd{x}{A}{\lbd{y}{A}{f(cxy)}}}$ of $\lt$. Then there is no term $\typd{Q}{(A\rightarrow B)\rightarrow (A\rightarrow A\rightarrow B)}$ of $\mltup$ in the empty signature such that $\lt\models P=Q$.
\end{prop}

\begin{proof}
Suppose not, and let $\typd{Q}{(A\rightarrow B)\rightarrow (A\rightarrow A\rightarrow B)}$ be such a term of $\mltup$ in the empty signature. Consider the standard model $\mathcal{N}$ of $\lt$ with $\mathcal{N}(C)$ equal to the natural numbers~$\mathbb{N}$ for each atomic type $C$, and with $c_{\mathcal{N}}:\mathbb{N}\rightarrow \mathbb{N} \rightarrow \mathbb{N}$ being given by the curried version of a binary function $\oplus:\mathbb{N}\times \mathbb{N}\rightarrow \mathbb{N}$ subject to the constraint that there are $w,v\geq 0 $ and a permutation $\pi$ of $\mathbb{N}$ such that $\pi^{-1}(w\oplus v)\neq \pi^{-1}(w)\oplus \pi^{-1}(v)$. By restricting to the types and constants of $\mltup$, this model induces a standard model $\mathcal{M}$ of $\mltup$. By our reductio hypothesis, we have $\db{P}_{\mathcal{N},\rho}=\db{Q}_{\mathcal{N},\rho}$ for any variable assignment $\rho$ relative to $\mathcal{N}$. Since $\mathcal{M}$ agrees with $\mathcal{N}$ on terms of type $\mltup$ in the empty signature, we have  $\db{P}_{\mathcal{N},\rho}=\db{Q}_{\mathcal{M},\rho_0}$ for any for any variable assignment $\rho$ relative to $\mathcal{N}$ with $\rho_0$ denoting its restriction to $\mathcal{M}$. The proof then proceeds similarly to the previous proposition, using the aforementioned constraint on $\oplus$.
\begin{proofdetail}

In more detail, let $\pi:\mathbb{N}\rightarrow \mathbb{N}$ be any permutation of the natural numbers. This induces an internal automorphism $\pi$ of $\mathcal{M}$ by setting $\pi_A=\pi$, and by setting $\pi_C$ to be the identity for any other atomic type $C$, and by setting $\pi_{C\rightarrow D}(f) = \pi_{D} \circ f \circ \pi_{C}^{-1}$ for any functional type $C\rightarrow D$ and any $f$ in $\mathcal{M}(C\rightarrow D)$. Note that (\ref{exa:standout:1})-(\ref{exa:standout:2}) from the proof of the previous proposition still hold.

Let $\pi:\mathbb{N}\rightarrow \mathbb{N}$ be any permutation of the natural numbers and which satisfies $w\oplus v\neq \pi(\pi^{-1}(w)\oplus \pi^{-1}(v))$ for some $w,v\geq 0$. As in the previous paragraph this induces an internal automorphism of $\mathcal{M}$.

For any $G$ in $\mathcal{M}((A\rightarrow B)\rightarrow A\rightarrow A\rightarrow B)$ and any $f$ in $\mathcal{M}(A\rightarrow B)$ one has that
\begin{align*}
& (\pi_{(A\rightarrow B)\rightarrow A\rightarrow A\rightarrow B} G) f  = (\pi_{A\rightarrow A\rightarrow B} \circ G)(\pi^{-1}_{A\rightarrow B} f) = (\pi_{A\rightarrow A\rightarrow B} \circ G)(\pi^{-1}_B \circ f \circ \pi_A) \\
& = \pi_{A\rightarrow A\rightarrow B}(G(f\circ \pi_A)) = \pi_{A\rightarrow B} \circ G(f\circ \pi_A) \circ \pi_A^{-1}
\end{align*}
For any $w$ in $\mathcal{M}(A)$ one then has that:
\begin{align*}
&(\pi_{(A\rightarrow B)\rightarrow A\rightarrow A\rightarrow B} G) f w = \pi_{A\rightarrow B}( G(f\circ \pi_A) \pi_A^{-1}(w)) \\
&=\pi_B \circ \pi_{A\rightarrow B}( G(f\circ \pi_A) \pi_A^{-1}(w)) \circ \pi_A^{-1} = \pi_{A\rightarrow B}( G(f\circ \pi_A) \pi_A^{-1}(w)) \circ \pi_A^{-1}
\end{align*}
Then for any $v$ in $\mathcal{M}(A)$ one then has that:
\begin{equation*}
(\pi_{(A\rightarrow B)\rightarrow A\rightarrow A\rightarrow B} G) f w v = G (f\circ \pi_A) \pi_A^{-1}(w) \pi_A^{-1}(v)
\end{equation*}

Apply this to the element $\db{Q}_{\mathcal{M},\rho}$ of $\mathcal{M}((A\rightarrow B)\rightarrow A\rightarrow A\rightarrow B)$ and the identity map $f$ on natural numbers, and let $w,v\geq 0$ be the natural numbers mentioned in the definition of $\pi$. Then we have the following, where $\rho^{\prime}$ (resp. $\rho^{\prime\prime}$) is any variable assignment relative to~$\mathcal{N}$ which agrees with variable assignment $\pi\circ \rho$ (resp. $\rho$) relative to $\mathcal{M}$ on variables of~$\mltup$:
\begin{align*}
& w\oplus v = f(w\oplus v)=\db{P}_{\mathcal{N},\rho^{\prime}} f\,w\,v = \db{Q}_{\mathcal{M},\pi\circ \rho}\, f\,w\,v = \pi_{(A\rightarrow B)\rightarrow A\rightarrow A\rightarrow B} (\db{Q}_{\mathcal{M}, \rho}) \, f\,w\,v 
\\ & =  \db{Q}_{\mathcal{M}, \rho} \,(f\circ \pi_A)\, \pi_A^{-1}(w)\, \pi_A^{-1}(v) = \db{P}_{\mathcal{N}, \rho^{\prime\prime}}\, (f\circ \pi_A)\, \pi_A^{-1}(w)\, \pi_A^{-1}(v) \\
& = \pi(\pi^{-1}(w)\oplus \pi^{-1}(v))
\end{align*}

\end{proofdetail}
\end{proof}

The example in Proposition~\ref{prop:comonadinexpressible} is not \emph{ad hoc}, but is natural in the study of comonads. Monads and its dual comonads are important in contemporary theoretical computer science: whereas monads were deployed by Moggi to model impure effects in a purely functional language,\footnote{\cite{Moggi1989-ph}, \cite{Moggi1991-ej}, \cite{Benton2002-dn}.} comonads have been used, since Brookes and Geva, as a ``semantic model in which sensible comparisons can be made between programs with the same extensional behavior.''\footnote{\cite[2]{Brookes1992-bm}.} 

In Brookes and Geva's original example, the type $A$ is interpreted as the natural numbers with an infinite element on top, $\mathbb{N}\cup \{\omega\}$, and the constant $\typd{c}{A\rightarrow A\rightarrow A}$ is interpreted as minimum; then, on its intended interpretation, the term $P$ from Proposition~\ref{prop:comonadinexpressible} is the comultiplication map which sends $f$ to $\lbd{x}{A}{\lbd{y}{A}{\min(x,y)}}$.\footnote{\cite[8]{Brookes1992-bm}.} This comultiplication map ``shows how a computation may itself be computed'' by breaking $(f(0), f(1), \ldots)$ into its initial segments which repeat the last element.\footnote{\cite[3]{Brookes1992-bm}. Another natural comonad is the stream comonad (cf. \cite[p. 160, Example 5.3.2]{Perrone2019-js}).}

Just as the ability of $\mltup$ to express the actuality operator should be seen as a virtue of the system, so the inability of $\mltup$ to express the comultiplication maps of natural comonads should be seen as a deficit. When $A$ is state type and $B$ is a regular type, one might try to remedy this deficit by admitting terms like  $\lbd{x}{A}{\lbd{y}{A}{f(cxy)}}$ into $\mltup$ while still banning  $\lbd{x}{A}{\lbd{y}{A}{cxy}}$ or its $\eta$-equivalent $c$. But this revision would complicate severely the metatheory developed in this paper since there would now be more complicated terms of state type like $cxy$. Indeed, on this revision, Proposition~\ref{prop:state-free} would no longer hold, which is used extensively throughout the paper, in particular in the proof of Theorem~\ref{prop:combinatoryeffect} of the next section, on which much else in this paper depends. 

%% file: 04-combinatory/00-intro.tex
\section{Combinatory logic and conservation and expressibility}\label{sec:combo}

%% file: 04-combinatory/01-typed-combinators.tex
\subsection{Typed combinator terms and their reductions}\label{sec:typedcomboterms}

The following definition provides a small list of typed combinators terms $\mathsf{X}_{A_1, \ldots, A_n}$ terms, where its type is a function of the types $A_1, \ldots, A_n$. For each combinator term, we give
\begin{itemize}[leftmargin=*]
\item the traditional choice of letter $\mathsf{X}$ along with the Smullyan mnemonic (cf. \cite{Smullyan2000-fo}), 
\item its defining term, 
\item its type built up out of $A_1, \ldots, A_n$,
\item an intuitive gloss, using informal descriptions of input-output behaviour of functions, as well as informal functional notation such as $(x,y)\mapsto x(y)$ and $(x,i)\mapsto x_i$.
\item an identification of the conditions on the types $A_1, \ldots, A_n$, and the variables of these types, required in order for this to be a term of $\mltup$.
\end{itemize}

\begin{defi}[Typed combinator terms of $\mltup$]\label{defn:typedcombo}
An \emph{Identity Bird term} $\ctI{A}{}$ \emph{of} $\mltup$ is a term of 
of the following form and type:
\begin{equation*}
\lbd{x}{A}{x} \; : \; A\rightarrow A
\end{equation*}
It is required that $A$ has regular type. Intuitively $\ctI{A}{}$ is the identity function on type $A$.

A \emph{Kestrel term} $\ctK{A}{B}{}{}$ \emph{of} $\mltup$  is a term of the following form and type:
\begin{equation*}
\lbd{x}{A}{\lbd{y}{B}{x}} \; : \; A  \rightarrow B\rightarrow   A
\end{equation*}
It is required that $A$ has regular type, and that $\typd{x}{A}, \typd{y}{B}$ are distinct variables. Intuitively given a value in $A$, Kestrel $\ctK{A}{B}{}{}$ returns the constant function from $B$ to $A$ with that value.

A \emph{Cardinal term} $\ctC{A}{B}{C}{}{}{}$ \emph{of} $\mltup$ is a term of the following form and type:
\begin{equation*}
\lbd{x}{A \rightarrow B\rightarrow C}{\lbd{y}{B}{\lbd{z}{A}{xzy}}} \;:\; (A\rightarrow B\rightarrow   C) \rightarrow B\rightarrow  A\rightarrow  C
\end{equation*}
It is required that $C$ has regular type, and that $\typd{y}{B}, \typd{z}{A}$ are distinct variables. Intuitively Cardinal $\ctC{A}{B}{C}{}{}{}$ takes a function $x$ of two arguments and returns the function of two arguments which permutes the two inputs. I.e. it maps function $x$ to the function $(y,z)\mapsto x(z,y)$.

A \emph{Dardinal term} $\ctD{c}{A}{B}{C}{}{}$ \emph{of} $\mltup$  is a term of the following form and type:
\begin{equation*}
\lbd{x}{A \rightarrow B\rightarrow C}{\lbd{z}{A}{xzc}} \; : \; (A\rightarrow B\rightarrow   C) \rightarrow  A\rightarrow  C
\end{equation*}
It is required that $C$ has regular type and that $B$ has state type and that $\typd{c}{B}$ is a constant.  \emph{Dardinal} is short for \emph{decorated Cardinal}. The number of Dardinals varies with the signature, and intuitively $\ctD{c}{A}{B}{C}{}{}$ takes a function $x$ of two arguments and returns the function of one argument which slots this value into the first spot and $c$ into the second spot. I.e. it maps function $x$ to the function $z\mapsto x(z,c)$.

A \emph{Starling term} $\ctS{A}{B}{C}{}{}{}$ \emph{of} $\mltup$  is a term of the following form and type:
\begin{equation*}
\lbd{x}{C\rightarrow  A\rightarrow  B}{\lbd{y}{C\rightarrow   A}{\lbd{z}{C}{xz(yz)}}} \;:\; (C\rightarrow   A\rightarrow  B) \rightarrow (C\rightarrow A)\rightarrow C\rightarrow B
\end{equation*}
It is required that $A,B$ are regular types. Like Cardinal, Starling $\ctS{A}{B}{C}{}{}{}$ permutes some of the order of the inputs, but it also is a basic example of a combinator which duplicates an input. As for its intended behaviour, as an argument of $x,y$, it is just an ``indexed''  version of functional application $z\mapsto x_z(y_z)$.

A \emph{Warbler term} $\ctW{A}{B}{}{} $ \emph{of} $\mltup$  is a term of the following form and type:
\begin{equation*}
\lbd{x}{A \rightarrow A \rightarrow B}{\lbd{y}{A}{xyy}} \; : \; (A\rightarrow  A\rightarrow  B) \rightarrow A\rightarrow  B
\end{equation*}
It is required that $B$ is of regular type. Intuitively, Warbler $\ctW{A}{B}{}{} $ takes a curried function $x$ defined on $A\times A$ and returns the diagonal function on $A$ given by $y\mapsto x(y,y)$.

A \emph{Bluebird term} $\ctB{A}{B}{C}{}{}{}$ \emph{of} $\mltup$  is a term of the following form and type: 
\begin{equation*}
\lbd{x}{B\rightarrow C}{\lbd{y}{A\rightarrow B}{\lbd{z}{A}{x(yz)}}} : (B\rightarrow C) \rightarrow (A \rightarrow B) \rightarrow A\rightarrow  C
\end{equation*}
It is required that $B,C$ are regular types; and if $A$ is regular then it is required that $\typd{x}{B\rightarrow C}, \typd{z}{A}$ are distinct variables and that $\typd{x}{B\rightarrow C}, \typd{y}{A\rightarrow B}$ are distinct variables. The Bluebird $\ctB{A}{B}{C}{}{}{}$ returns the composition of two functions $x,y$ whose domains and codomains match appropriately. 
 
\end{defi}

Each of these terms is closed. Further, each of these terms has pairwise distinct bound variables: this follows from the stipulated distinctness in the above definition together with the distinctness of type due to some types being functions of others. This is important to take note of because distinctness of variables is a part of $\beta$-reduction (cf. Definition~\ref{defn:betareduction}(\ref{defn:betareduction:2})). However, we note the following:

\begin{rem}[Limited availability of Cardinal]\label{rmk:oncardinalandstate}
The only combinator in Definition~\ref{defn:typedcombo} which can have more than one bound variable of state type is Cardinal. This happens when both $A,B$ in Cardinal are of state type. When both $A,B$ are identical state types, the requirement that  $\typd{y}{B}, \typd{z}{A}$ are distinct variables implies that the associated Cardinal term is simply not available in $\mlt{1}$.

Later in \S\ref{sec:purecombinatory}, when we develop pure intensional combinatory logic, we will have to pay attention to this limited availability of Cardinal. For this purpose, we note now that in Cardinal, the stated requirement that $\typd{y}{B}, \typd{z}{A}$ are distinct variables implies: either $A,B$ are distinct types, or that $A,B$ are identical types and $\param(A)=\param(B)>1$.
\end{rem}

While obvious from definition, we note the following about the other combinators:
\begin{rem}[Availability of all typed combinator terms besides Cardinal]\label{rmk:onalphandcombinators}
Besides Cardinal, all typed combinator terms are $\alpha$-equivalent to a term of $\mlt{1}$. This is because, by inspection, each contains at most one bound variable of a given state type.
\end{rem}

The following theorem and subsequent two propositions say that the combinator terms have their expected behavior in $\mltup$. The proof of these results also contain many useful examples of $\beta$-reductions of distance 1 and 2. For later purposes, we take note that all of these reductions are indeed regular $\beta$-reductions (cf. Definition~\ref{defn:betareduction}~(\ref{defn:betareduction:5})).\footnote{While we are not pursuing weakness until \S\ref{sec:purecombinatory}, it is worth noting that the reductions in this theorem and the subsequent two propositions are weak.}

\begin{thm}[Combinatory behaviour in $\mltup$]\label{prop:combinatoryeffect}
Suppose that the combinators on the below left are terms of $\mltup$. For each item, suppose that the terms $P,Q,R$ are terms of $\mltup$ of the appropriate type to make the below applications well-formed. Then one has the $\beta$-reductions to the terms on the below right, and indeed all of these reductions are regular:
\begin{align*}
& \ctI{A}{P}\transreduces{\beta}{\param} P  \hspace{7mm} \ctK{A}{B}{P}{Q}\transreduces{\beta}{\param} P & \hspace{1mm}\ctC{A}{B}{C}{P}{Q}{R}\transreduces{\beta}{\param} PRQ & \hspace{5mm}\ctD{c}{A}{B}{C}{P}{R}\transreduces{\beta}{\param} PRc\\
& \hspace{3mm}\ctS{A}{B}{C}{P}{Q}{R}\transreduces{\beta}{\param} PR(QR) & \hspace{5mm} \ctW{A}{B}{P}{Q}\transreduces{\beta}{\param} PQQ  &\hspace{5mm}  \ctB{A}{B}{C}{P}{Q}{R}\transreduces{\beta}{\param} P(QR) 
\end{align*}
\end{thm}
\begin{proof}
For  Identity Bird $\ctI{A}{}$, we have $(\lbd{x}{A}{x})P\reduces{\beta}{\param} P$ since the term $\typd{x}{A}$ has no lambda abstracts.

Now we turn to the remaining cases. In this part of the proof:
\begin{itemize}[leftmargin=*]
\item We restrict attention to the case where at least one of the constituent types $A,B,C$ (or just $A,B$ in the case of Kestrel and Warbler) is a state type. For, if all of the constituent types are regular, then by inspection of the definitions, all of the bound variables in the combinator terms are regular and we can change them by $\alpha$-conversion to avoid variable capture. 
\item We repeatedly use Proposition~\ref{prop:state-free}, which says that terms of state type in $\mltup$ are either constants or variables of that very state type. 
\item We often ensure that the ``free for'' condition of $\beta$-reduction (Definition~\ref{defn:betareduction}(\ref{defn:betareduction:1})) is met by doing $\alpha$-conversion on the bound variables of the combinatory terms to ensure that they are disjoint from the free variables of the inputs. We refer to this simply as ``disjointness.''
\item As we proceed, we note that the reductions meet the regularity condition Definition~\ref{defn:betareduction}(\ref{defn:betareduction:5}), but do not say more than this since it follows clearly from the displayed instances and the case assumptions. Further, so as not to clutter the proof, we just mark regularity in the text and write the simpler $\beta$ instead of $\beta_r$. Finally, since all $\beta$-reductions of distance $\leq 1$ are trivially regular, we only need to explicitly take note of regularity for instances of distance $\geq 2$ (as a matter of fact, all $\beta$-reductions in this proof have distance $\leq 2$).
\end{itemize}

For Kestrel $\ctK{A}{B}{}{}$, suppose $A$ is a regular type, $B$ is a state type, and $\typd{P}{A}$ and $\typd{Q}{B}$ are terms of $\mltup$. We can use $\alpha$-conversion on the bound variable of type $A$ in Kestrel so that it does not appear free in $Q$. Further, $\typd{Q}{B}$ is free for Kestrel's second variable $\typd{y}{B}$ in the term $\typd{x}{A}$, since the latter has no lambda abstracts. These two points get us the following, where the first is a $\beta$-reduction of distance~$1$:
\begin{equation}
\big(\lbd{x}{A}{\lbd{y}{B}{x}}\big) P Q \reduces{\beta}{\param} (\lbd{x}{A}{x[y:=Q]})P \equiv (\lbd{x}{A}{x})P \reduces{\beta}{\param} P
\end{equation}
The second $\beta$-reduction follows just as in Identity Bird.

For Cardinal $\ctC{A}{B}{C}{}{}{}$, suppose that $C$ is a regular type and suppose we have the terms $\typd{P}{A\rightarrow B\rightarrow C}$ and $\typd{Q}{B}$ and $\typd{R}{A}$ of $\mltup$. Suppose the corresponding three bound variables of Cardinal $\ctC{A}{B}{C}{}{}{}$  are $\typd{x}{A\rightarrow B\rightarrow C}$ and $\typd{y}{B}$ and $\typd{z}{A}$. We respectively refer to these in the following discussion as the \emph{first}, \emph{second}, and \emph{third bound variables} of Cardinal (and we adopt similar conventions for the subsequent combinator terms). 

There are three cases to consider. 

If $A$ is a regular type and $B$ is a state type, then since the first and third bound variables are of regular type, we may change them by $\alpha$-conversion so that they do not appear free in $P$ or $R$; and they do not appear free in $\typd{Q}{B}$ since $B$ is of state type. Then $\typd{Q}{B}$ is free for $\typd{y}{B}$ in $\lbd{z}{A}{xzy}$. Further, the first bound variable does not appear free in $\typd{Q}{B}$ since $B$ is of state type. This gives us the first step in the following, which is a $\beta$-reduction of distance~$1$:
\begin{align}
  (\lbd{x}{A\rightarrow B\rightarrow C}{\lbd{y}{B}{\lbd{z}{A}{xzy}}})PQR  &    \reduces{\beta}{\param} \;  (\lbd{x}{A\rightarrow B\rightarrow C}{\lbd{z}{A}{xzQ}})PR \notag \\
  \reduces{\beta}{\param} \;   (\lbd{z}{A}{PzQ})R  & \reduces{\beta}{\param} \; PRQ 
\end{align}
The second $\beta$-reduction follows by disjointness: the third bound variable does not appear free in $P$ by the previous $\alpha$-conversion; and $\typd{Q}{B}$ does not contain any bound variables since $B$ is of state type . The third $\beta$-reduction follows since the displayed free occurrence of $\typd{z}{A}$ is the only free occurrence in $PzQ$, since by previous $\alpha$-conversion it does not appear free in $P$, and since $\typd{Q}{B}$ is of state type $B$.

If $A$ is a state type and $B$ is a regular type, then since the first and the second bound variables are of regular type, we may change them by $\alpha$-conversion so that they do not appear free in $P$ or $Q$; and they do not occur free in $\typd{R}{A}$ since $A$ is of state type. Since $xzy$ contains no lambda abstracts, one has that $\typd{R}{A}$ is free for $\typd{z}{A}$ in $xzy$. This gives us the first step in the following, which is a regular $\beta$-reduction of distance~$2$:
\begin{align}
 (\lbd{x}{A\rightarrow B\rightarrow C}{\lbd{y}{B}{\lbd{z}{A}{xzy}}}) PQR & \reduces{\beta}{\param}\; (\lbd{x}{A\rightarrow B\rightarrow C}{\lbd{y}{B}{xRy}}) PQ \notag \\
\reduces{\beta}{\param}\; (\lbd{y}{B}{PRy}) Q & \reduces{\beta}{\param}\; PRQ
\end{align}
The second $\beta$-reduction follows by disjointness: the second bound variable does not appear free in $P$ by the previous $\alpha$-conversion; and $\typd{R}{A}$ does not contain any bound variables since $A$ is of state type. The third $\beta$-reduction follows since the displayed free occurrence of $\typd{y}{B}$ is the only free occurrence in $PRy$, since by previous $\alpha$-conversion it does not appear free in $P$, and since $\typd{R}{A}$ is of state type.

If $A,B$ are both state types, then since the first bound variable is of regular type it does not appear free in $\typd{Q}{B}$ or $\typd{R}{A}$ since these are of state type. Since the last two bound variables $\typd{y}{B}, \typd{z}{A}$ of Cardinal are distinct by definition (cf. Definition~\ref{defn:typedcombo}) by $\alpha$-conversion we can assume that if $\typd{Q}{B}$ is a variable then it is the second bound variable $\typd{y}{B}$. This implies that $\typd{Q}{B}$ is free for $\typd{y}{B}$ in $\lbd{z}{A}{xzy}$. Then we have the following, where the first is a $\beta$-reduction of distance 1:
\begin{align}
(\lbd{x}{A\rightarrow B\rightarrow C}{\lbd{y}{B}{\lbd{z}{A}{xzy}}}) PQR  & \reduces{\beta}{\param}\; (\lbd{x}{A\rightarrow B\rightarrow C}{\lbd{z}{A}{xzQ}}) PR \notag \\
 \reduces{\beta}{\param}\; (\lbd{x}{A\rightarrow B\rightarrow C}{xRQ}) P & \reduces{\beta}{\param}\; PRQ
\end{align}
The second $\beta$-reduction is of distance~1 and follows because of the fact that $\typd{R}{A}$ is free for $\typd{z}{A}$ in $xzQ$ since the term $xzQ$ has no lambda abstracts in it since $\typd{Q}{A}$ is of state type. The third $\beta$-reduction follows since the displayed free occurrence of $\typd{x}{A\rightarrow B\rightarrow C}$ is the only free occurrence in $xRQ$, due to $\typd{Q}{B}$ and $\typd{R}{A}$ being of state type.

For Dardinal $\ctD{c}{A}{B}{C}{}{}$, suppose that $C$ is a regular type and $A,B$ are state types with $\typd{c}{B}$ a constant. Suppose we have the terms $\typd{P}{A\rightarrow B\rightarrow C}$ and $\typd{R}{A}$. Then $\typd{R}{A}$ is free for $\typd{z}{A}$ in $xzc$ since this term has no lambda abstracts. And the first bound variable $\typd{x}{A\rightarrow B\rightarrow C}$ does not appear free in $\typd{R}{A}$ since $A$ is of state type. Then we have the following, where the first $\beta$-reduction is of distance~1:
\begin{equation}
\big( \lbd{x}{A\rightarrow B\rightarrow C}{\lbd{z}{A}{xzc}}\big) P R \reduces{\beta}{\param} (\lbd{x}{A\rightarrow B\rightarrow C}{ x R c})P 
 \reduces{\beta}{\param} PRc
\end{equation}
The second $\beta$-reduction happens since the displayed free occurrence of $\typd{x}{A\rightarrow B\rightarrow C}$ in $xRc$ is the only occurrence since $\typd{R}{A}$ is of state type.

For Starling $\ctS{A}{B}{C}{}{}{}$, suppose that $A,B$ are regular types and $C$ is a state type. Starling's first two bound variables are of regular type, and by $\alpha$-conversion we may assume that they do not appear free in $P,Q$; and they do not appear free in $\typd{R}{C}$ since it is of state type. Since the term $xz(yz)$ contains no lambda abstracts, one has that $\typd{R}{C}$ is free for $\typd{z}{C}$ in $xz(yz)$. Then we have the following, where the first application of $\beta$ is regular of distance~2:
\begin{align}
& (\lbd{x}{C\rightarrow A\rightarrow B}{\lbd{y}{C\rightarrow A}{\lbd{z}{C}{xz(yz)}}})PQR \notag \\ & \reduces{\beta}{\param} \; (\lbd{x}{C\rightarrow A\rightarrow B}{\lbd{y}{C\rightarrow A}{xR(yR)}})PQ \notag \\
& \reduces{\beta}{\param} \; (\lbd{y}{C\rightarrow A}{PR(yR)})Q \reduces{\beta}{\param} \; PR(QR) 
\end{align}
The second $\beta$-reduction follows by disjointness: the second bound variable does not appear free in $P$ by the previous $\alpha$-conversion; and $\typd{R}{C}$ does not contain any bound variables since $C$ is of state type. The third $\beta$-reduction follows since the displayed free occurrence of $\typd{y}{C\rightarrow A}$ is the only free occurrence in $PR(yR)$, since by previous $\alpha$-conversion it does not appear free in $P$, and it does not appear free in $\typd{R}{C}$ since this is of state type $C$.

For Warbler $\ctW{A}{B}{}{}$, suppose $B$ is a regular type, and $A$ is a state type, and $\typd{P}{A\rightarrow A\rightarrow B}$ and $\typd{Q}{A}$ are terms. Since Warbler's first bound variable is of regular type, by $\alpha$-conversion we may assume that it does not appear free in $P$; and it does not appear free in $\typd{Q}{A}$ since this is of state type. Since the term $xyy$ does not contain any lambda abstracts, one has that $\typd{Q}{A}$ is free for $\typd{y}{A}$ in $xyy$. Then we have the following, where the first instance of $\beta$-reduction is of distance~$1$:
\begin{equation}
(\lbd{x}{A\rightarrow A\rightarrow B}{\lbd{y}{A}{xyy}})PQ  \reduces{\beta}{\param} \; (\lbd{x}{A\rightarrow A\rightarrow B}{xQQ})P\reduces{\beta}{\param} \;  PQQ
\end{equation}
The last $\beta$-reduction follows since the displayed instance of $\typd{x}{A\rightarrow A\rightarrow B}$ is the only free instance in $xQQ$ since $\typd{Q}{A}$ is of state type. 

For Bluebird $\ctB{A}{B}{C}{}{}{}$, suppose that $B,C$ are regular types, and $A$ is a state type. Since Bluebird's first two bound variables are of regular type, by $\alpha$-conversion we may assume that they do not appear free in $P,Q$; and they do not appear free in $\typd{R}{A}$ since it is of state type. Further, $\typd{R}{A}$ is free for $\typd{z}{A}$ in $x(yz)$ since the term $x(yz)$ contains no lambda abstracts. Then we have the following, where the first instance of $\beta$ is regular of distance~2:
\begin{align}
 (\lbd{x}{B\rightarrow C}{\lbd{y}{A\rightarrow B}{\lbd{z}{A}{x(yz)}}})PQR & \reduces{\beta}{\param} \; (\lbd{x}{B\rightarrow C}{\lbd{y}{A\rightarrow B}{x(yR)}})PQ \notag \\
\reduces{\beta}{\param} \; (\lbd{y}{A\rightarrow B}{P(yR)})Q & \reduces{\beta}{\param} \;  P(QR)
\end{align}
The second $\beta$-reduction follows by disjointness: the second bound variable does not appear free in $P$ by the previous $\alpha$-conversion; and $\typd{R}{A}$ does not contain any bound variables since $A$ is of state type. The third $\beta$-reduction follows since the displayed free occurrence of $\typd{y}{B}$ is the only free occurrence in $P(yR)$, since by previous $\alpha$-conversion it does not appear free in $P$, and it does not appear free in $\typd{R}{A}$ since this is of state type $A$.
\end{proof}

The following proposition is more elementary:
\begin{prop}\label{prop:combinatoryeffecteasy}
Suppose that the combinators on the below left are terms of $\mltup$. For each item, suppose that the terms $P,Q$ are terms of $\mltup$ of the appropriate regular type to make the below applications well-formed. Further, suppose that the \emph{only} free variables of $P,Q$ are themselves of regular type. Then one has the $\beta$ reductions to the terms on the below right, and indeed these are $\beta_0$-reductions:
\begin{align*}
& \ctK{A}{B}{P}\reduces{\beta}{\param} \lbd{y}{B}{P} & \hspace{-10mm} \ctC{A}{B}{C}{P}{}{}\reduces{\beta}{\param} \lbd{y}{B}{\lbd{z}{A}{Pzy}}  & \hspace{5mm} \ctD{c}{A}{B}{C}{P}{}\reduces{\beta}{\param} \lbd{z}{A}{Pzc}  \\ 
& \ctS{A}{B}{C}{P}{Q}{}\transreduces{\beta}{\param} \lbd{z}{C}{Pz(Qz)} & \ctW{A}{B}{P}{}\reduces{\beta}{\param} \lbd{y}{A}{Pyy} & \hspace{5mm} \ctB{A}{B}{C}{P}{Q}{}\transreduces{\beta}{\param} \lbd{z}{A}{P(Qz)} 
\end{align*}
\end{prop}
\begin{proof}
We give the argument for Cardinal, since the argument for the other combinators is similar. For the second and third bound variables of Cardinal, if they are of regular type then we may use $\alpha$-conversion to convert them to variables which do not occur free in $P$; while if they are of state type then by hypothesis they do not occur free in $P$. Hence, after this $\alpha$-conversion, the second and third bound variables of the Cardinal term do not occur free in $P$, and so we can use disjointness to $\beta$-reduce and indeed $\beta_0$-reduce.
\end{proof}

Lastly, for later (cf. Remark~\ref{rmk:cardinal2dardingal}), we need to take note of the following reduction of Cardinal to Dardinal:
\begin{prop}\label{prop:thefirstcardinaltodardinal}
Suppose that the below displayed Cardinal is a term of $\mltup$. Suppose that the below terms $P,c$ are terms of $\mltup$ of the appropriate type to make the applications well-formed. Then one has the following regular reduction:
\begin{equation*}
\ctC{A}{B}{C}{P}{c}{}\reduces{\beta}{\param}  \ctD{c}{A}{B}{C}{P}
\end{equation*}
\end{prop}
\begin{proof}
We simply use a $\beta$-reduction of distance 1:
\begin{equation*}
\big(\lbd{x}{A\rightarrow B\rightarrow C}{\lbd{y}{B}{\lbd{z}{A}{xzy}}}\big) Pc \reduces{\beta}{\param} (\lbd{x}{A\rightarrow B\rightarrow C}{\lbd{z}{A}{xzc}})P
\end{equation*}
For, one has that $\typd{c}{B}$ is free for $\typd{y}{B}$ in $\lbd{z}{A}{xzy}$ since $\typd{c}{B}$ is closed. And no variables appear free in a constant.
\end{proof}

%% file: 04-combinatory/02-BCKW.tex
\subsection{The BCDKW-combinatorial terms}\label{sec:combo:bcdkw}

\begin{defi}
Let $\mathcal{X}$ a collection of typed combinators in a signature. 

The \emph{$\mathcal{X}$-combinatorial terms of} $\mltup$ in that signature are the smallest collection of terms in $\mltup$ which is closed under application and which contains the constants of the signature, the variables specified by the parameter $\param$, and all instances of combinators $\mathsf{X}_{A_1, \ldots, A_n}$ in $\mathcal{X}$ which are terms of $\mltup$.

The \emph{expanded $\mathcal{X}$-combinatorial terms of} $\mltup$ in that signature is the collection of terms $N$ of $\mltup$ such that there is a $\mathcal{X}$-combinatorial term $M$ of $\mltup$ with the same free variables as $N$ satisfying $M\transreduces{\beta}{\param} N$.
\end{defi}

We will be mostly concerned in what follows with $\mathcal{X}$ being $\{\mathsf{B}, \mathsf{C}, \mathsf{D}, \mathsf{K}, \mathsf{W}\}$, which we abbreviate as $\mathsf{BCDKW}$. In the untyped setting, one can take $\mathsf{BCKW}$ or $\mathsf{SK}$ as primitive (cf. \cite[Lemma 1.3.9 p. 17]{Bimbo2011-me}), and indeed historically Sch\"onfinkel did the latter and Curry initially did the former (\cite[\S{2.2}]{Seldin2009-kc}). We opt to use a $\mathsf{BCKW}$ basis rather than a $\mathsf{SK}$ basis, although see Question~\ref{qu:whatabouts} and Remark~\ref{rmk:whatabouts} below.

We first note that Starling and identity are $\mathsf{BCDKW}$ combinatorial.

\begin{prop}[Recovery of Starling]\label{prop:wegotstarlingback}
Suppose that $A,B,C$ are types and $A,B$ are regular types. Then $\ctS{A}{B}{C}{}{}{}$ is an expanded $\mathsf{BCDKW}$-combinatorial term of $\mltup$.
\end{prop}
\begin{proof}
One can check that:
\begin{equation}\label{eqn:wegotstarlingback}
\ctB{A_1}{B_1}{C_1}{}{}{} (\ctB{A_2}{B_2}{C_2}{}{}{} (\ctB{A_3}{B_3}{C_3}{}{}{} \ctW{A_4}{B_4}{}{}) \ctC{A_5}{B_5}{C_5}{}{}{}) (\ctB{A_6}{B_6}{C_6}{}{}{} \ctB{A_7}{B_7}{C_7}{}{}{}) \transreduces{\beta}{\param} \ctS{A}{B}{C}{}{}{}
\end{equation}
where the types of $A_i, B_i, C_i$ are calculated in terms of $A,B,C$ as follows:
\begin{center}
\begin{tabular}{|l|l|l|l|}
\hline
$i$ & $A_i$ & $B_i$ & $C_i$  \\ \hline
 $1$ & $C\rightarrow A\rightarrow B$ & $C\rightarrow (C\rightarrow A)\rightarrow C\rightarrow B$ & $(C\rightarrow A)\rightarrow C\rightarrow B$  \\ \hline
 $2$ & $C\rightarrow (C\rightarrow A)\rightarrow C\rightarrow B$ & $(C\rightarrow A)\rightarrow C\rightarrow C\rightarrow B$ &  $(C\rightarrow A)\rightarrow C\rightarrow B$  \\ \hline
 $3$ & $C\rightarrow A$ & $C\rightarrow C\rightarrow B$ & $C\rightarrow B$  \\ \hline
 $4$ & $C$ & $B$ & n/a  \\ \hline
 $5$ & $C$ & $C\rightarrow A$ & $C\rightarrow B$  \\ \hline
 $6$ & $C$ & $A\rightarrow B$ & $(C\rightarrow A)\rightarrow C\rightarrow B$  \\ \hline
 $7$ & $C$ & $A$ & $B$  \\ \hline
\end{tabular}
\end{center}
After checking that these types are in $\mltup$ and that the applications in (\ref{eqn:wegotstarlingback}) are well-formed, the proof of (\ref{eqn:wegotstarlingback}) follows the usual untyped reduction (cf. \cite[p. 155]{Curry1958-yi}), using Theorem~\ref{prop:combinatoryeffect} and Proposition~\ref{prop:combinatoryeffecteasy}.
\begin{proofdetail}
In more detail, note that since $A,B,C$ are types and $A,B$ are regular types, we have that:
\begin{itemize}[leftmargin=*]
    \item $\ctB{A_i}{B_i}{C_i}{}{}{}$ for $i\in \{1,2,3,6,7\}$ is a term of $\mltup$ because $B_i,C_i$ are regular types by inspection of the table.
    \item $\ctW{A_4}{B_4}{}{}$ is a term of $\mltup$ because $B_4$ is a regular type by inspection of the table.
    \item $\ctC{A_5}{B_5}{C_5}{}{}{}$ is a term of $\mltup$ because $C_5$ is a regular type and because $A_5, B_5$ are either distinct types or are both regular types. In particular, when $C$ is a state type $A_5$ is a state type and $B_5$ is a regular type, while when $C$ is a regular type both $A_5, B_5$ are regular types.
\end{itemize}

Before proving the claim, we first verify that 
\begin{equation}\label{eqn:s:calc:whatis}
\ctB{A_1}{B_1}{C_1}{}{}{} (\ctB{A_2}{B_2}{C_2}{}{}{} (\ctB{A_3}{B_3}{C_3}{}{}{} \ctW{A_4}{B_4}{}{}) \ctC{A_5}{B_5}{C_5}{}{}{}) (\ctB{A_6}{B_6}{C_6}{}{}{} \ctB{A_7}{B_7}{C_7}{}{}{})
\end{equation}
is a term of $\mltup$ of the same type as $\ctS{A}{B}{C}{}{}{}$.
\begin{itemize}[leftmargin=*]
\item The Bluebird $\ctB{A_6}{B_6}{C_6}{}{}{}$ takes inputs of type $B_6\rightarrow C_6$, which by inspection of the table is
\begin{equation}\label{eqn:s:masdfasd}
(A\rightarrow B)\rightarrow (C\rightarrow A)\rightarrow C\rightarrow B 
\end{equation}
The Bluebird $\ctB{A_7}{B_7}{C_7}{}{}{}$ has type $(B_7\rightarrow C_7)\rightarrow (A_7\rightarrow B_7)\rightarrow A_7\rightarrow C_7$, which by inspection of the table is 
\begin{equation*}
(A\rightarrow B)\rightarrow (C\rightarrow A)\rightarrow C\rightarrow B
\end{equation*}
Since this agrees with (\ref{eqn:s:masdfasd}), we have that the application $\ctB{A_6}{B_6}{C_6}{}{}{} \ctB{A_7}{B_7}{C_7}{}{}{}$ is a term of $\mltup$. Its type is the output type of the first Bluebird  $\ctB{A_6}{B_6}{C_6}{}{}{}$, which is $(A_6\rightarrow B_6)\rightarrow A_6\rightarrow C_6$, which by inspection of the table is
\begin{equation}\label{eqn:s:calc:calc:calc}
(C\rightarrow A\rightarrow B)\rightarrow C\rightarrow (C\rightarrow A)\rightarrow C\rightarrow B
\end{equation}
\item The Bluebird $\ctB{A_3}{B_3}{C_3}{}{}{}$ has input type $B_3\rightarrow C_3$, which by inspection of the table is 
\begin{equation}\label{eqn:s:calc:sodifficult}
\big(C\rightarrow C\rightarrow B\big)\rightarrow C\rightarrow B
\end{equation}
The Warbler $\ctW{A_4}{B_4}{}{}$ has type $(A_4\rightarrow A_4\rightarrow B_4)\rightarrow A_4\rightarrow B_4$, which by inspection of the table is 
\begin{equation*}
(C\rightarrow C\rightarrow B)\rightarrow C\rightarrow B
\end{equation*}
Since this agrees with the earlier result in (\ref{eqn:s:calc:sodifficult}), the application $\ctB{A_3}{B_3}{C_3}{}{}{} \ctW{A_4}{B_4}{}{}$ is a term of $\mltup$. Its type is the output type of $\ctB{A_3}{B_3}{C_3}{}{}{}$, which is $(A_3\rightarrow B_3)\rightarrow A_3\rightarrow C_3$. By inspection of the table, this is
\begin{equation}\label{eqn:s:calc:bw}
((C\rightarrow A)\rightarrow C\rightarrow C\rightarrow B)\rightarrow (C\rightarrow A)\rightarrow C\rightarrow B
\end{equation}
\item The Bluebird $\ctB{A_2}{B_2}{C_2}{}{}{}$ has input type $B_2\rightarrow C_2$, which by inspection of the table is 
\begin{equation*}
\big((C\rightarrow A)\rightarrow C\rightarrow C\rightarrow B \big)\rightarrow (C\rightarrow A)\rightarrow C\rightarrow B
\end{equation*}
Since this agrees with (\ref{eqn:s:calc:bw}), we have that the application $\ctB{A_2}{B_2}{C_2}{}{}{} (\ctB{A_3}{B_3}{C_3}{}{}{} \ctW{A_4}{B_4}{}{})$ is a term of $\mltup$. Its type is the output type of $\ctB{A_2}{B_2}{C_2}{}{}{}$, which is $(A_2\rightarrow B_2)\rightarrow A_2\rightarrow C_2$.
\item The input type of $\ctB{A_2}{B_2}{C_2}{}{}{} (\ctB{A_3}{B_3}{C_3}{}{}{} \ctW{A_4}{B_4}{}{})$ is then $A_2\rightarrow B_2$, which by inspection of the table is
\begin{equation}\label{eqn:s:inputb(bw)}
\big( C\rightarrow (C\rightarrow A)\rightarrow C\rightarrow B\big) \rightarrow (C\rightarrow A)\rightarrow C\rightarrow C\rightarrow B
\end{equation}
\item The type of $\ctC{A_5}{B_5}{C_5}{}{}{}$ is $(A_5\rightarrow B_5\rightarrow C_5)\rightarrow B_5\rightarrow A_5\rightarrow C_5$, which by inspection of the table is
\begin{equation*}
(C\rightarrow (C\rightarrow A)\rightarrow C\rightarrow B)\rightarrow (C\rightarrow A)\rightarrow C\rightarrow C\rightarrow B
\end{equation*}
Since this agrees with~(\ref{eqn:s:inputb(bw)}), the application $\ctB{A_2}{B_2}{C_2}{}{}{} (\ctB{A_3}{B_3}{C_3}{}{}{} \ctW{A_4}{B_4}{}{}) \ctC{A_5}{B_5}{C_5}{}{}{}$ is a term of $\mltup$. Its type the output type of $\ctB{A_2}{B_2}{C_2}{}{}{} (\ctB{A_3}{B_3}{C_3}{}{}{} \ctW{A_4}{B_4}{}{})$, which is $A_2\rightarrow C_2$. By inspection of the table this is
\begin{equation}\label{eqn:s:manymany}
\big(C\rightarrow (C\rightarrow A)\rightarrow C\rightarrow B \big)\rightarrow (C\rightarrow A)\rightarrow C\rightarrow B
\end{equation}
\item The input type of $\ctB{A_1}{B_1}{C_1}{}{}{}$ is $B_1\rightarrow C_1$, which by inspection of the table is:
\begin{equation*}
\big(C\rightarrow (C\rightarrow A)\rightarrow C\rightarrow B \big)\rightarrow (C\rightarrow A)\rightarrow C\rightarrow B
\end{equation*}
Since this agrees with (\ref{eqn:s:manymany}), we have that the application
\begin{equation*}
\ctB{A_1}{B_1}{C_1}{}{}{} \big(\ctB{A_2}{B_2}{C_2}{}{}{} (\ctB{A_3}{B_3}{C_3}{}{}{} \ctW{A_4}{B_4}{}{}) \ctC{A_5}{B_5}{C_5}{}{}{}\big)
\end{equation*}
is a term of $\mltup$. Its type is the output type of $\ctB{A_1}{B_1}{C_1}{}{}{}$, and hence we have
\begin{equation}\label{eqn:s:calc:whatispre:type}
\ctB{A_1}{B_1}{C_1}{}{}{} \big(\ctB{A_2}{B_2}{C_2}{}{}{} (\ctB{A_3}{B_3}{C_3}{}{}{} \ctW{A_4}{B_4}{}{}) \ctC{A_5}{B_5}{C_5}{}{}{}\big) : (A_1\rightarrow B_1)\rightarrow A_1\rightarrow C_1
\end{equation}
Hence its \emph{input type} is $A_1\rightarrow B_1$, which by inspection of the table is:
\begin{equation*}
\big( C\rightarrow A\rightarrow B\big) \rightarrow C\rightarrow (C\rightarrow A)\rightarrow C\rightarrow B
\end{equation*}
Since this agrees with (\ref{eqn:s:calc:calc:calc}), we have that the application in (\ref{eqn:s:calc:whatis}) is a term of $\mltup$. From (\ref{eqn:s:calc:whatispre:type}) we have that the term in  (\ref{eqn:s:calc:whatis}) has type $A_1\rightarrow C_1$, which by inspection of the table is:
\begin{equation*}
\big( C\rightarrow A\rightarrow B\big) \rightarrow (C\rightarrow A)\rightarrow C\rightarrow B
\end{equation*}
which is exactly the type of $\ctS{A}{B}{C}{}{}{}$.

\end{itemize}

Now we turn to showing the claim. We start by applying Proposition~\ref{prop:combinatoryeffecteasy}, which we can do since the terms are closed:
\begin{align}
& \notag \ctB{A_1}{B_1}{C_1}{}{}{} \big(\ctB{A_2}{B_2}{C_2}{}{}{} (\ctB{A_3}{B_3}{C_3}{}{}{} \ctW{A_4}{B_4}{}{}) \ctC{A_5}{B_5}{C_5}{}{}{}\big) \big(\ctB{A_6}{B_6}{C_6}{}{}{} \ctB{A_7}{B_7}{C_7}{}{}{}\big) \\ 
\label{eqn:s:step1}\transreduces{\beta}{\param}\;\; & \lbd{z_1}{A_1}{\ctB{A_2}{B_2}{C_2}{}{}{} (\ctB{A_3}{B_3}{C_3}{}{}{} \ctW{A_4}{B_4}{}{}) \ctC{A_5}{B_5}{C_5}{}{}{} \big(\ctB{A_6}{B_6}{C_6}{}{}{} \ctB{A_7}{B_7}{C_7}{}{}{} \;z_1\big) }
\end{align}
By inspection of the table, the bound variable $\typd{z_1}{A_1}$ is of the same type as the first bound variable of $\ctS{A}{B}{C}{}{}{}$ and this is of regular type.

Then, under this bound variable, we apply Theorem~\ref{prop:combinatoryeffect} and then Proposition~\ref{prop:combinatoryeffecteasy} which we can do since $\typd{z_1}{A_1}$ is of regular type:
\begin{align}
 & \ctB{A_2}{B_2}{C_2}{}{}{} (\ctB{A_3}{B_3}{C_3}{}{}{} \ctW{A_4}{B_4}{}{}) \ctC{A_5}{B_5}{C_5}{}{}{} \big(\ctB{A_6}{B_6}{C_6}{}{}{} \ctB{A_7}{B_7}{C_7}{}{}{} \;z_1\big) \notag  \\
\transreduces{\beta}{\param} \;\;&  \ctB{A_3}{B_3}{C_3}{}{}{} \ctW{A_4}{B_4}{}{} \bigg( \ctC{A_5}{B_5}{C_5}{}{}{} \big(\ctB{A_6}{B_6}{C_6}{}{}{} \ctB{A_7}{B_7}{C_7}{}{}{} \;z_1\big)\bigg) \notag\\
\transreduces{\beta}{\param}\;\;\; & \lbd{z_3}{A_3}{\ctW{A_4}{B_4}{}{} \bigg( \ctC{A_5}{B_5}{C_5}{}{}{} \big(\ctB{A_6}{B_6}{C_6}{}{}{} \ctB{A_7}{B_7}{C_7}{}{}{} \;z_1\big) \;z_3 \bigg)}\notag
\end{align}
By inspection of the table, the bound variable $\typd{z_3}{A_3}$ is of the same type as the second bound variable of $\ctS{A}{B}{C}{}{}{}$ and this is of regular type.

Under this bound variable we apply Proposition~\ref{prop:combinatoryeffecteasy}, which we can do since $\typd{z_1}{A_1}$ and $\typd{z_3}{A_3}$ are of regular type:
\begin{align*}
& \ctW{A_4}{B_4}{}{} \bigg( \ctC{A_5}{B_5}{C_5}{}{}{} \big(\ctB{A_6}{B_6}{C_6}{}{}{} \ctB{A_7}{B_7}{C_7}{}{}{} \;z_1\big) \;z_3 \bigg) \\
 \reduces{\beta}{\param} \;\;&  \lbd{z_4}{A_4}{\ctC{A_5}{B_5}{C_5}{}{}{} \big(\ctB{A_6}{B_6}{C_6}{}{}{} \ctB{A_7}{B_7}{C_7}{}{}{} \;z_1\big) \; z_3 \; z_4 \; z_4} 
 \end{align*}
 By inspection of the table, the bound variable $\typd{z_4}{A_4}$ is of the same type as the third bound variable of $\ctS{A}{B}{C}{}{}{}$.

Finally, under this bound variable we repeatedly apply Theorem~\ref{prop:combinatoryeffect}:
\begin{align*}
 &  \ctC{A_5}{B_5}{C_5}{}{}{} \big(\ctB{A_6}{B_6}{C_6}{}{}{} \ctB{A_7}{B_7}{C_7}{}{}{} \;z_1\big) \; z_3 \; z_4 \; z_4 \\
 \transreduces{\beta}{\param} \;\;& \ctB{A_6}{B_6}{C_6}{}{}{} \ctB{A_7}{B_7}{C_7}{}{}{} \;z_1 \; z_4 \; z_3 \; z_4\\
\transreduces{\beta}{\param}\;\;& \ctB{A_7}{B_7}{C_7}{}{}{} \;(z_1 \; z_4)  \; z_3 \; z_4\\
\transreduces{\beta}{\param}\; \; &  z_1 \; z_4  \; (z_3 \; z_4)
\end{align*}
\end{proofdetail}
\end{proof}

\begin{prop}[Recovery of identity]\label{prop:wegotidentitybirdback}
If $B$ is a regular type then $\ctI{B}{}$ is an expanded $\mathsf{BCDKW}$-combinatorial term of $\mltup$. In particular, there is a $\mathsf{BCDKW}$-combinatorial term $N$ witnessing this which satisfies $N\transreduces{\beta}{\param}  \ctI{B}{}$.
\end{prop}
\begin{proof}
One can check that $\ctS{B \rightarrow B}{B}{B}{}{}{} \ctK{B}{B \rightarrow B}{}{} \ctK{B}{B}{}{}$ works.
\begin{proofdetail}
In more detail, one has that $\ctS{B \rightarrow B}{B}{B}{}{}{} \ctK{B}{B \rightarrow B}{}{} \ctK{B}{B}{}{}$ is a term of $\mltup$, since $\ctK{B}{B \rightarrow B}{}{}$ has type $B\rightarrow (B\rightarrow B)\rightarrow B$ which is also the type of the first bound variable of $\ctS{B \rightarrow B}{B}{B}{}{}{}$, and since $\ctK{B}{B}{}{}$ has type $B\rightarrow B\rightarrow B$ which is also which is also the type of the second bound variable of $\ctS{B \rightarrow B}{B}{B}{}{}{}$. Then one uses Proposition~\ref{prop:combinatoryeffecteasy} and Theorem~\ref{prop:combinatoryeffect} to obtain: 
\begin{equation*}
\ctS{B \rightarrow B}{B}{B}{}{}{} \ctK{B}{B \rightarrow B}{}{} \ctK{B}{B}{}{} \transreduces{\beta}{\param} \lbd{z}{B}{\ctK{B}{B \rightarrow B}{}{} \, z \; (\ctK{B}{B}{}{} \, z)} \transreduces{\beta}{\param} \lbd{z}{B}{z}
\end{equation*}
\end{proofdetail}
\end{proof}

It is not obvious, from the definitions, whether the following is true:

\begin{qu}[Relation between $\mathsf{BCDKW}$ and $\mathsf{SK}$]\label{qu:whatabouts}
Are the $\mathsf{BCDKW}$ combinators expanded $\mathsf{SK}$ combinatorial in $\mltup$?
\end{qu}

However, we can show that one traditional definition of Cardinal in terms of Starling and Kestrel fails:

\begin{rem}[On recovering Cardinal from Starling and Kestrel]\label{rmk:whatabouts}
To show that $\ctC{A}{A}{B}{}{}{}$ is expanded $\mathsf{SK}$-combinatorial in $\mlt{n}$ for $n>1$ when $A$ is a state type and $B$ is a regular type, one might try to follow the untyped reduction (\cite[pp. 158-159]{Curry1958-yi}), and seek to find types $A_i, B_i, C_i$ of $\mlt{n}$ such that
\begin{equation}\label{eqn:rmk:whatabouts}
\ctS{A_1}{B_1}{C_1}{\big(\ctB{A_2}{B_2}{C_2}{}{}{}\,\ctB{A_3}{B_3}{C_3}{}{}{}\,\ctS{A_4}{B_4}{C_4}{}{}{}\big)}{\big(\ctK{A_5}{B_5}{}{}\ctK{A_6}{B_6}{}{}\big)}{} \transreduces{\beta}{n} \ctC{A}{A}{B}{}{}{}
\end{equation}

But we will argue that this is not possible. For reductio, suppose not, and suppose that we indeed had (\ref{eqn:rmk:whatabouts}). Then the type of the terms in (\ref{eqn:rmk:whatabouts}) is both $C_1\rightarrow B_1$ and $(A\rightarrow A\rightarrow B)\rightarrow A\rightarrow A\rightarrow B$. Hence in particular $C_1$ is a regular type. 

But then consider the following reductions:
\begin{align*}
& \ctS{A_1}{B_1}{C_1}{\big(\ctB{A_2}{B_2}{C_2}{}{}{}\,\ctB{A_3}{B_3}{C_3}{}{}{}\,\ctS{A_4}{B_4}{C_4}{}{}{}\big)}{\big(\ctK{A_5}{B_5}{}{}\ctK{A_6}{B_6}{}{}\big)}{} & \\
\transreduces{\beta}{n} \; & \lbd{x}{C_1}{\big(\ctB{A_2}{B_2}{C_2}{}{}{}\,\ctB{A_3}{B_3}{C_3}{}{}{}\,\ctS{A_4}{B_4}{C_4}{}{}{}\big)x \, (\ctK{A_5}{B_5}{}{}\,\ctK{A_6}{B_6}{}{}\,x)} & \mbox{by Prop.~\ref{prop:combinatoryeffecteasy}}\\
\transreduces{\beta}{n} \; & \lbd{x}{C_1}{\big(\ctB{A_2}{B_2}{C_2}{}{}{}\,\ctB{A_3}{B_3}{C_3}{}{}{}\,\ctS{A_4}{B_4}{C_4}{}{}{}\big)x \; \ctK{A_6}{B_6}{}{}} & \mbox{by Thm.~\ref{prop:combinatoryeffect}} \\
\transreduces{\beta}{n} \; & \lbd{x}{C_1}{\big( \lbd{y}{A_2}{\ctB{A_3}{B_3}{C_3}{}{}{}  (\ctS{A_4}{B_4}{C_4}{}{}{} y)}\big)x \; \ctK{A_6}{B_6}{}{}}& \mbox{by Prop.~\ref{prop:combinatoryeffecteasy}}  \\
\reduces{\beta}{n} \; & \lbd{x}{C_1}{ \ctB{A_3}{B_3}{C_3}{}{}{}  (\ctS{A_4}{B_4}{C_4}{}{}{}\, x)  \; \ctK{A_6}{B_6}{}{}} \\
\transreduces{\beta}{n} \; & \lbd{x}{C_1}{ \lbd{y}{A_3}{\ctS{A_4}{B_4}{C_4}{}{}{}\, x \, (\ctK{A_6}{B_6}{}{}}\, y) }   \;  & \mbox{by Prop.~\ref{prop:combinatoryeffecteasy}}
\end{align*}
Note that $\ctK{A_6}{B_6}{}{}\, y$ being well-formed implies that $y$ has type $A_6$, which is required to be regular (cf. Kestrel in Definition~\ref{defn:typedcombo}), and together with the earlier lambda abstract $\lbd{y}{A_3}{\ldots}$ implies that $A_3, A_6$ are identical and are regular. Then we can continue:
\begin{align*}
 \transreduces{\beta}{n} \;& \lbd{x}{C_1}{ \lbd{y}{A_3}{\lbd{z}{C_4}{ xz (\ctK{A_6}{B_6}{}{}}\, y\, z)}}  & \mbox{by Prop.~\ref{prop:combinatoryeffecteasy}} \\
  \transreduces{\beta}{n} \;& \lbd{x}{C_1}{ \lbd{y}{A_3}{\lbd{z}{C_4}{ xzy}}} & \mbox{by Thm.~\ref{prop:combinatoryeffect}.}
\end{align*}
If this is to be the Cardinal $\ctC{A}{A}{B}{}{}{}$, then one has $C_1$ is $A\rightarrow A\rightarrow B$ and $A_3, C_4$ are $A$. But then $A_3$ is a state type, contradicting our earlier conclusion that $A_3$ is a regular type. 

\end{rem}

In the previous remark, our redutio hypothesis involved Bluebirds, and unlike Cardinal these are indeed recoverable from Starling and Kesterl:

\begin{prop}[Recovering Bluebird from Starling and Kestrel]\label{prop:recoverblubirdinsk}
Suppose that $A,B,C$ are types and $B,C$ are regular types. Then $\ctB{A}{B}{C}{}{}{}$ is 
an expanded $\mathsf{SK}$-combinatorial term of $\mltup$.
\end{prop}
\begin{proof}
Following the untyped reduction (cf. \cite[158]{Curry1958-yi}), one can check that:
\begin{equation}\label{eqn:recoverblubirdinsk}
\ctS{A_1}{B_1}{C_1}{}{}{} \big(\ctK{A_2}{B_2}{}{} \, \ctS{A_3}{B_3}{C_3}{}{}{}  \big) \, \ctK{A_4}{B_4}{}{} \transreduces{\beta}{\param} \ctB{A}{B}{C}{}{}{}
\end{equation}
where the types of $A_i, B_i, C_i$ are calculated in terms of $A,B,C$ as follows:\\[5pt]
\begin{minipage}[b][5\baselineskip][t]{0.95\textwidth}
\centering
\begin{tabular}{|l|l|l|l|}
\hline
$i$ & $A_i$ & $B_i$ & $C_i$  \\ \hline
 $1$ & $A\rightarrow B\rightarrow C$ & $(A\rightarrow B)\rightarrow A\rightarrow C$ & $B\rightarrow C$  \\ \hline
 $2$ & $(A\rightarrow B\rightarrow C)\rightarrow (A\rightarrow B)\rightarrow A\rightarrow C$ & $B\rightarrow C$ &  n/a  \\ \hline 
 $3$ & $B$ & $C$ & $A$  \\ \hline
 $4$ & $B\rightarrow C$ & $A$ & n/a  \\ \hline
\end{tabular}
\end{minipage}
\qedhere

\begin{proofdetail}
In more detail, note that since $A,B,C$ are types and $B,C$ are regular types, we have that:
\begin{itemize}[leftmargin=*]
    \item $\ctS{A_i}{B_i}{C_i}{}{}{}$ for $i=1,3$ is a term of $\mltup$ because $A_i,B_i$ are regular types by inspection of the table.
    \item $\ctK{A_i}{B_i}{}{}$ for $i=2,4$ is a term of $\mltup$ because $A_i$ is a regular type by inspection of the table.
\end{itemize}

Before proving the claim, we first verify that $\ctS{A_1}{B_1}{C_1}{}{}{} \big(\ctK{A_2}{B_2}{}{} \, \ctS{A_3}{B_3}{C_3}{}{}{}  \big) \, \ctK{A_4}{B_4}{}{}$ is a term of $\mltup$ of the same type as $\ctB{A}{B}{C}{}{}{}$:
\begin{itemize}
    \item The Kestrel $\ctK{A_2}{B_2}{}{}$ takes input type $A_2$, which by the table is $(A\rightarrow B\rightarrow C)\rightarrow (A\rightarrow B)\rightarrow A\rightarrow C$. The Starling $\ctS{A_3}{B_3}{C_3}{}{}{}$ has type $(C_3\rightarrow A_3\rightarrow B_3)\rightarrow (C_3\rightarrow A_3)\rightarrow C_3\rightarrow B_3$, which by inspection of the table is also $(A\rightarrow B\rightarrow C)\rightarrow (A\rightarrow B)\rightarrow A\rightarrow C$.
    \item The term $\ctK{A_2}{B_2}{}{} \, \ctS{A_3}{B_3}{C_3}{}{}{}$ has type the same as the output type of $\ctK{A_2}{B_2}{}{}$, which is $B_2\rightarrow A_2$. By inspection of the table this is $(B\rightarrow C)\rightarrow (A\rightarrow B\rightarrow C)\rightarrow (A\rightarrow B)\rightarrow A\rightarrow C$. The input type of $\ctS{A_1}{B_1}{C_1}{}{}{}$ is $C_1\rightarrow A_1\rightarrow B_1$, which by inspection of the table is also $(B\rightarrow C)\rightarrow (A\rightarrow B\rightarrow C)\rightarrow (A\rightarrow B)\rightarrow A\rightarrow C$.
    \item The term $\ctS{A_1}{B_1}{C_1}{}{}{} \big(\ctK{A_2}{B_2}{}{} \, \ctS{A_3}{B_3}{C_3}{}{}{}  \big) $ has type which is the same as the output type of $\ctS{A_1}{B_1}{C_1}{}{}{}$, namely $(C_1\rightarrow A_1)\rightarrow C_1\rightarrow B_1$. Note that $C_1\rightarrow A_1$ is $(B\rightarrow C)\rightarrow A\rightarrow B\rightarrow C$, and $C_1\rightarrow B_1$ is $(B\rightarrow C)\rightarrow (A\rightarrow B)\rightarrow A\rightarrow C$.
\begin{itemize}
    \item But $\ctK{A_4}{B_4}{}{} $ has type $A_4\rightarrow B_4\rightarrow A_4$, which is $(B\rightarrow C)\rightarrow A\rightarrow B\rightarrow C$, which is the same as our earlier result for $C_1\rightarrow A_1$.
    \item And $\ctB{A}{B}{C}{}{}{}$ has type $(B\rightarrow C)\rightarrow (A\rightarrow B)\rightarrow A\rightarrow C$, which is the same as our earlier result for $C_1\rightarrow B_1$.
\end{itemize}
\end{itemize}

Finally, we argue for (\ref{eqn:recoverblubirdinsk}) as follows:
\begin{align*}
& \ctS{A_1}{B_1}{C_1}{}{}{} \big(\ctK{A_2}{B_2}{}{} \, \ctS{A_3}{B_3}{C_3}{}{}{}  \big) \, \ctK{A_4}{B_4}{}{} & \\
\transreduces{\beta}{\param}\; & \lbd{x}{C_1}{\ctK{A_2}{B_2}{}{} \, \ctS{A_3}{B_3}{C_3}{}{}{} \, x \; ( \ctK{A_4}{B_4}{}{} \, x )} & \mbox{by Prop.~\ref{prop:combinatoryeffecteasy}.} \\
\transreduces{\beta}{\param}\; & \lbd{x}{C_1}{\ctS{A_3}{B_3}{C_3}{}{}{} \; ( \ctK{A_4}{B_4}{}{} \, x )} & \mbox{by Thm.~\ref{prop:combinatoryeffect}.} \\
\transreduces{\beta}{\param}\; & \lbd{x}{C_1}{\lbd{y}{C_3\rightarrow A_3}{\lbd{z}{C_3}{\ctK{A_4}{B_4}{}{} \, xz \, (yz) }}  } &  \\
\transreduces{\beta}{\param}\; & \lbd{x}{C_1}{\lbd{y}{C_3\rightarrow A_3}{\lbd{z}{C_3}{x (yz) }} } & \mbox{by Thm.~\ref{prop:combinatoryeffect}.} \\
\equiv \;\; & \lbd{x}{B\rightarrow C}{\lbd{y}{A\rightarrow B}{\lbd{z}{A}{x (yz) }} } &  
\end{align*}
In this, the fourth line follows from the third by two instances of $\beta_0$-reduction, which we can do since the only free variable of $\ctK{A_4}{B_4}{}{} \, x$ is the displayed $\typd{x}{C_1}$, which is $\typd{x}{B\rightarrow C}$ of regular type.
\end{proofdetail}
\end{proof}

%% file: 04-combinatory/03-combinatorial-models.tex
\subsection{Combinatorial characterisation of models}\label{sec:combinatorial-models}

The following is the key result needed to establish Theorem~\ref{corcreatemodel}. As mentioned in \S\ref{sec:intro}, this generalizes work of Andrews for the ordinary simply-typed lambda calculus (\cite[Lemma 1 p. 388]{Andrews1972-xk}). Whereas Andrews was able to use the Starling combinator at the induction step, in $\mltup$ we use Starling as well as Warbler, Cardinal, and Dardinal.\footnote{While we are not pursuing weakness until \S\ref{sec:purecombinatory}, we note in passing that the proof of this theorem contains many examples of non-weak $\beta$-reductions.}

\begin{thm}\label{thm:bigcombo}
Suppose $A,B$ are types and $B$ is a regular type and $\typd{v}{A}$ is a variable of $\mltup$ and $\typd{M}{B}$ is a $\mathsf{BCDKW}$-combinatorial term of $\mltup$. Then $\lbd{v}{A}{M}$ is an expanded $\mathsf{BCDKW}$-combinatorial term of $\mltup$.
\end{thm}
\begin{proof}
We show by induction on complexity of the $\mathsf{BCDKW}$-combinatorial term $\typd{M}{B}$ with regular type $B$ that for every type $A$ and variable $\typd{v}{A}$ there is a $\mathsf{BCDKW}$-combinatorial term $\typd{N}{A\rightarrow B}$ with the same free variables as $\lbd{v}{A}{M}$ such that $N\transreduces{\beta}{\param} \lbd{v}{A}{M}$.

As a first case, if $\typd{M}{B}$ is the variable $\typd{v}{A}$, then let $N$ be~$\ctI{A}{}$, so that we are done by Proposition~\ref{prop:wegotidentitybirdback}.

As a second case, suppose that $\typd{M}{B}$ is a variable $\typd{u}{B}$ distinct from $\typd{v}{A}$. Then $\ctK{B}{A}{M}$ is $\mathsf{BCDKW}$-combinatorial with the same free variables as $\lbd{v}{A}{M}$, namely $\typd{u}{B}$.   Further we have $\ctK{B}{A}{M}{} \reduces{\beta}{\param} \lbd{v}{A}{M}$ by Proposition~\ref{prop:combinatoryeffecteasy}.

As a third case, suppose that $\typd{M}{B}$ is a constant $\typd{c}{B}$. Then $\ctK{B}{A}{M}{}$ is $\mathsf{BCDKW}$-combinatorial with the same free variables as $\lbd{v}{A}{M}$, namely no free variables. Further we have $\ctK{B}{A}{M}{} \reduces{\beta}{\param} \lbd{v}{A}{M}$ by Proposition~\ref{prop:combinatoryeffecteasy}.

As a fourth case, suppose that $\typd{M}{B}$ is an instance of $\ctB{A^{\prime}}{B^{\prime}}{C^{\prime}}{}{}{}$, $\ctC{A^{\prime}}{B^{\prime}}{C^{\prime}}{}{}{}$, $\ctD{c}{A^{\prime}}{B^{\prime}}{C^{\prime}}{}{}{}$, $\ctK{A^{\prime}}{B^{\prime}}{}{}$, $\ctW{A^{\prime}}{B^{\prime}}{}{}$. These are closed terms, and are handled exactly as the previous case. 

As a fifth case, suppose that $\typd{M}{B}$ is $\typd{M_0 M_1}{B}$ where $\typd{M_0}{C\rightarrow B}$ and $\typd{M_1}{C}$.

There are several subcases to consider.

As a first subcase, suppose that $C$ is a regular type. By induction hypothesis, there are $\mathsf{BCDKW}$-combinatorial terms $\typd{N_0}{A \rightarrow C\rightarrow B}, \typd{N_1}{A\rightarrow C}$ such that $N_0\transreduces{\beta}{\param} \lbd{v}{A}{M_0}$ and $N_1\transreduces{\beta}{\param} \lbd{v}{A}{M_1}$, with the two terms $N_i$ and $\lbd{v}{A}{M_i}$ having the same free variables for $i=0,1$. Note that the first bound variable of $\ctS{C}{B}{A}{}{}{}$ has type $A \rightarrow C\rightarrow B$, which is the same type as $N_0$. And note that the second bound variable of $\ctS{C}{B}{A}{}{}{}$ has type $A\rightarrow C$, which is the type of $N_1$. Then $\typd{\ctS{C}{B}{A} N_0 N_1}{A\rightarrow B}$ is an expanded $\mathsf{BCDKW}$-combinatorial term by induction hypothesis and Proposition~\ref{prop:wegotstarlingback}; and one has $\ctS{C}{B}{A} N_0 N_1 \transreduces{\beta}{\param} \ctS{C}{B}{A} (\lbd{v}{A}{M_0})(\lbd{v}{A}{M_1})$. The third bound variable of $\ctS{C}{B}{A}{}{}{}$ is $\typd{z}{A}$, which we can switch to $\typd{v}{A}$ by $\alpha$-conversion. Further, the first two bound variables of $\ctS{C}{B}{A}{}{}{}$ are of regular type and so we may switch these so that they too do not appear free in $\lbd{v}{A}{M_0}, \lbd{v}{A}{M_1}$. Then we  continue to reduce as follows:
\begin{align*}
 \ctS{C}{B}{A} (\lbd{v}{A}{M_0})(\lbd{v}{A}{M_1}) & \reduces{\beta}{\param} \big(\lbd{y}{A\rightarrow C}{\lbd{v}{A}{(\lbd{v}{A}{M_0})v (yv)}}\big)(\lbd{v}{A}{M_1})  \\
& \reduces{\beta}{\param} \lbd{v}{A}{ (\lbd{v}{A}{M_0})v((\lbd{v}{A}{M_1})v)} \\
&  \transreduces{\beta}{\param} \lbd{v}{A}{M_0 M_1}
\end{align*}
The first $\beta$-reduction follows by disjointness: by previous $\alpha$-conversion the second bound variable $\typd{y}{A\rightarrow C}$ of Starling  does not appear free in $\lbd{v}{A}{M_0}$; and the third bound variable $\typd{v}{A}$ of Starling is not free in $\lbd{v}{A}{M_0}$. The second $\beta$-reduction follows since the displayed free occurrence of $\typd{y}{A\rightarrow C}$ in $(\lbd{v}{A}{M_0})v (yv)$ is its only free occurrence, since by previous $\alpha$-conversion this second bound variable does not appear free in  $\lbd{v}{A}{M_0}$. The two applications of $\beta$ in the last line simply follow from a variable being free for itself.

As a second subcase, suppose that $C$ is a state type. Since $C$ is a state type, the Starling $\ctS{C}{B}{A}{}{}{}$ is no longer a term of $\mltup$, and hence we need to use other combinators. By induction hypothesis, there is $\mathsf{BCDKW}$-combinatorial term $\typd{N_0}{A\rightarrow C\rightarrow B}$ such that $N_0\transreduces{\beta}{\param} \lbd{v}{A}{M_0}$ and such that $N_0$ and $\lbd{v}{A}{M_0}$ have the same free variables. Then $N_0 v\transreduces{\beta}{\param} M_0$. By Proposition~\ref{prop:state-free} the term $\typd{M_1}{C}$ is a variable or a constant of state type $C$. There are three further subcases to consider. 
\begin{itemize}[leftmargin=*]
\item First suppose that $\typd{M_1}{C}$ is $\typd{v}{A}$, which implies in particular that $A,C$ are identical state types. By $\alpha$-conversion if necessary change the second bound variable of $\ctW{A}{B}{}$ to $\typd{v}{A}$. This Warbler term is a term of $\mltup$ since $B$ is regular by hypothesis of the theorem. Since $\typd{v}{A}$ is not free in $N_0$ we have by disjointness that 
\begin{equation*}
\ctW{A}{B}{N_0}{} \reduces{\beta}{\param} \lbd{v}{A}{N_0 v v} \transreduces{\beta}{\param} \lbd{v}{A}{M_0 v} 
\end{equation*} 
The last $\beta$-reduction follows from the aforementioned consequence of the induction hypothesis, namely $N_0 v\transreduces{\beta}{\param} M_0$.
\item Second suppose that $\typd{M_1}{C}$ is a variable but not $\typd{v}{A}$. Suppose in particular that $\typd{M_1}{C}$ is the variable $\typd{y}{C}$ distinct from $\typd{v}{A}$. By this distinctness and $B$ being regular by the hypothesis of the theorem, the following is a Cardinal term of $\mltup$ (cf. distinctness in the definition of Cardinal term of $\mltup$ in Definition~\ref{defn:typedcombo}): 
\begin{equation}\label{eqn:instnaceofcinmmiddle}
\ctC{A}{C}{B}{}{}{} \equiv \lbd{x}{A\rightarrow C\rightarrow B}{\lbd{y}{C}{\lbd{v}{A}{xv y}}}
\end{equation}
Since $\typd{M_1}{C}$ is a variable of state type $C$ which is distinct from $\typd{v}{A}$, it follows that the first and third bound variables of $\ctC{A}{C}{B}{}{}{}$ are not free in $\typd{M_1}{C}$. Further $\typd{M_1}{C}$ is free for $\typd{y}{C}$ in $\lbd{v}{A}{xvy}$. Then we have the first step in the following, which is a $\beta$-reduction of distance~1:
\begin{align}
& \ctC{A}{C}{B}{N_0}{M_1}{}\reduces{\beta}{\param} \big(\lbd{x}{A\rightarrow C\rightarrow B}{\lbd{v}{A}{x v M_1}}\big) N_0 \notag \\
& \reduces{\beta}{\param} \lbd{v}{A}{N_0 v M_1}\transreduces{\beta}{\param} \lbd{v}{A}{M_0 M_1}
\end{align}
The second $\beta$-reduction follows since $\typd{v}{A}$ is not free in $N_0$, and since the displayed free occurrence of $\typd{x}{A\rightarrow C\rightarrow B}$ is the only one since $\typd{M_1}{C}$ is of state type. The third $\beta$-reduction follows from the aforementioned consequence of the induction hypothesis, namely $N_0 v\transreduces{\beta}{\param} M_0$.
\item Third suppose that $\typd{M_1}{C}$ is a constant $\typd{c}{C}$. Then the following is a Dardinal term of $\mltup$:
\begin{equation*}
\ctD{c}{A}{C}{B} \equiv \lbd{x}{A\rightarrow C\rightarrow B}{\lbd{v}{A}{xvc}}
\end{equation*}
Since $\typd{v}{A}$ is not free in $N_0$, we have that $N_0$ is free for $x$ in $\lbd{v}{A}{xvc}$, and so we have the first $\beta$-reduction in the following:
\begin{equation*}
\ctD{c}{A}{C}{B} N_0 \reduces{\beta}{\param} \lbd{v}{A}{N_0 vc} \transreduces{\beta}{\param} \lbd{v}{A}{M_0 c}
\end{equation*}
The second $\beta$-reduction follows from the aforementioned consequence of the induction hypothesis, namely $N_0 v\transreduces{\beta}{\param} M_0$. \qedhere
\end{itemize}
\end{proof}

\begin{cor}\label{prop:keyandrews}
Every term of $\mltup$ is expanded $\mathsf{BCDKW}$-combinatorial in $\mltup$.
\end{cor}
\begin{proof}
This is by induction on complexity of term. A variable or constant is itself $\mathsf{BCDKW}$-combinatorial. The inductive step for application is trivial. The inductive step for lambda abstraction follows from the previous theorem.
\end{proof}

\corcreatemodel*
\begin{proof}
The necessity is obvious. For sufficiency, suppose that $\typd{M}{A}$ is a term. We must show that its denotation is well-defined in $\mathcal{M}$. By the previous Corollary, there is $\mathsf{BCDKW}$-combinatorial term $\typd{N}{A}$ such that $N\transreduces{\beta}{\param} M$. By the hypothesis, the denotation of $N$ is well-defined in $\mathcal{M}$. By Proposition~\ref{prop:reducebetadelta}, the denotation of $M$ is well-defined in $\mathcal{M}$. 
\end{proof}

%% file: 04-combinatory/04-conservation.tex
\subsection{Semantic conservation and expressibility}\label{sec:conservation}

\thmconservation*
\begin{proof}
Let $\mathcal{M}$ be a model of $\mltup$. By Theorem~\ref{corcreatemodel}, it suffices to ensure that the denotations of all instances of $\mathsf{B}, \mathsf{C}, \mathsf{D}, \mathsf{K}, \mathsf{W}$ which are terms of $\mltpure$ are well-defined in $\mathcal{M}$. By inspection of $\mathsf{B}, \mathsf{C}, \mathsf{D}, \mathsf{K}, \mathsf{W}$ in Definition~\ref{defn:typedcombo} the only instance we need to check is the following instance of Cardinal where $C$ is a regular type and $B$ is a state type and $j>0$:
\begin{equation}\label{eqn:niceM1}
\ctC{B}{B}{C}{}{}{}\equiv  \lbd{v}{B \rightarrow B\rightarrow   C}{\lbd{v_0}{B}{\lbd{v_j}{B}{vv_jv_0}}} 
\end{equation}
If $\param(B)>1$, then we are done since we can $\alpha$-convert $\typd{v_j}{B}$ to $\typd{v_1}{B}$. Henceforth assume $\param(B)=1$. Without loss of generality, we can consider ourselves to be in the case where $\mathcal{M}$ is a model of $\mlt{1}$.

Consider the following closed term $M$ of $\mlt{1}$:
\begin{align}
 M \hspace{2mm} \equiv \hspace{2mm} & \lambda v: B\rightarrow B\rightarrow C . \;\;\lambda v_0: B \label{eqn:niceM2} \\
& \hspace{10mm} {\bigg( \lbd{V}{(B\rightarrow C)\rightarrow C}{\lbd{v_0}{B}{V(vv_0)}}\bigg) \big( \lbd{U}{B\rightarrow C}{Uv_0} \big) } \notag
\end{align}
Since $\mathcal{M}$ is a model of $\mlt{1}$, one has that the denotation of $M$ is well-defined in $\mathcal{M}$.

As in the discussion of Example~\ref{ex:thelbock}, one cannot do $\beta$-reduction in $\mlt{1}$ directly on $M$ to get $\ctC{B}{B}{C}{}{}{}$, since the term $\lbd{U}{B\rightarrow C}{Uv_0}$ is not free for $V$ in $\lbd{v_0}{B}{(V(vv_0))}$; and one cannot do $\alpha$-conversion to convert $\lbd{v_0}{B}{(V(vv_0))}$ to $\lbd{v_j}{B}{(V(vv_j))}$ since $v_0$ is the only variable of type $B$ when $\param(B)=1$.

However, the semantics for lambda abstraction are given by lambda abstraction out in the metatheory (cf. discussion subsequent to Definition~\ref{defn:model}). In particular, the denotation $\db{M}_{\mathcal{M}}$ in $\mathcal{M}$ of the closed term $M$ of $\mlt{1}$ is equal to the following:
\begin{align*}
&\Lbd{v}{\mathcal{M}(B\rightarrow B\rightarrow C)}{\Lbd{v_0}{\mathcal{M}(B)}{}} \\
& \hspace{5mm} \bigg( \Lbd{V}{(\mathcal{M}((B\rightarrow C)\rightarrow C)}{\Lbd{v_0}{\mathcal{M}(B)}{V(vv_0)}}\bigg) \big( \Lbd{U}{\mathcal{M}(B\rightarrow C)}{Uv_0} \big) 
\end{align*}
Since one can do $\alpha$-conversion out in the metatheory, one can $\alpha$-convert the meta-term $\Lbd{v_0}{\mathcal{M}(B)}{(V(vv_0))}$ to $\Lbd{v_j}{\mathcal{M}(B)}{(V(vv_j))}$ for some $j>0$; and then by $\beta$-reduction out in the metatheory, following the reduction given in Example~\ref{ex:thelbockxi}, one can get the denotation $\db{\ctC{B}{B}{C}{}{}{}}_{\mathcal{M}}$ of the closed term $\ctC{B}{B}{C}{}{}{}$ of $\mltpure$.
\end{proof}

We can similarly establish:

\express*
\begin{proof}
We fix a signature which contains exactly the constants of $\typd{N}{A}$. By Corollary~\ref{prop:keyandrews} applied to $\mltpure$, one has that $\typd{N}{A}$ is expanded $\mathsf{BCDKW}$-combinatorial in $\mltpure$, and hence there is $\mathsf{BCDKW}$-combinatorial term $\typd{L}{A}$ in $\mltpure$ with the same free variables as $\typd{N}{A}$ such that $L\transreduces{\beta}{\omega} N$. Let $\mathcal{M}$ be a model of $\mltpure$. By Proposition~\ref{prop:reducebetadelta}, one has that $\db{L}_{\mathcal{M}, \rho}=\db{N}_{\mathcal{M}, \rho}$ for all variable assignments $\rho$ relative to $\mathcal{M}$.
 
 By the same argument as the previous theorem, there is term $\typd{M}{A}$ of $\mltup$ with the same free variables as $\typd{L}{A}$ such that $\db{M}_{\mathcal{M}, \rho}=\db{L}_{\mathcal{M}, \rho}$ for all variable assignments $\rho$ relative to~$\mathcal{M}$. Further, this choice of $\typd{M}{A}$ does not depend on $\mathcal{M}$: one simply uniformly replaces (\ref{eqn:niceM1}) in $\typd{L}{A}$ with (\ref{eqn:niceM2}) to form $\typd{M}{A}$. By the previous paragraph, we then have that $\db{M}_{\mathcal{M}, \rho}=\db{L}_{\mathcal{M}, \rho}=\db{N}_{\mathcal{M}, \rho}$ for all variable assignments $\rho$ relative to $\mathcal{M}$.

 Since the free variables of $\typd{M}{A}, \typd{N}{A}$ are the same and are in $\mltup$, and since we have fixed a signature throughout, we are done.
\end{proof}

%% file: 04-combinatory/05-deductive.tex
\subsection{Deductive completeness and conservation}\label{sec:infty}

In the previous section, we discussed semantic notions like semantic conservation and semantic expressibility. We turn now to deductive matters. 

As described in \S\ref{sec:intro}, we define the $\infty$-rule in $\mltup$ as follows:
\begin{defi}[$\infty$-rule]\label{defn:inftyrule}
Suppose that $M,N:A$ are terms of $\mltup$. Then one defines the reduction relation $\infty$, called the \emph{$\infty$-rule}, by $M\reduces{\infty}{\param} N$ iff $M\equals{\beta\eta}{\omega} N$. 
\end{defi}

\noindent Since $\beta\eta$-equality is compatible (cf. \S\ref{subsec:conventionsredu}), the $\infty$-rule is itself compatible. One then has:

\begin{thm}[Deductive conservation of $\mltpure$ over $\mltup$ with $\infty$-rule]\label{thmconservemltupinfty}
For terms $\typd{M,N}{A}$ of $\mltup$, one has $\mltpure\vdash_{\beta\eta} M=N$ iff $\mltup\vdash_{\beta\eta\infty} M=N$.

\end{thm}
\begin{proof}
Suppose $\typd{M,N}{A}$ are terms of $\mltup$.

First suppose $\mltpure\vdash_{\beta\eta} M=N$. Then $M\equals{\beta\eta}{\omega} N$. By definition, $M\reduces{\infty}{\param} N$. Then $M\reduces{\beta\eta\infty}{\param} N$.

Second suppose that  $\mltup\vdash_{\beta\eta\infty} M=N$. Let $M_0, \ldots, M_n$ be terms of $\mltup$ such that $M$ is $M_0$ and $N$ is $M_n$ for some $n\geq 1$ such that for each $i<n$ one has that $M_i\reduces{\gamma}{\theta} M_{i+1}$ or $M_{i+1}\reduces{\gamma}{\theta} M_i$ for some $\gamma$ in $\{\beta,\eta,\infty\}$. In the case of $\gamma$ in $\{\beta,\eta\}$ we continue to have $M_i\reduces{\gamma}{\omega} M_{i+1}$ or $M_{i+1}\reduces{\gamma}{\omega} M_i$ and in the case of $\gamma=\infty$ we have $M_i\equals{\beta\eta}{\omega} M_{i+1}$ by definition of the $\infty$-rule and the aforementioned fact that it is compatible. Hence the chain is also a chain of $\beta\eta$ equalities in $\mltpure$, and so we have $\mltpure\vdash_{\beta\eta} M=N$.
\end{proof}

Further one has:

\thmcompletemltup*

\begin{proof}
 Suppose $\typd{M,N}{A}$ are terms of $\mltup$. By Theorem~\ref{thmconservemltupinfty}, one has that $\mltup\vdash_{\beta\eta\infty} M=N$ iff $\mltpure\vdash_{\beta\eta} M=N$, which by the Completeness Theorem for $\mltpure$ (Theorem~\ref{thmcompletelambda}) happens iff $\mltpure\models M=N$, which by the semantic conservation result (Theorem~\ref{thmconservation}) happens iff $\mltup\models M=N$.
\end{proof}

\noindent Note further that deductive completeness, together with the semantic expressibility result (Theorem~\ref{thm:express}) likewise entails a deductive expressibility result of $\mltpure$ in $\mltup$ equipped with the $\infty$-rule.

The natural question is the following:
\begin{qu}[Deductive conservation and completeness with just $\beta\eta$-rules]\label{q:dedcompcon}
Suppose that $\typd{M,N}{A}$ are terms $\mltup$. Then do the following hold?

\begin{enumerate}[label=(\alph*), ref=\alph*]
    \item\label{q:dedcompcon:1} $\mltup\vdash_{\beta\eta} M=N$ iff $\mltpure\vdash_{\beta\eta} M=N$.
    \item\label{q:dedcompcon:2}  $\mltup\vdash_{\beta\eta} M=N$ iff $\mltup\models M=N$.
\end{enumerate}

\end{qu}

We make a couple of remarks on these questions. First, it is not hard to see that the semantic conservation result (Theorem~\ref{thmconservation}) and Completeness Theorem for $\mltpure$ (Theorem~\ref{thmcompletelambda}) entails that (\ref{q:dedcompcon:1}) happens iff (\ref{q:dedcompcon:2}) happens. Second, by Theorem~\ref{thm:bigcombo}, any failure of (\ref{q:dedcompcon:1}) can be expressed in terms of ``missing'' equalities between expanded $\mathsf{BCDKW}$ combinatorial terms. This then recommends a study of pure $\mathsf{BCDKW}$ combinatorial logic, which we make a start on in~\S\ref{sec:purecombinatory}. Third, failure of (\ref{q:dedcompcon:2}) might recommend either (i)~finding the an extension of the deductive system (perhaps by finding the aforementioned missing equalities), or (ii)~expanding the semantics beyond the ordinary set theoretic ones considered in this paper.

Fourth and finally, our work in \S\ref{sec:conservation} suggests a potential route to a negative resolution to Question~\ref{q:dedcompcon}.

\begin{exa}\label{exa:howaboutpotentialcounter}

Let $B$ be a state type, and $C$ a regular type. Recall the term $M$ of $\mlt{1}$ from equation~(\ref{eqn:niceM2}) of the proof of Theorem~\ref{thmconservation}. Let $\kappa>1$, i.e. $\kappa=n$ for some $n>1$ or $\kappa=\omega$. As in Example~\ref{ex:thelbockxi} one has that $M\transreduces{\beta_0}{\kappa}  \ctC{B}{B}{C}{}{}{}$, using $\beta_0$-reductions under a lambda abstract. Hence $\mlt{\kappa}\vdash_{\beta\eta} M= \ctC{B}{B}{C}{}{}{}$. Suppose $\typd{v}{B\rightarrow B\rightarrow C}$ is a variable of functional type. Then Cardinal $\ctC{B}{B}{C}{}{}{}$ has the behavior in $\mlt{\kappa}$ of taking a function $v$ of two inputs and returning the function $(v_0, v_1)\mapsto vv_1 v_0$, i.e. permuting the order of the inputs (cf. Definition~\ref{defn:typedcombo} and Theorem~\ref{prop:combinatoryeffect}). Since permuting the order twice has the effect of returning the original function, one can check that $\mlt{\kappa}\vdash_{\beta\eta} \ctC{B}{B}{C}{}{}{} (\ctC{B}{B}{C}{}{}{} v) =v$. Then one also has $\mlt{\kappa}\vdash_{\beta\eta} M (M v) =v$. By the Soundness Theorem~\ref{thm:soundness}, one also has $\mlt{\kappa}\models M (M v) =v$. Since $M$ is a term of $\mlt{1}$, by Theorem~\ref{thmconservation} we have $\mlt{1}\models M (M v) =v$.
\end{exa}

We can then ask:
\begin{qu}\label{qu:howaboutpotentialcounter}
With term $M$ of $\mlt{1}$ as in the previous example (also recorded in equation~(\ref{eqn:niceM2})), is it the case that $\mlt{1}\vdash_{\beta\eta} M (M v) =v$?
\end{qu}

\noindent If this negative resolution went through, then one could also consider analogous examples for $\mlt{n}$ for $n>1$ which involved generalizations of Cardinal which took a function of $(n+1)$-inputs of a given state type and cyclically permuted the order of the inputs. The difficulty in completing these negative resolutions is that the usual methods for establishing non-provability are unavailable: (1)~we cannot use the usual set theoretic models and the Soundness Theorem since we already know that $\mlt{1}\models M (M v) =v$, and (2)~Church-Rosser is in general unavailable in $\mlt{n}$ (cf. Example~\ref{ex:friedmanwarren}) and hence one cannot convert a $\beta\eta$-equality into $\beta\eta$-reductions to a common third term. In \S\ref{sec:purecombinatory} we develop pure combinatory logic and prove the deductive conservation result of Theorem~\ref{thm:conservemltpureovermltup}, which allows us to transfer \emph{weak} $\beta$-equalities in $\mltpure$ to $\beta\eta$-equalities in $\mlt{1}$. But the natural reduction $M\transreduces{\beta_0}{\kappa}  \ctC{B}{B}{C}{}{}{}$ for $\kappa>1$ is not a weak $\beta_0$-reduction, as we noted in Example~\ref{ex:thelbockxi}. And the natural proof of $\mlt{\kappa}\vdash_{\beta\eta} \ctC{B}{B}{C}{}{}{} (\ctC{B}{B}{C}{}{}{} v) =v$ uses $\eta$. Hence Theorem~\ref{thm:conservemltpureovermltup} cannot obviously be used to answer Question~\ref{qu:howaboutpotentialcounter} in the affirmative.

%% file: 05-ty2/00-intro.tex
\section{The simply-typed lambda calculus: modal and non-modal}\label{sec:ty2}

In this section, we return to the relation between the simply-typed modal lambda calculus $\mltpure$ and the ordinary simply-typed lambda calculus $\lt$. Per the discussion in \S\ref{sec:nonmodal} (cf. Definition~\ref{defn:nonmodal}), we assume in this section that the basic entity types of $\lt$ are the union of the state types and basic entity types of $\mltpure$. As a reminder, if $A$ is a state type of $\mltpure$, then $A\rightarrow A$ is a type of $\lt$ but not of $\mltpure$ (cf. Examples~\ref{ex:montre}-\ref{ex:finvar} for concrete examples). That is, the types of $\mltpure$ are a proper subset of the types of $\lt$, and the same is true of the terms.

%% file: 05-ty2/01-conserve.tex
\subsection{Conservation and expressibility}\label{sec:conserlt}

We first prove:

\thmconservationlt*

\begin{proof}
Suppose that $A$ is a type of $\mltpure$ and suppose that $\typd{M,N}{A}$ are terms of $\mltpure$ such that $\lt\vdash_{\beta\eta} M=N$. By Church-Rosser for $\lt$, we have that there is a term $\typd{L}{A}$ of $\lt$ such that $M\twoheadrightarrow_{\beta\eta} L$ and $N\twoheadrightarrow_{\beta\eta} L$. But since $\beta\eta$ reduction preserves being a term of $\mltpure$, we have that $\typd{L}{A}$ is also a term of $\mltpure$, along with all the other terms in the two $\beta\eta$-chains. Then we have $\mltpure\vdash_{\beta\eta} M=L$ and $\mltpure\vdash_{\beta\eta} N=L$ and so $\mltpure\vdash_{\beta\eta} M=N$. 
\end{proof}

A more complicated argument shows:

\expresslt*

\begin{proof}
We show by induction on length of $\beta\eta$-reduced term $\typd{N}{A}$ of $\lt$ that if the free variables and constants of $\typd{N}{A}$ have types in  $\mltpure$, then there is a term $\typd{M}{A}$ of $\mltpure$ with the same free variables and constants such that $\lt\vdash_{\beta\eta} M=N$. The result then follows by using normalization for $\lt$ to pass to a $\beta\eta$-normal form.

For length $\ell=1$, we have that $\typd{N}{A}$ is a variable or a constant, and so these must be variables of $\mltpure$ or constants of $\mltpure$; and so we can take $\typd{M}{A}$ to be $\typd{N}{A}$.

Suppose that the result holds for length $<\ell$; we show it holds for length $\leq \ell$. 

First suppose that $\typd{N}{A}$ is an application. Since each term contains at most finitely many applications, for some $k\geq 1$ there are terms 
\begin{equation*}
\typd{N_0}{A_1\rightarrow\cdots \rightarrow A_k \rightarrow A},\hspace{10mm} \typd{N_1}{A_1}, \ldots \typd{N_k}{A_k}
\end{equation*}
such that $N$ is $N_0 N_1 \cdots N_k$ and $N_0$ is not an application. Since $N$ is $\beta\eta$-reduced, $N_0$ is not a lambda abstract. Hence, $N_0$ is a constant or a variable. But by hypothesis this constant or variable has type in $\mltpure$. Then $A_1\rightarrow\cdots \rightarrow A_k \rightarrow A$ is a type of $\mltpure$. Hence $A_1, \ldots, A_k$ are types of $\mltpure$. Further, since $\typd{N_1}{A_1}, \ldots \typd{N_k}{A_k}$ are subterms of $\typd{N}{A}$, we also have that their free variables and constants must be variables of $\mltpure$ or constants of $\mltpure$, and further these subterms must themselves be $\beta\eta$-reduced. Hence, by induction hypothesis applied $k$-many times to $\typd{N_1}{A_1}, \ldots \typd{N_k}{A_k}$, these are respectively expressed by $\typd{M_1}{M_1}, \ldots \typd{M_k}{A_k}$, and then $\typd{M}{A}$ is expressed by $N_0 M_1 \cdots M_k$.

Second suppose that $\typd{N}{A}$ is a lambda abstract $\lbd{v}{C}{N_0}$, where $\typd{N_0}{D}$, so that $A$ is $C\rightarrow D$. Since $A$ is a type of $\mltpure$, we have that $C,D$ are types of $\mltpure$. Hence $\typd{v}{C}$ has type in $\mltpure$. Then all the free variables $\typd{N_0}{D}$ have types in $\mltpure$. Further, since the constants of $\typd{N}{A}$ and $\typd{N_0}{D}$ are the same, we have that the constants of $\typd{N_0}{D}$ are in $\mltpure$. Finally, since $\typd{N}{A}$ is in $\beta\eta$-normal form, we have that $\typd{N_0}{D}$ is in $\beta\eta$-normal form. Hence, by induction hypothesis, we have that $\typd{N_0}{D}$ is expressed by $\typd{M_0}{D}$, and so $\typd{M}{A}$ is expressed by $\lbd{v}{A}{M_0}$.
\end{proof}

In the immediate sequel to this paper (\cite{Walsh2024-mltII}), we extend these results to Church's simple theory of types. The above proof breaks down in that setting since there are object-language identities in the simple theory of types, and one can have such an identity $N_0=_A N_1$ of type $T$ (the type of truth-values), while $A$ is not a type of $\mltpure$ but only a type of $\lt$.

%% file: 05-ty2/02-termredeux.tex
\subsection{Application: open term models revisited}\label{sec:opentermmodelsredeux}

In this section, we revisit the discussion of open term models from \S\ref{sec:completeness-open-term}. For, there it was shown that the open term applicative structure $\mathcal{O}_1$ of $\mlt{1}$ does not determine a frame, while the open term applicative structure of $\mathcal{O}_{\omega}$ of $\mltpure$ does determine a frame and indeed determines a model. Using the theorems proved thus far, we round out the discussion by showing that for $n>1$ the open term applicative structure $\mathcal{O}_n^{\infty}$ of $\mlt{n}$ with the addition of the $\infty$-rule (cf. Definition~\ref{defn:inftyrule}) determines both a frame and a model.

Before doing that, we begin with the following preliminary result purely about $\lt$.\footnote{In the statement, we use the notation mentioned in footnote~\ref{fn:vecsame}.} 
\begin{prop}\label{prop:longnormalform}
Suppose that $n>1$ and that $A,B$ are distinct types of $\lt$ with $A$ atomic. Suppose that $B$ is not of the form $C_1\rightarrow \cdots \rightarrow C_{\ell}\rightarrow A$ for any types $C_1, \ldots, C_{\ell}$. Suppose that $\typd{M,N}{A^{n+1}\rightarrow B}$ are terms of $\lt$ that have no free variables of type $A$, and that $\typd{v_0, \ldots, v_{n-1}}{A}$ are distinct variables which we abbreviate as $\typd{\vec{v}}{A}$, such that for each $i<n$ one has that $\lt\vdash_{\beta\eta} M\vec{v}v_i =N\vec{v}v_i$. Then $\lt\vdash_{\beta\eta} M=N$.
\end{prop}
\noindent Before beginning the proof, we note that the proposition and its proof emerge rather naturally by considering the impediments to generalizing Example~\ref{exa:failureopenterm1} to $n>1$.
\begin{proof}
Let $B$ be $C_1\rightarrow \cdots \rightarrow C_{\ell}\rightarrow D$, where $D$ is atomic, and $D$ is by hypothesis distinct from $A$. Note we may have $\ell=0$ when $B$ is an atomic distinct from $A$. We fix variables $\typd{y_1}{C_1}, \ldots, \typd{y_n}{C_{\ell}}$ which we abbreviate as $\typd{\vec{y}}{\vec{C}}$, which we suppose to be disjoint from $\typd{\vec{v}}{A}$. We show by induction on $m\geq 1$ that if $M,N$ as in the statement of the proposition have respective long normal forms $\lbd{\vec{x}}{A}{\lbd{\vec{y}}{\vec{C}}{UM_1\cdots M_k}}$ and $\lbd{\vec{x}}{A}{\lbd{\vec{y}}{\vec{C}}{VN_1\cdots N_{k^{\prime}}}}$ where $\typd{\vec{x}}{A}$ is a vector $\typd{x_0, \ldots x_n}{A}$ of distinct variables of length $n+1$ which is disjoint from both $\typd{\vec{v}}{A}$ and $\typd{\vec{y}}{\vec{C}}$, and finally where $M_i,N_i$ is in long normal form, and where each $M_i,N_i$ has length $\leq m$, then $M\equals{\beta\eta}{} N$.\footnote{See \cite[p. 32]{Barendregt2013-eb} and \cite[p. 79]{Sorensen2006-ey} for information on long normal forms.} In this, $U,V$ are variables or constants. Note that $U,V$ are not of type $A$: for if $U$ was of type $A$, then since $A$ is atomic we would have that $k=0$ and then $B$ would have type $C_1\rightarrow \cdots \rightarrow C_{\ell}\rightarrow A$, contrary to hypothesis. Further, the hypotheses imply that for each $i<n$ one has:
\begin{align*}
& (UM_1\cdots M_k)[x_0:=v_0, \ldots, x_{n-1}=v_{n-1}, x_n:=v_i]  \\ & \hspace{20mm} \equals{\beta\eta}{}  (VN_1\cdots N_{k^{\prime}})[x_0:=v_0, \ldots, x_{n-1}=v_{n-1}, x_n:=v_i] 
\end{align*}
By Church-Rosser, for each $i<n$, these will $\beta\eta$-reduce to a common term, and in these reductions the heads $U,V$ will not change and all of the $\beta\eta$ reductions will happen internal to $M_j[x_0:=v_0, \ldots, x_{n-1}=v_{n-1}, x_n:=v_i]$ and $N_j[x_0:=v_0, \ldots, x_{n-1}=v_{n-1}, x_n:=v_i]$. This implies that $U,V$ are identical and that $k^{\prime}=k$ and that for each $1\leq j\leq k$ and each $i<n$ one has:
\begin{equation}
M_j[x_0:=v_0, \ldots, x_{n-1}=v_{n-1}, x_n:=v_i] \equals{\beta\eta}{} N_j[x_0:=v_0, \ldots, x_{n-1}=v_{n-1}, x_n:=v_i] \label{eqn:longnf1}
\end{equation}

First, we consider the base case of $m=1$. In this case, one has that each $M_j,N_j$ is a constant or a variable. If it is a constant or a variable of a different type than $A$ or from variables of type $A$ distinct from those in the vector~$\typd{\vec{x}}{A}$, then we can conclude from (\ref{eqn:longnf1}) that $M_j,N_j$ are identical. If it is a variable of type $A$ from the vector~$\typd{\vec{x}}{A}$, then we argue that $M_j,N_j$ are the same variable from the vector $\typd{\vec{x}}{A}$. Suppose that $M_j$ is $x_a$ and $N_j$ is $x_b$, where $a,b\leq n$. Then we argue that $a=b$ by considering the three possible cases: 

Case 1: Suppose $a,b<n$ are distinct. Then (\ref{eqn:longnf1}) for any $i<n$ says that $v_a\equals{\beta\eta}{} v_b$, a contradiction.

Case 2: Suppose $a<n$ and $b=n$. Since $n>1$ and $a<n$, we can find $i<n$ with $i\neq a$. Then (\ref{eqn:longnf1}) for this $i<n$ says that $v_a\equals{\beta\eta}{} v_i$, a contradiction.

Case 3: Suppose $a=n$ and $b<n$. This follows as in Case 2, but with the role of $a,b$ reversed.

Second, we consider the induction step $m>1$. It suffices to show that $M_j\equals{\beta\eta}{} N_j$ for all $1\leq j\leq k$. For the rest of the proof, fix $1\leq j\leq k$. Since $M_j,N_j$ are in long normal form, they can be written respectively as $\lbd{\vec{w}}{\vec{D}}{X\vec{P}}$ and $\lbd{\vec{w}}{\vec{D}}{Y\vec{Q}}$, with each $P_i$ and each $Q_i$ having length $<m$. If we define $P:=\lbd{\vec{x}}{A}{M_j}$ and $Q:=\lbd{\vec{x}}{A}{N_j}$, then these are also in long normal form. Further, (\ref{eqn:longnf1}) implies that $P\vec{v}v_i\equals{\beta\eta}{} Q\vec{v}v_i$ for all $i<n$. Thus by induction hypothesis applied to $P,Q$ we are done.
\end{proof}

One could not do the proof of the previous proposition in non-maximal $\mltup$ since the proof uses Church-Rosser (cf. Example~\ref{ex:friedmanwarren}). Further, if one tried to extend the proof of the previous proposition to the case of $n=1$ in $\lt$, then the proof breaks down in Cases~2-3. Further, the conclusion of the previous proposition is in fact false in the case of $n=1$ and~$\lt$, as one can see by the below example, which is closely related to Example~\ref{exa:failureopenterm1}:

\begin{exa}\label{exa:failureopenterm2} 
Let $\typd{U}{A^2\rightarrow B}$ be a variable. Let $M$ be the term $\lbd{v_0}{A}{\lbd{v_1}{A}{Uv_0v_0}}$ of $\lt$ and let $N$ be $U$. Then one has $Mv_0v_0\equals{\beta\eta}{} Nv_0v_0$ but one does not have $M\equals{\beta\eta}{} N$.
\end{exa}

We now apply Proposition~\ref{prop:longnormalform} to $\mlt{n}$ for $n>1$:
\begin{prop}\label{prop:determineatermmodel}
Suppose that $n>1$. Suppose that $A$ is a state type and $B$ is a regular type of $\mlt{n}$, and that $\typd{P,Q}{A\rightarrow B}$ are terms of $\mlt{n}$, and that distinct variables $\typd{v_0, \ldots, v_{n-1}}{A}$ are such that for each $i<n$ one has $Pv_i\equals{\beta\eta\infty}{n} Qv_i$. Then $P\equals{\beta\eta\infty}{n} Q$.
\end{prop}
\begin{proof}
Since $A$ is a state type and $B$ is a regular type, the hypotheses on the types from Proposition~\ref{prop:longnormalform} is satisfied. By Corollary~\ref{prop:keyandrews}, there is a $\mathsf{BCDKW}$-combinatorial term $P^{\prime}$ of $\mlt{n}$ with the same free variables as $P$ such that $P^{\prime}\transreduces{\beta}{n} P$. Likewise, there is a $\mathsf{BCDKW}$-combinatorial term $Q^{\prime}$ of $\mlt{n}$ with the same free variables as $Q$ such that $Q^{\prime}\transreduces{\beta}{n} Q$. We abbreviate the distinct variables $\typd{v_0, \ldots, v_{n-1}}{A}$ as $\typd{\vec{v}}{A}$. Let $M$ be the term $\lbd{\vec{v}}{A}{P^{\prime}}$ of $\mlt{n}$ and let $N$ be the term $\lbd{\vec{v}}{A}{Q^{\prime}}$ of $\mlt{n}$. Then one has that  $M\vec{v}v_i\equals{\beta\eta\infty}{n} N\vec{v}v_i$ for each $i<n$. Then one has $\lt\vdash_{\beta\eta} M\vec{v}v_i=N\vec{v}v_i$ for each $i<n$. Then by Proposition~\ref{prop:longnormalform} one has $\lt\vdash_{\beta\eta} M=N$. Then by Theorem~\ref{thmconservationlt} one has that $\mltpure\vdash_{\beta\eta} M=N$. Then by Theorem~\ref{thmcompletemltup} one has that $\mlt{n}\vdash_{\beta\eta\infty} M=N$. Then $\mlt{n}\vdash_{\beta\eta\infty} M\vec{v}=N\vec{v}$. Then $\mlt{n}\vdash_{\beta\eta\infty} P^{\prime}=Q^{\prime}$. Then $\mlt{n}\vdash_{\beta\eta\infty} P=Q$. 
\end{proof}

We then modify the definition of the open term applicative structure $\mathcal{O}_{\param}$ from Definition~\ref{defn:openterm} so that it is defined in terms of $\beta\eta\infty$-equality rather than $\beta\eta$-equality, and let $\mathcal{O}_{\param}^{\infty}$ be this modification. Then the previous proposition implies, following the proof of Proposition~\ref{prop:opentermstructureprop}, that for $n>1$ one has that the open term applicative structure $\mathcal{O}_{\param}^{\infty}$ determines a frame. Then one can show:

\begin{prop}
For each $n>1$, the frame determined by the open applicative structure~$\mathcal{O}_{n}^{\infty}$ is a model of $\mlt{n}$.
\end{prop}

\begin{proof}
Fix $n>1$. We show that the frame associated to the open term applicative structure $\mathcal{O}_n^{\infty}$ is a model. By Theorem~\ref{corcreatemodel}, it suffices to show that the denotations of all the $\mathsf{BCDKW}$-combinatorial terms of $\mlt{n}$ are well-defined in $\mathcal{O}_n^{\infty}$. For this, in turn, it suffices to show that this holds for the combinators. We show this for Cardinal since the others are similar. For all $\typd{P}{A\rightarrow B\rightarrow C}, \typd{Q}{B}, \typd{R}{A}$ one has 
$\db{xzy}_{\mathcal{O}_n^{\infty}, [x:=[P], y:=[Q], z:=[R]]} = [PRQ] = [\ctC{A}{B}{C}{P}{Q}{R}]$ by Theorem~\ref{prop:combinatoryeffect}. Then one has the following, where the first identity follows from the semantics for lambda abstraction (cf. Definition~\ref{defn:model}(\ref{defn:model:4})) and the last follows from the definition of the function from $\mathcal{O}^{\infty}_n(A)\rightarrow \mathcal{O}^{\infty}_n(C)$ determined by $[\ctC{A}{B}{C}{P}{Q}{}]$ (here we use that the open term applicative structure $\mathcal{O}_n^{\infty}$ determines a frame), and where the middle identity follows from our previous work:
\begin{align*}
& \db{\lbd{z}{A}{xyz}}_{\mathcal{O}^{\infty}_n, [x:=[P], y:=[Q]]} = \Lbd{[R]}{\mathcal{O}_n^{\infty}(A)}{\db{xzy}_{\mathcal{O}_n^{\infty}, [x:=[P], y:=[Q]]\,[z:=[R]]}}  \\
& = \Lbd{[R]}{\mathcal{O}_n^{\infty}(A)} [\ctC{A}{B}{C}{P}{Q}{}] [R] = [\ctC{A}{B}{C}{P}{Q}{}]
\end{align*}
We use this work now for the middle identity of the next step:
\begin{align*}
& \db{\lbd{y}{B}{\lbd{z}{A}{xyz}}}_{\mathcal{O}^{\infty}_n, [x:=[P]]} = \Lbd{[Q]}{\mathcal{O}_n^{\infty}(B)}{\db{\lbd{z}{A}{xzy}}_{\mathcal{O}_n^{\infty}, [x:=[P]]\, [y:=[Q]]}}  \\
& = \Lbd{[Q]}{\mathcal{O}_n^{\infty}(B)} [\ctC{A}{B}{C}{P}{}{}] [Q] = [\ctC{A}{B}{C}{P}{}{}]
\end{align*}
And finally we use this for the middle identity of the last step, where we abbreviate $A\rightarrow B\rightarrow C$ as $D$:
\begin{align*}
& \db{\lbd{x}{D}{\lbd{y}{B}{\lbd{z}{A}{xyz}}}}_{\mathcal{O}^{\infty}_n}  = \Lbd{[P]}{\mathcal{O}_n^{\infty}(D)}{\db{\lbd{y}{B}{\lbd{z}{A}{xzy}}}_{\mathcal{O}_n^{\infty}, [x:=[P]]}}  \\
& = \Lbd{[P]}{\mathcal{O}_n^{\infty}(B)} [\ctC{A}{B}{C}{}{}{}] [P] = [\ctC{A}{B}{C}{}{}{}]
\end{align*}
Hence one has $\db{\ctC{A}{B}{C}{}{}{}}_{\mathcal{O}_n^{\infty}(D)} = [\ctC{A}{B}{C}{}{}{}]$.
\end{proof}

%% file: 06-pure-combinatory/00-intro.tex
\section{Pure combinatory logic and the weak lambda calculus}\label{sec:purecombinatory}

In this section, we turn to proving Theorem~\ref{thm:conservemltpureovermltup}, which is finally established in \S\ref{sec:pureclcompleteconserve}. 

This theorem concerns \emph{weak} $\beta$-reduction in $\mltup$, designated as $\reduces{w\beta}{\param}$. Recall from the close of \S\ref{subsec:conventionsredu} that $\reduces{w\beta}{\param}$ is the smallest relation which includes $\beta$-reduction and is closed on application on both sides; and that $\equals{w\beta}{\param}$ is the smallest equivalence relation containing $\reduces{w\beta}{\param}$. Hence, these reductions and equalities are formed without assuming that they are preserved under taking of lambda abstracts.\footnote{That is, these equalities are formed without the $\xi$-rule (cf. footnote~\ref{fn:xirule}). As mentioned in \S\ref{sec:intro}, weak systems were studied in untyped settings 
by \c{C}a{\u{g}}man-Hindley \cite{Cagman1998-pi} and in typed settings by Sestini \cite{Sestini2019-cc}. But many evaluation strategies are weak, and many of them are surveyed in \cite{Sestoft2001-gf}, \cite{Biernacka2022-hr}} To see the effect of this restriction, it is useful to take note of the following elementary example:

\begin{exa}[Paradigmatic example of a $\beta$-reduction which is not a $w\beta$-reduction]
Suppose that $A$ is a regular type of $\mltup$ and $\typd{x,y}{A}$ are distinct variables. Then one has $\lbd{x}{A}{(\lbd{y}{A}{y})x} \reduces{\beta}{\param} \lbd{x}{A}{x}$, but one does \emph{not} have $\lbd{x}{A}{(\lbd{y}{A}{y})x} \reduces{w\beta}{\param} \lbd{x}{A}{x}$.
\end{exa}

\noindent For another natural example, see Example~\ref{ex:thelbockxi}. 

However, since our $\beta$-reduction notion is distanced, any instance of distance $\geq 1$ contains a simplification of a lambda abstract under another lambda abstract:
\begin{exa}[Distanced $\beta$-reduction and weak reduction, illustrated]
Consider the following, where $A$ is of state type and $B$  is regular type and $\typd{v}{A}$ and $\typd{y}{A\rightarrow B}$ and $\typd{U}{A\rightarrow B\rightarrow B}$ are variables:
\begin{equation}\label{eqn:distanceoneweakill}
(\lbd{x}{B}{\lbd{v}{A}{Uvx}}) (yv) v \reduces{w\beta}{\param} (\lbd{x}{B}{Uvx})(yv)\reduces{w\beta_0}{\param} Uv(yv)
\end{equation}
We have that the first is a weak $\beta$-reduction of distance one, and the second is a weak $\beta$-reduction of distance zero, i.e. a weak $\beta_0$-reduction. But if we lambda abstracted over $\lbd{U}{A\rightarrow B\rightarrow B}{\ldots}$ in (\ref{eqn:distanceoneweakill}) then these would not be weak $\beta$-reductions but rather just ordinary $\beta$-reductions. 
\end{exa}
\noindent This elementary example is important because it underscores that weak $\beta$-reduction means ``the closure of $\beta$-reduction (Definition~\ref{defn:betareduction}) under application on both sides'' and not ``no simplification of lambda abstracts under other lambda abstracts, in some more general sense of simplification.'' 

There is a similar moral for the equality relation:
\begin{exa}[Weak equality illustrated]\label{exa:weakequality}
Suppose that $\typd{M,N}{A}$ and $\typd{L}{B}$ are terms of $\mltup$ with $B$ is regular, and suppose that $M \reduces{w\beta}{\param} N$ and that $M,N$ are free for $\typd{v}{A}$ in $L$. Then one has $L[x:=M]\equals{w\beta}{\param} L[x:=N]$ since $L[x:=M] \reducesbw{w\beta}{\param} (\lbd{x}{A}{L})M \reduces{w\beta}{\param} (\lbd{x}{A}{L})N  \reduces{w\beta}{\param} L[x:=N]$. If $x$ occurs free in $L$ under lambda abstractions, then weak $\beta$-equality does allow some $\beta$-reduction under lambda abstracts.\footnote{Both \c{C}a{\u{g}}man-Hindley (\cite[Definition 4.5 p. 245]{Cagman1998-pi}) and Sestini (cf. \cite[$\sigma$-rule p. 1]{Sestini2019-cc}) assume that if $M \reduces{w\beta}{\param} N$ then $L[x:=M] \reduces{w\beta}{\param} L[x:=N]$. Example~\ref{exa:weakequality} shows that we do not need to build this in, if we are just interested in weak $\beta$-equality.}
\end{exa}

Again, the aim of this last section is to prove Theorem~\ref{thm:conservemltpureovermltup}. We prove it by developing a theory of pure intensional combinatory logic~$\cltup$ and by proving results which allow us to partially transfer information back and forth between weak~$\beta$-equalities in~$\mltup$ and the associated equalities in~$\cltup$. The outline of this section is as follows. In \S\ref{sec:cl:combosandweak} we set up the pure combinators and their reduction relation, which also bears the traditional name of weak reduction;\footnote{Historically, the name for weak reduction in combinatory logic came first; then the adjective ``weak'' was subsequently used for the arguably correspondent notion in the lambda calculus. See \c{C}a{\u{g}}man-Hindley \cite[p. 240]{Cagman1998-pi} for more details of the history.} in \S\ref{subsec:clcr} we prove the Church-Rosser property for~$\cltup$ using Takahashi's method of complete developments; in \S\ref{sec:simabs} we show how to simulate abstraction using~$\mathsf{BCDKW}$, which differs from the usual simulation using~$\mathsf{SKI}$; in \S\ref{sec:fromltocl} we show how to move from weak $\beta$-equality in~$\mltup$ to weak equality in~$\cltup$, culminating in Theorem~\ref{thm:frommltuptoclt}; in \S\ref{sec:fromcombotolambda} we show how to move from weak equality in~$\cltup$ to $\beta\eta$ equality in $\mltup$, culminating in Theorem~\ref{thm:bigone}. Finally, in \S\ref{sec:pureclcompleteconserve} we put this together to establish Theorem~\ref{thm:conservemltpureovermltup}.

%% file: 06-pure-combinatory/01-typed-combinators-CL.tex
\subsection{Pure typed combinators and weak reduction}\label{sec:cl:combosandweak}

The following definition simply postulates special typed constants corresponding to the typed combinatory terms from Definition~\ref{defn:typedcombo}. We omit Starling and Identity bird since we can take them as defined (cf. Propositions~\ref{prop:clrecoverystarling}, \ref{prop:clrecoveryidentity}).

\begin{defi}[Typed combinator terms of $\cltup$]\label{defn:typedcombo2}
Let $\param$ be a parameter, and let $A,B,C$ be types. Then the \emph{combinator terms of $\cltup$} are 
\begin{enumerate}[leftmargin=*]
    \item \emph{Kestrel} $\ctK{A}{B}{}{}$ which has type $A\rightarrow B\rightarrow A$. It is required that $A$ has regular type.
    \item\label{defn:typedcombo2:2} \emph{Cardinal} $\ctC{A}{B}{C}{}{}{}$ which has type $(A\rightarrow B\rightarrow C)\rightarrow B\rightarrow A\rightarrow C$. It is required that $C$ has regular type and that either $A,B$ are distinct types, or $A,B$ are identical types with $\param(A)=\param(B)>1$. 
    \item \emph{Dardinal} $\ctD{c}{A}{B}{C}{}{}$ which has type $(A\rightarrow B\rightarrow C)\rightarrow A\rightarrow C$. It is required that $C$ has regular type, that $B$ is a state type, and that $\typd{c}{B}$ is a constant.
    \item \emph{Warbler} $\ctW{A}{B}{}{}$ which has type $(A\rightarrow  A\rightarrow  B) \rightarrow A\rightarrow  B$. It is required that $B$ is of regular type.
    \item \emph{Bluebird} $\ctB{A}{B}{C}{}{}{}$, of type  $(B\rightarrow C) \rightarrow (A \rightarrow B) \rightarrow A\rightarrow  C$. It is required that $B,C$ are regular types.
\end{enumerate}
\end{defi}

Each of the terms depends on the regularity of certain of their constitutive types $A,B,C$ (or just $A,B$ in the case of Kestrel and Warbler). But only Cardinal depends on the parameter $\param$ and will not be available in e.g. $\clt{1}$ when $A,B$ are identical state types. The constraint on Cardinal in (\ref{defn:typedcombo2:2}) was prefigured in Remark~\ref{rmk:oncardinalandstate}. The reason for this constraint is that we want to develop a combinatory logic that corresponds as much as possible to $\mltup$.

In parallel to Definition~\ref{defn:term}, we define:
\begin{defi}[Terms of $\cltup$]
Let $\param$ be a parameter and let a signature be fixed. Then the terms $\typd{M}{A}$ of $\cltup$ are defined as follows:
\begin{enumerate}[leftmargin=*]
    \item \label{defn:clterm:1} \emph{Variables}: the variables $\typd{v_i}{A}$ for $i<\param(A)$ are terms of $\cltup$.
    \item \label{defn:clterm:2} \emph{Constants}: the constants $\typd{c}{A}$ from the signature are terms of $\cltup$.
    \item \label{defn:clterm:3} \emph{Combinator terms}: the combinator terms of  $\cltup$ are terms of $\cltup$.
    \item \label{defn:clterm:4} \emph{Application}: If $\typd{M}{A\rightarrow B}$ and $\typd{N}{A}$ are terms of $\cltup$ then the application $\typd{(MN)}{B}$ is a term of $\cltup$.
\end{enumerate}
\end{defi}

We associate application to the left, so that the term $MNL$ is $(MN)L$. And we drop outer parentheses.

As with Proposition~\ref{prop:state-free}, we have:
\begin{prop}\label{prop:state-free:cl}
The only terms of $\cltup$ of state type are the variables and the constants.
\end{prop}
\begin{proofdetail}
Suppose $B$ is a state type. A term of type $B$ cannot be an application $MN$ since then we would have $\typd{M}{A\rightarrow B}$ and $\typd{N}{A}$, but $A\rightarrow B$ is not a type since $B$ is a state type. Also, a term of type $B$ cannot be a combinatory term since these always have functional type. Hence, the only remaining options for terms are constants and variables.
\end{proofdetail}

Parallel to Theorem~\ref{prop:combinatoryeffect} and Proposition~\ref{prop:thefirstcardinaltodardinal}, we define:
\begin{defi}[Weak reduction; redex and contractum]\label{defn:weakreduction}
We define \emph{weak reduction} $\reduces{w}{\param}$ to be the reduction relation on terms of $\cltup$ given by the following
\begin{align}
 \ctK{A}{B}{P}{Q}\reduces{w}{\param} P & \hspace{5mm}\ctC{A}{B}{C}{P}{Q}{R}\reduces{w}{\param} PRQ & \hspace{3mm} \ctD{c}{A}{B}{C}{P}{R}\reduces{w}{\param} PRc \notag \\
 \ctW{A}{B}{P}{Q}\reduces{w}{\param} PQQ  &\hspace{5mm}  \ctB{A}{B}{C}{P}{Q}{R}\reduces{w}{\param} P(QR) \label{eqn:weakreduction} & \ctC{A}{B}{C}{P}{c}\reduces{w}{\param} \ctD{c}{A}{B}{C}{P}{}
\end{align}
provided that the combinatory terms are combinatory terms of $\cltup$ and that the types are appropriate to make the applications well-defined (the typing will vary with the combinatory term).

We refer to terms on the left-hand side  of the $\reduces{w}{\param}$-arrows in (\ref{eqn:weakreduction}) as the \emph{redex} and we refer to the associated right-hand side as the \emph{contractum}.
\end{defi}

As in \S\ref{subsec:conventionsredu}, we let $\reduces{w}{\param}$ be the compatible closure of the relation defined by the above schemas, i.e. the smallest binary relation on terms of $\cltup$ containing the weak reductions which is closed under application on both sides. The latter means: if $\typd{P,Q}{A}$ and $P\reduces{w}{\param}Q$ then $MP\reduces{w}{\param} MQ$ for all terms $\typd{M}{A\rightarrow B}$ of $\cltup$; and likewise if $\typd{M,N}{A\rightarrow B}$ and $M\reduces{w}{\param} N$ then $MP\reduces{w}{\param} NP$ for all terms $\typd{P}{A}$ of $\cltup$.\footnote{As in \S\ref{subsec:conventionsredu}, we hasten to say that ``closed under application on both sides'' does \emph{not} mean: if $P\reduces{w}{\param}Q$ and $M\reduces{w}{\param}N$, then $MP\reduces{w}{\param}NQ$. This would be a parallel reduction notion (cf. \S\ref{subsec:clcr}), whereas the idea described in the body of the text is explicating the idea of a single weak reduction happening somewhere inside the term.} Note that there is \emph{no analogue} of closure under lambda abstraction in the weak reductions of combinatory logic. Finally, $\transreduces{w}{\param}$ is the reflexive transitive closure of $\reduces{w}{\param}$, while $\equals{w}{\param}$ is the smallest equivalence relation containing~$\reduces{w}{\param}$. 

For ease of future reference, we number the following remarks:
\begin{rem}[Omitting the typing in developing combinatory logic]\label{rmk:implicittyping}
In what follows, for ease of readability, we omit explicit descriptions of the typing of combinatory terms. This is because in combinatory logic, everything is done in terms of a large number of applications, and explicitly typing all of these would excessively complicate the description of even the simplest of inferences, like the weak reductions.
\end{rem}

\begin{rem}[Remark on the Cardinal-to-Dardinal weak reduction]\label{rmk:cardinal2dardingal}
The last weak reduction in (\ref{eqn:weakreduction}) has, as its parallel in $\mltup$, Proposition~\ref{prop:thefirstcardinaltodardinal} rather than Theorem~\ref{prop:combinatoryeffect}. We call this last weak reduction, namely the weak reduction $\ctC{A}{B}{C}{P}{c}\reduces{w}{\param} \ctD{c}{A}{B}{C}{P}{}$, \emph{the Cardinal-to-Dardinal weak reduction}. 
    
This is a weak reduction which is available in $\clt{n}$ for $n\geq 2$ but is not in general available in $\clt{1}$ since, when $A,B$ are identical state types, then $\ctC{A}{B}{C}{P}{c}$ is not a term of $\clt{1}$ since $\ctC{A}{B}{C}{}{}$ is not a term of $\clt{1}$. However, since Dardinal does not depend on the parameter $\param$ but does depend on the signature, we have that Dardinal is a term of $\cltup$ for all parameters $\param$ when the signature has constants of state type.
\end{rem}

\subsubsection{Appearance of variables}

In $\cltup$, like in all combinatory logics, there is no primitive binding of variables, although we can later introduce a simulation thereof (cf. \S\ref{sec:simabs}). Hence, we just speak of variables \emph{appearing in a term} or \emph{occurring in a term}. Terms with no variables appearing in them are called \emph{closed}.

We use $L[x:=N]$ for the result of substituting all occurrences of variable $\typd{x}{A}$ by term $\typd{N}{A}$ in term $L$. If $\typd{\vec{x}}{\vec{A}}$ is a pairwise distinct set of variables, then we use $L[\vec{x}:=\vec{N}]$ for the result of simultaneously substituting, in term $L$, all occurrences any variable in the vector $\typd{\vec{x}}{\vec{A}}$ by the corresponding term in $\typd{\vec{N}}{\vec{A}}$. The substitution lemma then reads as follows:

\begin{lem}[Substitution Lemma]\label{lem:substitution}
Suppose $P\transreduces{w}{\param} Q$. Then:
\begin{enumerate}
    \item\label{lem:substitution:1} The variables appearing in $Q$ are a subset of the variables appearing in $P$.
    \item\label{lem:substitution:2} $R[v:=P] \transreduces{w}{\param} R[v:=Q]$
    \item\label{lem:substitution:3} $P[\vec{x}:=\vec{N}] \transreduces{w}{\param} Q[\vec{x}:=\vec{N}] $
\end{enumerate}
\end{lem}
\begin{proof}
The proof is identical to \cite[Lemma 2.14 p. 25]{Hindley2008-vw}.
\end{proof}

\begin{proofdetail}
\begin{proof}
For reference, here is the detailed proof. It suffices to prove that (\ref{lem:substitution:1})-(\ref{lem:substitution:3}) hold when $P \reduces{w}{\param} Q$. Suppose that  $P\reduces{w}{\param} Q$. 

For (\ref{lem:substitution:1}), one just looks through the weak reductions in Definition~\ref{defn:weakreduction} and notes that any of the displayed subterms $P,Q,R$ of the contractum in which variables appear are formed from applications applied to subterms of the redex.  

For (\ref{lem:substitution:2}), this just follows by induction on $R$. If $R$ is $v$, then it follows from the hypothesis that $P\reduces{w}{\param} Q$. If $R$ is another variable $u$, then $R[v:=P] \reduces{w}{\param} R[v:=Q]$ is just $u\reduces{w}{\param} u$, which follows from reflexivity of $\transreduces{w}{\param}$. Similarly if $R$ is a constant or combinator. Suppose that $R$ is $MN$. Then by induction hypothesis we have $M[v:=P] \transreduces{w}{\param} M[v:=Q]$ and $N[v:=P] \transreduces{w}{\param} N[v:=Q]$. Then since $\transreduces{w}{\param}$ is closed under application, we have $M[v:=P]N[v:=P] \transreduces{w}{\param} M[v:=Q] N[v:=Q]$.

For (\ref{lem:substitution:3}), one has e.g.  $(\ctK{A}{B}{P}{Q})[\vec{x}:=\vec{N}] \equiv  (\ctK{A}{B}{P[\vec{x}:=\vec{N}]}{Q[\vec{x}:=\vec{N}]})  \reduces{w}{\param} P[\vec{x}:=\vec{N}]$ and the other base cases are similar. The induction steps for the compatible closure are trivial.
\end{proof}
\end{proofdetail}

\subsubsection{Combinatory logic and the partial order on parameters}\label{subsec:partialorder}

Recall from \S\ref{subsec:typesandterms} the natural partial order on parameters: $\param \leq \param^{\prime}$ iff for all state types $A$ one has $\param(A)\leq \param^{\prime}(A)$. 

Obviously given the definition of terms of $\cltup$ in Definition~\ref{defn:typedcombo2}, we have that if $\param \leq \param^{\prime}$, then all terms of $\cltup$ are terms of $\clt{\param^{\prime}}$. Further often it is a proper subset. For example, if $A,B$ are identical state types with $\param(A)=\param(B)=1$, then $\ctC{A}{B}{C}{}{}{}$ is not a term of $\cltup$, but it would be a term of $\clt{\param^{\prime}}$ for any $\param^{\prime} > \param$ with $\param^{\prime}(A)=\param^{\prime}(B)>1$.

\begin{lem}[Weak reduction preserves $\cltup$]\label{lem:preservescltup}
Suppose that $\param\leq \param^{\prime}$.

Suppose that $\typd{P}{A}$ is a term of $\cltup$, and suppose $\typd{Q}{A}$ is a term of $\clt{\param^{\prime}}$.

If $P\transreduces{w}{\param} Q$ then $Q$ is also a term of $\cltup$.

\end{lem}
\begin{proof}
It suffices to show it for $\reduces{w}{\param}$. But this follows by inspection of Definition~\ref{defn:weakreduction}: for we see that the contractum is formed by $\leq 3$ applications to subterms of the redex. The only exception to this is the Cardinal-to-Dardinal weak reduction (cf. Remark~\ref{rmk:cardinal2dardingal}), which additionally includes a new Dardinal term in the contractum. But since the Dardinal terms do not depend on the parameter it too is a term of $\cltup$.
\end{proof}

\subsubsection{Recovery of other combinators}

In parallel to Proposition~\ref{prop:wegotstarlingback} we have:
\begin{prop}[Recovery of Starling]\label{prop:clrecoverystarling}
Suppose $A,B,C$ are types and $A,B$ are regular. Then there is a closed term $\ctS{A}{B}{C}{}{}{}$ of $\cltup$ of type $(C\rightarrow   A\rightarrow  B) \rightarrow (C\rightarrow A)\rightarrow C\rightarrow B$ such that $ \ctS{A}{B}{C}{P}{Q}{R}\transreduces{w}{\param} PR(QR)$ for all terms $P,Q,R$ of $\cltup$ of the appropriate type to make the applications well-formed. 
\end{prop}
\begin{proof}
In $\cltup$, we may take $\ctS{A}{B}{C}{}{}{}$ to be the following term, where $A_i,B_i,C_i$ are defined in terms of $A,B,C$ as the proof of Proposition~\ref{prop:wegotstarlingback}:
\begin{equation*}
\ctB{A_1}{B_1}{C_1}{}{}{} (\ctB{A_2}{B_2}{C_2}{}{}{} (\ctB{A_3}{B_3}{C_3}{}{}{} \ctW{A_4}{B_4}{}{}) \ctC{A_5}{B_5}{C_5}{}{}{}) (\ctB{A_6}{B_6}{C_6}{}{}{} \ctB{A_7}{B_7}{C_7}{}{}{}) 
\end{equation*}
Then just use weak reductions.
\end{proof}

In parallel to Proposition~\ref{prop:wegotidentitybirdback} we have:
\begin{prop}[Recovery of Identity]\label{prop:clrecoveryidentity}
Suppose $B$ is a regular type. Then there is a closed term $\ctI{B}{}$ of $\cltup$ of type $B\rightarrow B$ such that $\ctI{B}{P}\transreduces{w}{\param} P$ for all terms $\typd{P}{B}$ of $\cltup$.
\end{prop}
\begin{proof}
Again we use $\ctS{B \rightarrow B}{B}{B}{}{}{} \ctK{B}{B \rightarrow B}{}{} \ctK{B}{B}{}{}$.
\end{proof}

%% file: 06-pure-combinatory/02-CR.tex
\subsection{Church-Rosser}\label{subsec:clcr}

In this subsection, we prove Church-Rosser for weak reduction in $\cltup$ (Theorem~\ref{thm:crforclt}). The proof follows closely the outline of Takahashi's proof of Church-Rosser for the untyped lambda calculus.\footnote{\cite{Takahashi1995-kp}. See \cite[\S{7.2}]{Cardone2006-qk} for discussion of the history of related proofs of Church-Rosser.}

It begins with a parallel reduction notion. As its name suggests, it is trying to isolate a notion where multiple weak reductions are happening simultaneously. 
\begin{defi}[Parallel reduction]\label{defn:takahashi}
The binary relation $\Reduces{w}{\param}$ is the least binary relation on terms of $\cltup$ of the same type
which satisfies:
\begin{enumerate}[leftmargin=*]
    \item \label{defn:takahashi1} $P\Reduces{w}{\param} P$ whenever $P$ is a variable, constant, or combinatory term.    
    \item \label{defn:takahashi2} If $P\Reduces{w}{\param} P^{\prime}$ and $Q\Reduces{w}{\param} Q^{\prime}$ and  $R\Reduces{w}{\param} R^{\prime}$, then
\begin{align*}
  \hspace{3mm} \ctK{A}{B}{P}{Q}\Reduces{w}{\param} P^{\prime} & \hspace{3mm}\ctC{A}{B}{C}{P}{Q}{R}\Reduces{w}{\param} P^{\prime}R^{\prime}Q^{\prime} &  \ctD{c}{A}{B}{C}{P}{R}\Reduces{w}{\param} P^{\prime}R^{\prime}c \\
\hspace{3mm} \ctW{A}{B}{P}{Q}\Reduces{w}{\param} P^{\prime}Q^{\prime}Q^{\prime}  &\hspace{3mm}  \ctB{A}{B}{C}{P}{Q}{R}\Reduces{w}{\param} P^{\prime}(Q^{\prime}R^{\prime}) & \hspace{3mm} \ctC{A}{B}{C}{P}{c}\Reduces{w}{\param} \ctD{c}{A}{B}{C}{P^{\prime}}{}
\end{align*}
provided that the combinatory terms are combinatory terms of $\cltup$ and that the types are appropriate to make the applications well-defined (the typing will vary with the combinatory term).
\item \label{defn:takahashi3} If  $P\Reduces{w}{\param} P^{\prime}$ and $Q\Reduces{w}{\param} Q^{\prime}$, then  $PQ\Reduces{w}{\param} P^{\prime}Q^{\prime}$, provided that  the types are appropriate to make the applications well-defined.
\end{enumerate}
\end{defi}

As with many inductive definitions, it can be built up from below:
\begin{prop}[Characterisation of parallel reduction ``from below'']
The definition of $\Reduces{w}{\param}$ in Definition~\ref{defn:takahashi} is equivalent to the union of $\Reduces{w,s}{\param}$ where we define this recursively in~$s\geq 0$:
\begin{enumerate}[leftmargin=*]
    \item \label{defn:takahashibelow1} For stage $s=0$, the relation $\Reduces{w,s}{\param}$ is the identity relation on variables, constants, and combinatory terms.
    \item \label{defn:takahashibelow2} For even stages $s\geq 0$, the relation $\Reduces{w,s+1}{\param}$ is the union of the previous stages plus $\ctW{A}{B}{P}{Q}\Reduces{w,s+1}{\param} P^{\prime}Q^{\prime}Q^{\prime}$ for all $P\Reduces{w,r}{\param} P^{\prime}$ and $Q\Reduces{w,t}{\param} Q^{\prime}$ with $r,t\leq s$; and similarly for the other weak reductions.
    \item \label{defn:takahashibelow3} For odd stages $s\geq 0$, the relation $\Reduces{w,s+1}{\param}$ is the union of the previous stages plus $PQ\Reduces{w,s+1}{\param} P^{\prime}Q^{\prime}$ for all $P\Reduces{w,r}{\param} P^{\prime}$ and $Q\Reduces{w,t}{\param} Q^{\prime}$ with $r,t\leq s$.
\end{enumerate}
\end{prop}
The proof is standard and so we omit it.

The characterization ``from below'' can be used to show the following:
\begin{prop}[Successors of base cases under parallel reduction; successors of non-redexes under parallel reduction]\label{prop:helperfortaka}\hfill
\begin{enumerate}[leftmargin=*]
    \item \label{prop:helperfortaka:1} If $M$ is a variable, constant, or combinatory term and $M\Reduces{w}{\param} N$ then $N$ is $M$.
    \item \label{prop:helperfortaka:2} If $M$ is an application $PQ$ which is not a redex and  $M\Reduces{w}{\param} N$, then $N$ is $P^{\prime}Q^{\prime}$ where $P\Reduces{w}{\param} P^{\prime}$ and $Q\Reduces{w}{\param} Q^{\prime}$.
\end{enumerate}
\end{prop}
\begin{proof}
For (\ref{prop:helperfortaka:1}), we show by induction on $s\geq 0$ that if $M$ is a variable, constant, or combinatory term and $M\Reduces{w,s}{\param} N$ then $N$ is $M$:
\begin{itemize}[leftmargin=*]
    \item For $s=0$, if we add $M\Reduces{w,s}{\param} N$, then  $M$ is a variable, constant, or combinatory term and $N$ is $M$.
    \item At stage $s+1$, we do not add any parallel reductions $M\Reduces{w,s+1}{\param} N$ with $M$ a variable, constant, or combinatory term; hence we are done by induction hypothesis.
\end{itemize}
In this argument and subsequent inductive arguments, we use ``add at a stage'' to mean that it is in the stage but not in any of the previous stages. 

For (\ref{prop:helperfortaka:2}), we show by induction on $s\geq 0$ that if $M$ is an application $PQ$ which is not a redex and  $M\Reduces{w,s}{\param} N$, then $N$ is $P^{\prime}Q^{\prime}$ where $P\Reduces{w,s}{\param} P^{\prime}$ and $Q\Reduces{w,s}{\param} Q^{\prime}$:
\begin{itemize}[leftmargin=*]
\item For $s=0$, we do not add any parallel reductions $M\Reduces{w,s}{\param} N$ where $M$ is an application.
\item At stage $s+1$ with $s$ even, we do not add any parallel reductions $M\Reduces{w,s+1}{\param} N$ where $M$ is an application which is not a redex; and hence we are done by induction hypothesis.
\item At stage $s+1$ with $s$ odd, if we add a parallel reduction $M\Reduces{w,s+1}{\param} N$ where $M$ is an application $PQ$ which is not a redex, then $N$ is $P^{\prime} Q^{\prime}$ where $P\Reduces{w,s}{\param} P^{\prime}$ and $Q\Reduces{w,s}{\param} Q^{\prime}$.
\end{itemize}
\end{proof}

Using simple inductive proofs which we omit, one can also identify the successors of the other combinators under parallel reduction:
\begin{prop}[Successors of Warblers under parallel reduction]\label{prop:stunted}\hfill
\begin{enumerate}[leftmargin=*]
    \item\label{prop:stunted:1} If $\ctW{A}{B}{}{} \Reduces{w}{\param} N$, then $N$ is identical to $\ctW{A}{B}{}{}$.
    \item\label{prop:stunted:2} If $\ctW{A}{B}{P}{} \Reduces{w}{\param} N$, then $N$ is identical to $\ctW{A}{B}{P_1}{}$ for some term $P_1$ such that $P\Reduces{w}{\param} P_1$.
    \item \label{prop:stunted:3} If $\ctW{A}{B}{P}{Q}\Reduces{w}{\param} N$, then one of the following occurs:
\begin{enumerate}[leftmargin=*, label=(\alph*), ref=\alph*]
    \item \label{eqn:rmk:twosucc1} $N$ is $P_1Q_1Q_1$ for some  terms $P_1,Q_1$ such that $P\Reduces{w}{\param} P_1$ and $Q\Reduces{w}{\param} Q_1$. 
    \item \label{eqn:rmk:twosucc2} $N$ is $\ctW{A}{B}{P_1}{Q_1}$ for some terms $P_1,Q_1$ such that $P\Reduces{w}{\param} P_1$ and $Q\Reduces{w}{\param} Q_1$. 
\end{enumerate}   
\end{enumerate}
\end{prop}

\begin{proofdetail}
If one is interested in the proof detail, we include it here. 

For (\ref{prop:stunted:1}), we show by induction on $s\geq 0$ that if $\ctW{A}{B}{}{} \Reduces{w,s}{\param} N$, then $N$ is identical to $\ctW{A}{B}{}{}$:
\begin{itemize}[leftmargin=*]
    \item At stage $s=0$, if we add $\ctW{A}{B}{}{} \Reduces{w,s}{\param} N$, then  $N$ is identical to $\ctW{A}{B}{}{}$.
    \item At stage $s+1$, we do not add any parallel reductions  $\ctW{A}{B}{}{} \Reduces{w,s+1}{\param} N$. Hence, we are done by induction hypothesis.
\end{itemize}

For (\ref{prop:stunted:2}), we show by induction on $s\geq 0$ that if $\ctW{A}{B}{P}{} \Reduces{w,s}{\param} N$, then $N$ is identical to $\ctW{A}{B}{P_1}{}$ for some term $P_1$ such that $P\Reduces{w}{\param} P_1$:
\begin{itemize}[leftmargin=*]
    \item At stage $s=0$ we do not add any parallel reductions  $\ctW{A}{B}{P}{} \Reduces{w,s}{\param} N$.
    \item At stage $s+1$ with $s$ even, we do not add any parallel reductions  $\ctW{A}{B}{P}{} \Reduces{w,s+1}{\param} N$. Hence, we are done by induction hypothesis.
    \item At stage $s+1$ with $s$ odd, if we add a parallel reduction  $\ctW{A}{B}{P}{} \Reduces{w,s+1}{\param} N$ then $N$ is identical to $OP_1$ where $\ctW{A}{B}{}{} \Reduces{w,s}{\param} O$ and $P\Reduces{w,s}{\param} P_1$. By (\ref{prop:stunted:1}), we have that $O$ is $\ctW{A}{B}{}{}$, and so we are done.
\end{itemize}

For (\ref{prop:stunted:3}), we show by induction on $s\geq 0$ that if $\ctW{A}{B}{P}{Q}\Reduces{w,s}{\param} N$, then one of (\ref{prop:stunted:3}\ref{eqn:rmk:twosucc1})-(\ref{prop:stunted:3}\ref{eqn:rmk:twosucc2}) occurs.
\begin{itemize}[leftmargin=*]
    \item For stage $s=0$, we do not add any parallel reductions $\ctW{A}{B}{P}{Q}\Reduces{w,s}{\param} N$.
    \item At stage $s+1$ with $s$ even, if we add a parallel reduction $\ctW{A}{B}{P}{Q} \Reduces{w,s+1}{\param} N$ then $N$ is identical to $P_1Q_1Q_1$ where $P\Reduces{w,s}{\param} P_1$ and $Q\Reduces{w,s}{\param} Q_1$. Then we are in case~(\ref{eqn:rmk:twosucc1}).
    \item At stage $s+1$ with $s$ odd, if we add a parallel reduction  $\ctW{A}{B}{P}{Q} \Reduces{w,s+1}{\param} N$, then $N$ is indentical to $OQ_1$ where  $\ctW{A}{B}{P}{} \Reduces{w,s}{\param} O$ and $Q\Reduces{w,s}{\param} Q_1$. By (\ref{prop:stunted:2}), we have that $O$ is identical to $\ctW{A}{B}{P_1}{}$ where $P\Reduces{w}{\param} P_1$. Then we are in case~(\ref{eqn:rmk:twosucc2}).
\end{itemize}

\end{proofdetail}

\noindent There are analogous propositions for the other ``two input'' combinatory terms of Kestrel and Dardinal, with the only difference being that one modifies (\ref{prop:stunted:3}\ref{eqn:rmk:twosucc1}) appropriately.

For the ``three input'' combinatory term of Cardinal, one has:
\begin{prop}[Successors of Cardinals under parallel reduction]\label{prop:cardinalsucc}\hfill
\begin{enumerate}[leftmargin=*]
    \item\label{cardinalsucc1} If $\ctC{A}{B}{C}{}{} \Reduces{w}{\param} N$, then $N$ is identical to $\ctC{A}{B}{C}{}$.
    \item\label{cardinalsucc2} If $\ctC{A}{B}{C}{P}{} \Reduces{w}{\param} N$, then $N$ is identical to $\ctC{A}{B}{C}{P_1}{}{}$ for some term $P_1$ such that $P\Reduces{w}{\param} P_1$.
    \item \label{cardinalsucc3} If $\ctC{A}{B}{C}{P}{Q}{}\Reduces{w}{\param} N$, then one of the following occurs:
\begin{enumerate}[leftmargin=*, label=(\alph*), ref=\alph*]
    \item \label{cardinalsucc3a} $N$ is identical to $\ctD{c}{A}{B}{C}{P_1}{}$ for some constant $c$ and some term $P_1$ such that $P\Reduces{w}{\param} P_1$; further $Q$ is identical to $c$.
    \item \label{cardinalsucc3b} $N$ is identical to $\ctC{A}{B}{C}{P_1}{Q_1}{}$ for some term $P_1,Q_1$ such that $P\Reduces{w}{\param} P_1$ and $Q\Reduces{w}{\param} Q_1$.
\end{enumerate}   
    \item \label{cardinalsucc4} If $\ctC{A}{B}{C}{P}{Q}{R}\Reduces{w}{\param} N$, then one of the following occurs:
\begin{enumerate}[leftmargin=*, label=(\alph*), ref=\alph*]
    \item \label{cardinalsucc4a} $N$ is identical to $P_1R_1Q_1$ for some terms $P_1,Q_1, R_1$ such that $P\Reduces{w}{\param} P_1$ and $Q\Reduces{w}{\param} Q_1$ and $R\Reduces{w}{\param} R_1$.
    \item \label{cardinalsucc4b} $N$ is identical to $\ctD{c}{A}{B}{C}{P_1}{R_1}{}$ for some constant $c$ and some terms $P_1,R_1$ such that $P\Reduces{w}{\param} P_1$ and $R\Reduces{w}{\param} R_1$; further $Q$ is identical to $c$.
    \item \label{cardinalsucc4c} $N$ is identical to $\ctC{A}{B}{C}{P_1}{Q_1}{R_1}$ for some terms $P_1,Q_1, R_1$ such that $P\Reduces{w}{\param} P_1$ and $Q\Reduces{w}{\param} Q_1$ and $R\Reduces{w}{\param} R_1$.    
\end{enumerate}   
\end{enumerate}
\end{prop}
\noindent There is an analogous proposition for Bluebird, but is simpler in that clauses (\ref{cardinalsucc3}\ref{cardinalsucc3a}) and (\ref{cardinalsucc4}{\ref{cardinalsucc4b}}) can be omitted. And of course to obtain the analogous proposition for Bluebird, one modifies  (\ref{cardinalsucc4}\ref{cardinalsucc4a}) appropriately.

\begin{proofdetail}
If one is interested in the proof of Proposition~\ref{prop:cardinalsucc}, we include it here. 

For (\ref{cardinalsucc1}) we show by induction on $s\geq 0$ that if  $\ctC{A}{B}{C}{}{} \Reduces{w,s}{\param} N$, then $N$ is identical to $\ctC{A}{B}{C}{}$:
\begin{itemize}[leftmargin=*]
    \item At stage $s=0$, if we add $\ctC{A}{B}{C}{}{} \Reduces{w,s}{\param} N$, then $N$ is identical to $\ctC{A}{B}{C}{}$.
    \item At stage $s+1$, we do not add any parallel reductions $\ctC{A}{B}{C}{}{}{} \Reduces{w,s+1}{\param} N$. Hence we are done by induction hypothesis.
\end{itemize}
For (\ref{cardinalsucc2}) we show by induction on $s\geq 0$ that if $\ctC{A}{B}{C}{P}{} \Reduces{w,s}{\param} N$, then $N$ is identical to $\ctC{A}{B}{C}{P_1}{}{}$ for some term $P_1$ such that $P\Reduces{w}{\param} P_1$:
\begin{itemize}[leftmargin=*]
    \item At stage $s=0$, we do not add any parallel reductions $\ctC{A}{B}{C}{P}{} \Reduces{w,s}{\param} N$.
    \item At stage $s+1$ with $s$ even, we do not add any parallel reductions $\ctC{A}{B}{C}{P}{} \Reduces{w,s+1}{\param} N$; hence we are done by induction hypothesis.
    \item At stage $s+1$ with $s$ odd, if we add a parallel reduction  $\ctC{A}{B}{C}{P}{} \Reduces{w,s+1}{\param} N$ then $N$ is identical to $OP_1$ where $\ctC{A}{B}{C}{}{} \Reduces{w,s}{\param} O$ and $P\Reduces{w,s}{\param} P_1$. By (\ref{cardinalsucc1}), we have that $O$ is $\ctC{A}{B}{C}{}{}$ and so we are done.
\end{itemize}
For (\ref{cardinalsucc3})  we show by induction on $s\geq 0$ that if 
$\ctC{A}{B}{C}{P}{Q}{}\Reduces{w,s}{\param} N$, then one of  (\ref{cardinalsucc3}\ref{cardinalsucc3a})-(\ref{cardinalsucc3}\ref{cardinalsucc3b}) happens.
\begin{itemize}[leftmargin=*]
    \item At stage $s=0$, we do not add any parallel reductions $\ctC{A}{B}{C}{P}{Q}{}\Reduces{w,s}{\param} N$.
    \item At stage $s+1$ with $s$ even, if we add a parallel reduction  $\ctC{A}{B}{C}{P}{Q} \Reduces{w,s+1}{\param} N$ then $N$ is identical to $\ctD{c}{A}{B}{C}{P_1}{}$ for some constant $c$ and some term $P_1$ such that $P\Reduces{w,s}{\param} P_1$; further $Q$ is identical to $c$. Then we are in case~(\ref{cardinalsucc3}\ref{cardinalsucc3a}).
    \item At stage $s+1$ with $s$ odd, if we add a parallel reduction $\ctC{A}{B}{C}{P}{Q} \Reduces{w,s+1}{\param} N$, then $N$ is identical to $OQ_1$ where $\ctC{A}{B}{C}{P}{} \Reduces{w,s}{\param} O$ and $Q\Reduces{w,s}{\param} Q_1$. By (\ref{cardinalsucc3}), $O$ is identical to $\ctC{A}{B}{C}{P_1}{}$ for some term $P_1$ such that  $P\Reduces{w}{\param} P_1$. Then we are in~(\ref{cardinalsucc3}\ref{cardinalsucc3b}).
\end{itemize}
For (\ref{cardinalsucc4}), we show by induction on   $s\geq 0$ that if $\ctC{A}{B}{C}{P}{Q}{R}\Reduces{w,s}{\param} N$, then one of (\ref{cardinalsucc4}\ref{cardinalsucc4a})-(\ref{cardinalsucc4}\ref{cardinalsucc4c}) happens. 
\begin{itemize}[leftmargin=*]
    \item At stage $s=0$, we do not add any parallel reductions $\ctC{A}{B}{C}{P}{Q}{R}\Reduces{w,s}{\param} N$. 
    \item At stage $s+1$ with $s$ even, if we add a parallel reduction $\ctC{A}{B}{C}{P}{Q}{R}\Reduces{w,s}{\param} N$, then $N$ is identical to $P_1 R_1 Q_1$ for some terms $P_1,Q_1, R_1$ such that $P\Reduces{w,s}{\param} P_1$ and $Q\Reduces{w,s}{\param} Q_1$ and $R\Reduces{w,s}{\param} R_1$. Then we are in case (\ref{cardinalsucc4}\ref{cardinalsucc4a}).
    \item At stage $s+1$ with $s$ odd, if we add a parallel reduction $\ctC{A}{B}{C}{P}{Q}{R}\Reduces{w,s}{\param} N$, then $N$ is identical to $OR_1$ where  $\ctC{A}{B}{C}{P}{Q}{}\Reduces{w,s}{\param} O$ and  $R\Reduces{w,s}{\param} R_1$. By (\ref{cardinalsucc3}), there are then two subcases:
\begin{itemize}[leftmargin=*]
\item If $O$ is identical to $\ctD{c}{A}{B}{C}{P_1}{}$ for some constant $c$ and some term $P_1$ such that $P\Reduces{w}{\param} P_1$, and $Q$ is identical to $c$, then we are in case~(\ref{cardinalsucc4}\ref{cardinalsucc4b}).
\item If $O$ is identical to $\ctC{A}{B}{C}{P_1}{Q_1}{}$ for some term $P_1,Q_1$ such that $P\Reduces{w}{\param} P_1$ and $Q\Reduces{w}{\param} Q_1$, then we are in case~(\ref{cardinalsucc4}\ref{cardinalsucc4c}).
\end{itemize}
\end{itemize}

\end{proofdetail}

The following lemma is important because (\ref{lem:parallel:2})-(\ref{lem:parallel:3}) imply that weak reduction and parallel reduction have the same transitive closure (cf. \cite[p. 120 equations (1)-(3)]{Takahashi1995-kp}):
\begin{lem}\label{lem:parallel}\hfill
\begin{enumerate}[leftmargin=*]
    \item\label{lem:parallel:1} $P\Reduces{w}{\param} P$ 
    \item\label{lem:parallel:2} If $P\reduces{w}{\param} Q$ then $P\Reduces{w}{\param} Q$
    \item\label{lem:parallel:3} If $P\Reduces{w}{\param} Q$ then $P\transreduces{w}{\param} Q$ 
    \item\label{lem:parallel:4} If $P\Reduces{w}{\param} P^{\prime}$ and $Q\Reduces{w}{\param} Q^{\prime}$ then $P[x:=Q]\Reduces{w}{\param} P^{\prime}[x:=Q^{\prime}]$
\end{enumerate}
\end{lem}
\begin{proof}
For (\ref{lem:parallel:1}), this follows from an easy induction on complexity of $P$ from Definition~\ref{defn:takahashi}(\ref{defn:takahashi1}),~(\ref{defn:takahashi3}).

For (\ref{lem:parallel:2}), simply use (\ref{lem:parallel:1}) and Definition~\ref{defn:takahashi}(\ref{defn:takahashi2}) to handle the case when the reduction happens at the top level, and then use (\ref{lem:parallel:1}) and Definition~\ref{defn:takahashi}(\ref{defn:takahashi3}) to handle when the reduction happens embedded inside applications. 

For (\ref{lem:parallel:3}) use induction on $s\geq 0$ to show that $P\Reduces{w,s}{\param} Q$ implies $P\transreduces{w}{\param} Q$. 

\begin{proofdetail}
In more detail:
\begin{itemize}[leftmargin=*]
    \item The stage $s=0$ case follows since $\transreduces{w}{\param}$ is reflexive.
    \item At stage $s+1$ with $s$ even, if we add a parallel reduction $\ctW{A}{B}{P}{Q}\Reduces{w,s+1}{\param} P^{\prime}Q^{\prime}Q^{\prime}$ where $P\Reduces{w,s}{\param} P^{\prime}$ and $Q\Reduces{w,s}{\param} Q^{\prime}$, then by induction hypothesis $P\transreduces{w}{\param} P^{\prime}$ and $Q\transreduces{w}{\param} Q^{\prime}$, and so $\ctW{A}{B}{P}{Q}\reduces{w}{\param} PQQ \transreduces{w}{\param} P^{\prime}Q^{\prime}Q^{\prime}$; and the other weak reductions are similar.
    \item At stage $s+1$ with $s$ odd, if we add a parallel reduction $PQ\Reduces{w,s+1}{\param} P^{\prime}Q^{\prime}$ where $P\Reduces{w,s}{\param} P^{\prime}$ and $Q\Reduces{w,s}{\param} Q^{\prime}$, then by induction hypothesis $P\transreduces{w}{\param} P^{\prime}$ and $Q\transreduces{w}{\param} Q^{\prime}$, and so $PQ \transreduces{w}{\param} P^{\prime} Q^{\prime}$.     
\end{itemize}
\end{proofdetail}

For (\ref{lem:parallel:4}) we use an induction on $s\geq 0$ to show that $P\Reduces{w,s}{\param} P^{\prime}$ implies that for all  $Q\Reduces{w}{\param} Q^{\prime}$ we have $P[x:=Q]\Reduces{w}{\param} P^{\prime}[x:=Q^{\prime}]$.
\begin{proofdetail}
\begin{itemize}[leftmargin=*]
    \item The stage $s=0$ case follows since at this stage the only parallel reduction we add in which variables appear on either side is the parallel reduction $x\Reduces{w,s}{\param} x$ for a variable $x$; and then $P[x:=Q]\Reduces{w}{\param} P^{\prime}[x:=Q^{\prime}]$ is just identical to $Q\Reduces{w}{\param} Q^{\prime}$.
    \item At stage $s+1$ with $s$ even, if we add a parallel reduction $\ctW{A}{B}{M}{N}\Reduces{w,s+1}{\param} M^{\prime}N^{\prime}N^{\prime}$ where $M\Reduces{w,s}{\param} M^{\prime}$ and $N\Reduces{w,s}{\param} N^{\prime}$, then by induction hypothesis, if $Q\Reduces{w}{\param} Q^{\prime}$ then both $M[x:=Q]\Reduces{w}{\param} M^{\prime}[x:=Q^{\prime}]$ and $N[x:=Q]\Reduces{w}{\param} N^{\prime}[x:=Q^{\prime}]$; and then $\ctW{A}{B}{M[x:=Q]}{N[x:=Q]} \Reduces{w}{\param} M^{\prime}[x:=Q] N^{\prime}[x:=Q] N^{\prime}[x:=Q]$; and similarly for the other weak reductions.
    \item At stage $s+1$ with $s$ odd, if we add a parallel reduction $MN\Reduces{w,s+1}{\param} M^{\prime}N^{\prime}$ where $M\Reduces{w,s}{\param} M^{\prime}$ and $N\Reduces{w,s}{\param} N^{\prime}$, then by induction hypothesis if $Q\Reduces{w}{\param} Q^{\prime}$ then both $M[x:=Q]\Reduces{w}{\param} M^{\prime}[x:=Q^{\prime}]$ and $N[x:=Q]\Reduces{w}{\param} N^{\prime}[x:=Q^{\prime}]$; and then $M[x:=Q] N[x:=Q] \Reduces{w}{\param}M^{\prime}[x:=Q] N^{\prime}[x:=Q]$.
\end{itemize}
\end{proofdetail}
\end{proof}

Takahashi's concept complement development is, in our $\cltup$, the following (cf. 
\cite[p. 121]{Takahashi1995-kp}):
\begin{defi}[The complete development]\label{defn:comp}
The \emph{complete development} $\typd{M^{\ast}}{B}$ of a term $\typd{M}{B}$ of $\cltup$ is defined by induction on complexity of term as follows:
\begin{enumerate}[leftmargin=*]
    \item \label{defn:compdev:1} If $\typd{M}{B}$ is a variable, constant, or combinatory term, then $\typd{M^{\ast}}{B}$ is $\typd{M}{B}$.
    \item \label{defn:compdev:2} If $\typd{M}{B}$ is an application $\typd{P Q}{B}$ which is not a redex, then we define  $\typd{M^{\ast}}{B}$ to be $\typd{P^{\ast} Q^{\ast}}{B}$.
    \item \label{defn:compdev:3} If $\typd{M}{B}$ is a redex, then we define $\typd{M^{\ast}}{B}$ as follows:
\begin{equation*}
M^{\ast}=\begin{cases}
P^{\ast}      & \text{if $M$ is $\ctK{A}{B}{P}{Q}$}, \\
P^{\ast}R^{\ast}Q^{\ast}    & \text{if $M$ is $\ctC{A}{B}{C}{P}{Q}{R}$}, \\
P^{\ast}R^{\ast}c    & \text{if $M$ is $\ctD{c}{A}{B}{C}{P}{R}$}, \\
P^{\ast}Q^{\ast}Q^{\ast} & \text{if $M$ is $\ctW{A}{B}{P}{Q}$}, \\
P^{\ast}(Q^{\ast} R^{\ast}) & \text{if $M$ is $\ctB{A}{B}{C}{P}{Q}{R}$} \\
\ctD{c}{A}{B}{C}{P^{\ast}}{}{}  & \text{if $M$ is $\ctC{A}{B}{C}{P}{c}{}$}, \\
\end{cases}
\end{equation*}
\end{enumerate}
\end{defi}

The Takahashi proof of Church-Rosser then goes through the following proposition (cf. \cite[p. 121]{Takahashi1995-kp}):
\begin{prop}\label{prop:Takahashiprop}

If $M\Reduces{w}{\param} N$ then $N\Reduces{w}{\param} M^{\ast}$.

\end{prop}
\begin{proof}

This is by induction on complexity of $M$.

First suppose that $M$ is a variable, constant, or combinatory term. Suppose $M\Reduces{w}{\param} N$. Then Proposition~\ref{prop:helperfortaka}(\ref{prop:helperfortaka:1}) we have that $N$ is $M$. We are then done since by Definition~\ref{defn:comp}(\ref{defn:compdev:1}), we have that $M^{\ast}$ is also $M$.

Second suppose that $M$ is an application $PQ$ which is not a redex. Suppose $M\Reduces{w}{\param} N$. Then Proposition~\ref{prop:helperfortaka}(\ref{prop:helperfortaka:2}), we have that $N$ is $P_1 Q_1$ where $P\Reduces{w}{\param} P_1$ and $Q\Reduces{w}{\param} Q_1$. By induction hypothesis
$P_1\Reduces{w}{\param} P^{\ast}$ and $Q_1\Reduces{w}{\param} Q^{\ast}$. Then by Definition~\ref{defn:takahashi}(\ref{defn:takahashi3}), we have $P_1Q_1\Reduces{w}{\param} P^{\ast} Q^{\ast}$, which by Definition~\ref{defn:comp}(\ref{defn:compdev:2}) is equal to $(PQ)^{\ast}$.

 Third suppose that $M$ is a redex. 
 
 First consider the Warbler case where $M$ is $\ctW{A}{B}{P}{Q}$.  Suppose $M\Reduces{w}{\param} N$. By Proposition~\ref{prop:stunted}(\ref{prop:stunted:3}), there are two cases to consider:
 \begin{itemize}[leftmargin=*]
 \item First suppose $N$ is $P_1 Q_1 Q_1$ where $P\Reduces{w}{\param} P_1$ and $Q\Reduces{w}{\param} Q_1$. By induction hypothesis  
$P_1\Reduces{w}{\param} P^{\ast}$ and $Q_1\Reduces{w}{\param} Q^{\ast}$. Then by two applications of Definition~\ref{defn:takahashi}(\ref{defn:takahashi3}), we have $P_1Q_1Q_1\Reduces{w}{\param} P^{\ast} Q^{\ast} Q^{\ast}$, which by Definition~\ref{defn:comp}(\ref{defn:compdev:3}) is equal to $(\ctW{A}{B}{P}{Q})^{\ast}$.
\item Second suppose $N$ is $\ctW{A}{B}{P_1}{Q_1}$ where $P\Reduces{w}{\param} P_1$ and $Q\Reduces{w}{\param} Q_1$. By induction hypothesis 
$P_1\Reduces{w}{\param} P^{\ast}$ and $Q_1\Reduces{w}{\param} Q^{\ast}$. By an application of Definition~\ref{defn:takahashi}(\ref{defn:takahashi2}), we have $\ctW{A}{B}{P_1}{Q_1}\Reduces{w}{\param}  P^{\ast} Q^{\ast} Q^{\ast}$ which by Definition~\ref{defn:comp}(\ref{defn:compdev:3}) is equal to $(\ctW{A}{B}{P}{Q})^{\ast}$.
\end{itemize}
The proofs for the other ``two input'' combinatory terms of Kestrel and Dardinal are entirely identical.

Second consider the Cardinal case. There are two Cardinal redexes and so two subcases. 

First consider subcase where $M$ is the redex $\ctC{A}{B}{C}{P}{c}{}$. Suppose $M\Reduces{w}{\param} N$. By Proposition~\ref{prop:cardinalsucc}(\ref{cardinalsucc3}) (with $Q$ in that proposition set identical to $c$), there are two subcases to consider:
\begin{itemize}[leftmargin=*]
    \item  $N$ is identical to $\ctD{c}{A}{B}{C}{P_1}{}$ for some term $P_1$ such that $P\Reduces{w}{\param} P_1$. Then by induction hypothesis, $P_1\Reduces{w}{\param} P^{\ast}$. Then by an application of Definition~\ref{defn:takahashi}(\ref{defn:takahashi1}),(\ref{defn:takahashi3}), we have $\ctD{c}{A}{B}{C}{P_1}{} \Reduces{w}{\param} \ctD{c}{A}{B}{C}{P^{\ast}}{}$, and the latter is is equal to $(\ctC{A}{B}{C}{P}{c}{})^{\ast}$ by Definition~\ref{defn:comp}(\ref{defn:compdev:3}).
    \item $N$ is identical to $\ctC{A}{B}{C}{P_1}{Q_1}{}$ for some $P_1,Q_1$ such that $P\Reduces{w}{\param} P_1$ and $c\Reduces{w}{\param} Q_1$. By Proposition~\ref{prop:helperfortaka}(\ref{prop:helperfortaka:1}) we have that $Q_1$ is $c$. Further, by induction hypothesis, $P_1\Reduces{w}{\param} P^{\ast}$. Then by Definition~\ref{defn:takahashi}(\ref{defn:takahashi2}) we have $\ctC{A}{B}{C}{P_1}{c}{}\Reduces{w}{\param}  \ctD{c}{A}{B}{C}{P^{\ast}}{}$, and the latter is is equal to $(\ctC{A}{B}{C}{P}{c}{})^{\ast}$ by Definition~\ref{defn:comp}(\ref{defn:compdev:3}).
\end{itemize}

Second consider the subcase where $M$ is the redex $\ctC{A}{B}{C}{P}{Q}{R}$. Suppose $M\Reduces{w}{\param} N$. By Proposition~\ref{prop:cardinalsucc}(\ref{cardinalsucc4}), there are three subcases to consider:
\begin{itemize}[leftmargin=*]
    \item $N$ is identical to $P_1R_1Q_1$ for some $P_1,Q_1, R_1$ such that $P\Reduces{w}{\param} P_1$ and $Q\Reduces{w}{\param} Q_1$ and $R\Reduces{w}{\param} R_1$.  Then by induction hypothesis, $P_1\Reduces{w}{\param} P^{\ast}$ and  $Q_1\Reduces{w}{\param} Q^{\ast}$ and  $R_1\Reduces{w}{\param} R^{\ast}$. Then by two applications of Definition~\ref{defn:takahashi}(\ref{defn:takahashi3}), we have $P_1R_1Q_1 \Reduces{w}{\param} P^{\ast} R^{\ast} Q^{\ast}$, and the latter is is equal to $(\ctC{A}{B}{C}{P}{Q}{R})^{\ast}$ by Definition~\ref{defn:comp}(\ref{defn:compdev:3}).
    \item  $N$ is identical to $\ctD{c}{A}{B}{C}{P_1}{R_1}{}$ for constant $c$ and some terms $P_1,R_1$ such that $P\Reduces{w}{\param} P_1$ and  $R\Reduces{w}{\param} R_1$; further $Q$ is identical to $c$.  Then by induction hypothesis, $P_1\Reduces{w}{\param} P^{\ast}$ and $R_1\Reduces{w}{\param} R^{\ast}$. By Definition~\ref{defn:comp}(\ref{defn:compdev:1}), $Q^{\ast}$ is also identical to $c$.
    By Definition~\ref{defn:takahashi}(\ref{defn:takahashi2}) we have $\ctD{c}{A}{B}{C}{P_1}{R_1}{}\Reduces{w}{\param} P^{\ast}R^{\ast} c$, which is the same term as $P^{\ast}R^{\ast} Q^{\ast}$, and the latter is is equal to $(\ctC{A}{B}{C}{P}{Q}{R})^{\ast}$ by Definition~\ref{defn:comp}(\ref{defn:compdev:3}).   
    \item  $N$ is identical to $\ctC{A}{B}{C}{P_1}{Q_1}{R_1}$ for some $P_1,Q_1, R_1$ such that $P\Reduces{w}{\param} P_1$ and $Q\Reduces{w}{\param} Q_1$ and $R\Reduces{w}{\param} R_1$.  Then by induction hypothesis, $P_1\Reduces{w}{\param} P^{\ast}$ and  $Q_1\Reduces{w}{\param} Q^{\ast}$ and  $R_1\Reduces{w}{\param} R^{\ast}$.  Then by Definition~\ref{defn:takahashi}(\ref{defn:takahashi2}), we have $\ctC{A}{B}{C}{P_1}{Q_1}{R_1}\Reduces{w}{\param} P^{\ast}R^{\ast}Q^{\ast}$, and the latter is is equal to $(\ctC{A}{B}{C}{P}{Q}{R})^{\ast}$ by Definition~\ref{defn:comp}(\ref{defn:compdev:3}).
\end{itemize}

The other ``three input'' case of Bluebird is similar, but simpler.
\end{proof}

\begin{thm}[Church-Rosser for weak reduction in $\cltup$]\label{thm:crforclt}
Suppose that $M_1, M_2, M_3$ are terms of $\cltup$ such that  $M_1\transreduces{w}{\param} M_2$ and $M_1\transreduces{w}{\param} M_3$. Then there is a term $M_4$ of $\cltup$ such that $M_2 \transreduces{w}{\param} M_4$ and $M_3 \transreduces{w}{\param} M_4$. 
\end{thm}
\begin{proof}
The previous theorem implies that if $M\Reduces{w}{\param} N_0$ and $M\Reduces{w}{\param}  N_1$, then $N_0\Reduces{w}{\param} M^{\ast}$ and $N_1\Reduces{w}{\param} M^{\ast}$. By a classic diagram chase argument, this implies that the transitive closure of~$\Reduces{w}{\param}$ has the Church-Rosser property (cf. \cite[Lemma 3.2.2 p. 59]{Barendregt1981-ai}). Then we are done by Proposition~\ref{lem:parallel}(\ref{lem:parallel:2})-(\ref{lem:parallel:3}), which implies that the transitive closure of $\Reduces{w}{\param}$ is the same as~$\transreduces{w}{\param}$.
\end{proof}

A traditional proof then also gives (cf. \cite[Theorem 3.1.12 p. 54]{Barendregt1981-ai}):
\begin{cor}[Church-Rosser and weak equality in $\cltup$]\label{cor:crforequals}
Suppose that $M_1, M_2$ are terms of $\cltup$ such that  $M_1\equals{w}{\param} M_2$.

Then there is a term $M_3$ of $\cltup$ such that $M_1\transreduces{w}{\param} M_3$ and $M_2\transreduces{w}{\param} M_3$. 
\end{cor}

From this we can derive:

\begin{cor}[Conservation of $\clt{\omega}$ over $\cltup$]\label{cor:cltcon}
Suppose that $\typd{M,N}{A}$ are terms of $\cltup$. 

Then $\cltup\vdash_w M=N$ iff $\clt{\omega}\vdash_w M=N$.
\end{cor}
\begin{proof}
The forward direction is trivial since any $\reduces{w}{\param}$ reduction is a $\reduces{w}{\omega}$ reduction. 

Conversely, suppose that $\clt{\omega}\vdash_w M=N$. Then by definition $M\equals{w}{\omega} N$. By Corollary~\ref{cor:crforequals},  there is a term $L$ of $\clt{\omega}$ such that $M\transreduces{w}{\omega} L$ and $N\transreduces{w}{\omega} L$. But since $\typd{M,N}{A}$ are terms of $\cltup$, by Lemma~\ref{lem:preservescltup} so also we have that $L$ and everything else in the two $\transreduces{w}{\omega}$-chains are terms of $\cltup$. Hence we also have $M\equals{w}{\param} N$, which by definition means $\cltup\vdash_w M=N$.
\end{proof}

%% file: 06-pure-combinatory/03-simulate-abstraction.tex
\subsection{Simulating abstraction}\label{sec:simabs}

Despite its lack of a primitive binding apparatus, combinatory logic famously allows one to simulate core aspects of lambda abstraction. The usual definition of this goes through the combinatory terms $\mathsf{SKI}$, but the proof of Theorem~\ref{thm:bigcombo} suggests the following definition which deploys $\mathsf{BCDKW}$. While there is some suggestion that Curry did something like this in his early work, this definition is to my knowledge new.\footnote{See \cite[\S{5}]{Cardone2006-qk}. In Curry \cite[p. 238]{Curry1958-yi} one finds a description of why  $\mathsf{BCKW}$ should suffice, but the more formal discussion  in Curry \cite[p. 190]{Curry1958-yi} (list at the bottom) does not include Warbler, and does not contain the case breaks I have used.}

\begin{defi}[Combinatory abstraction]\label{defn:cbd}
Suppose $A,B$ are types and $B$ is a regular type and $\typd{v}{A}$ is a variable of $\cltup$ and $\typd{M}{B}$ is a term of $\cltup$. Then, by induction on complexity of $\typd{M}{B}$, we define term $\cbd{v}{A}{M}$ of type $A\rightarrow B$ of $\cltup$ as follows:
\begin{enumerate}[leftmargin=*]
\item\label{defn:cbd:1} If $\typd{v}{A}$ does not appear in $\typd{M}{B}$, then we define  $\cbd{v}{A}{M}$ to be $\ctK{B}{A}{M}{}$.
\item\label{defn:cbd:2} If $\typd{v}{A}$ does appear in $\typd{M}{B}$ and $\typd{M}{B}$ is $\typd{v}{A}$, then we define $\cbd{v}{A}{M}$ to be $\ctI{A}{}$.
\item\label{defn:cbd:3} If $\typd{v}{A}$ does appear in $\typd{M}{B}$ and $\typd{M}{B}$ is $\typd{M_0 M_1}{B}$ where $\typd{M_0}{C\rightarrow B}$ and $\typd{M_1}{C}$, then we define
\begin{enumerate}[leftmargin=*, label=(\alph*), ref=\alph*]
\item\label{defn:cbd:3.1} If $C$ is a state type then
\begin{equation*}
\cbd{v}{A}{M_0 M_1}=\begin{cases}
\ctW{A}{B}{\big(\cbd{v}{A}{M_0}\big)}      & \text{if $\typd{M_1}{C}$ is $\typd{v}{A}$}, \\
\ctC{A}{C}{B}{\big(\cbd{v}{A}{M_0} \big) }{M_1}{}{}      & \text{if $\typd{M_1}{C}$ is a variable but not $\typd{v}{A}$}, \\
\ctD{c}{A}{C}{B}{\big(\cbd{v}{A}{M_0} \big) }{}{}{}      & \text{if $\typd{M_1}{C}$ is the constant $\typd{c}{C}$}, \\
\end{cases}
\end{equation*}
    \item\label{defn:cbd:3.2} If $C$ is a regular type, then
\begin{equation*}    
\cbd{v}{A}{M_0 M_1} \; = \; \ctS{C}{B}{A}{\big(\cbd{v}{A}{M_0}\big)}{\big(\cbd{v}{A}{M_1}\big)}{}
\end{equation*}
\end{enumerate}
\end{enumerate}
\end{defi}

Regarding Definition~\ref{defn:cbd}(\ref{defn:cbd:3}\ref{defn:cbd:3.1}), note that the hypothesis that $C$ is a state type and the second case break hypothesis that $\typd{M_1}{C}$ is a variable but not $\typd{v}{A}$ has the consequence that either $A,C$ are distinct types, or $A,C$ are identical types and $\param(A)=\param(C)>1$. This, in conjunction with $B$ being regular implies that Cardinal $\ctC{A}{C}{B}{}{}{}$ is a term of $\cltup$ by Definition~\ref{defn:typedcombo2}.

Definition~\ref{defn:cbd} is a definition by induction on complexity of $\typd{M}{B}$. In particular, Definition~\ref{defn:cbd}(\ref{defn:cbd:1})-(\ref{defn:cbd:2}) cover the cases of variables, constants, and combinatory terms; and supposing it has been defined for  $\typd{M_0}{C\rightarrow B}$ and $\typd{M_1}{C}$, we define it for $\typd{M_0 M_1}{B}$ by breaking into cases according to Definition~\ref{defn:cbd}(\ref{defn:cbd:1}), (\ref{defn:cbd:3}), where we appeal to induction hypothesis only in the case of Definition~\ref{defn:cbd}(\ref{defn:cbd:3}).

The following proposition is a technical one. One can see it as simultaneously accomplishing two things: showing an elementary instance of the analogue of $\beta$-reduction is available, and then identifying the variables appearing in a combinatory abstract (cf. \cite[p. 27]{Hindley2008-vw}).  
\begin{prop}[Elementary instance of analogue of $\beta$-reduction; variables appearing in an combinatory abstract]\label{prop:clbetapre}
Suppose $A,B$ are types and $B$ is a regular type and $\typd{v}{A}$ is a variable of $\cltup$ and $\typd{M}{B}$ is a term of $\cltup$. Then
\begin{enumerate}[leftmargin=*]
    \item\label{prop:clbetapre:1}  $(\cbd{v}{A}{M})v \transreduces{w}{\param} M$.
    \item\label{prop:clbetapre:2} The variables appearing in  $\cbd{v}{A}{M}$ are precisely those appearing in $\typd{M}{B}$ minus $\typd{v}{A}$.
    In particular, $v$ does not appear in $\cbd{v}{A}{M}$.
\end{enumerate}
\end{prop}
\begin{proof}
We argue by induction on complexity of $\typd{M}{B}$. 

The most interesting case is the induction step where $v$ appears in $M$ and $\typd{M}{B}$ is the application $\typd{M_0 M_1}{B}$ with $\typd{M_0}{C\rightarrow B}$ and $\typd{M_1}{C}$ and $C$ being a state type. There are three subcases, corresponding to the three subcases of Definition~\ref{defn:cbd}(\ref{defn:cbd:3}\ref{defn:cbd:3.1}). Since they are similar we only do one of them.

Suppose that $\typd{M_1}{C}$ is $\typd{v}{A}$. Then $\cbd{v}{A}{M_0 M_1}$ is $\ctW{A}{B}{\big(\cbd{v}{A}{M_0}\big)}{}$. By induction hypothesis for $M_0$, the variables appearing in $\ctW{A}{B}{\big(\cbd{v}{A}{M_0}\big)}{}$ are precisely those appearing in $M_0$ minus $\typd{v}{A}$; and due to $\typd{M_1}{C}$ being $\typd{v}{A}$, this is equal to the those appearing in $M_0M_1$ minus $\typd{v}{A}$. Further, by Definition~\ref{defn:weakreduction} one has $\ctW{A}{B}{\big(\cbd{v}{A}{M_0}\big)}{v} \reduces{w}{\param} \big(\cbd{v}{A}{M_0}\big) vv \transreduces{w}{\param} M_0 v$, where the second weak reduction is by induction hypothesis for $M_0$.
\end{proof}

\begin{proofdetail}

\begin{proof}

If one wants more detail, we give the full proof by induction on complexity of $\typd{M}{B}$ here.
    
First suppose that $\typd{M}{B}$ is a variable. 

Suppose $\typd{M}{B}$ is $\typd{v}{A}$ itself, so that $B$ is $A$. Then by Definition~\ref{defn:cbd}(\ref{defn:cbd:2}) we have that  $\cbd{v}{A}{M}$ is $\ctI{A}{}$ which does not have any variables since it is a closed term of $\cltup$; likewise, the set of variables appearing in $\typd{M}{B}$ minus $\typd{v}{A}$ is empty. Further $\ctI{A}{} v \transreduces{w}{\param} v$ by Proposition~\ref{prop:clrecoveryidentity}.

Suppose $\typd{M}{B}$ is $\typd{u}{B}$, which is distinct from  $\typd{v}{A}$. Then by Definition~\ref{defn:cbd}(\ref{defn:cbd:1}), we have that $\cbd{v}{A}{M}$ is $\ctK{B}{A}{M}{}$, in which only variable $\typd{u}{B}$ appears; likewise the set of variables appearing in $\typd{M}{B}$ minus $\typd{v}{A}$ is precisely the variable $\typd{u}{B}$. Further, by Definition~\ref{defn:weakreduction}, we have $\ctK{B}{A}{M}{v} \reduces{w}{\param} M$.

The argument for constants and combinatory terms is exactly parallel to the previous paragraph, except no variables appear in $\typd{M}{B}$ in these cases, and likewise no variables appear in  $\cbd{v}{A}{M}$, which is $\ctK{B}{A}{M}{}$.

Suppose that $\typd{M}{B}$ is $\typd{M_0 M_1}{B}$ where $\typd{M_0}{C\rightarrow B}$ and $\typd{M_1}{C}$. The case where $v$ does not appear in $\typd{M}{B}$  is exactly like the previous paragraph. Henceforth assume that $v$ does appear in $\typd{M}{B}$.

First suppose that $C$ is a state type. There are three subcases, corresponding to the three subcases of Definition~\ref{defn:cbd}(\ref{defn:cbd:3}\ref{defn:cbd:3.1}). Since they are similar we only do one of them.

Suppose that $\typd{M_1}{C}$ is $\typd{v}{A}$. Then $\cbd{v}{A}{M_0 M_1}$ is $\ctW{A}{B}{\big(\cbd{v}{A}{M_0}\big)}{}$. By induction hypothesis for $M_0$, the variables appearing in $\ctW{A}{B}{\big(\cbd{v}{A}{M_0}\big)}{}$ are precisely those appearing in $M_0$ minus $\typd{v}{A}$; and due to $\typd{M_1}{C}$ being $\typd{v}{A}$, this is equal to the those appearing in $M_0M_1$ minus $\typd{v}{A}$. Further, by Definition~\ref{defn:weakreduction} one has $\ctW{A}{B}{\big(\cbd{v}{A}{M_0}\big)}{v} \reduces{w}{\param} \big(\cbd{v}{A}{M_0}\big) vv \transreduces{w}{\param} M_0 v$, where the second weak reduction is by induction hypothesis for $M_0$.

The case where $C$ is a regular type is similar to the previous paragraph but goes through Starling.
\end{proof}

\end{proofdetail}

The following technical proposition is the combinatory analogue of the lambda calculus identity $\big(\lbd{v}{A}{M}\big)[u:=N]\equiv \lbd{v}{A}{\big(M[u:=N]\big)}$, when variables $\typd{v}{A}, \typd{u}{B}$ are distinct and $v$ does not occur free in $N$. It is more complicated in combinatory logic simply because the combinatory abstracts are defined via a large number of case breaks. Further, unlike the usual $\mathsf{SKI}$-combinatory logic, we get a weak equality in our $\mathsf{BCDKW}$-combinatory logic rather than a literal identity of terms, unless we restrict to a signature with no constants of state type (cf. \cite[Lemma 2.28 (c) p. 29]{Hindley2008-vw}). We use this proposition in the proof of Proposition~\ref{prop:missinglink}, which in turn is used in the proof of Theorem~\ref{thm:frommltuptoclt}.

\begin{prop}[Substituting in a combinatory abstract]\label{prop:verydiff}
Suppose $A,B$ are types and $B$ is a regular type and $\typd{v}{A}$ is a variable of $\cltup$ and $\typd{M}{B}$ is a term of $\cltup$. 

Suppose $D$ is a type and $\typd{u}{D}$ is a variable of $\cltup$ distinct from $\typd{v}{A}$, and suppose that $\typd{N}{D}$ is a term of $\cltup$ in which $\typd{v}{A}$ does not appear.

Then:
\begin{enumerate}[leftmargin=*]
    \item\label{prop:verydiff:1} In a signature with constants of state type, $\big(\cbd{v}{A}{M}\big) [u:=N] \equals{w}{\param} \cbd{v}{A}{\big(M[u:=N]\big)}$.
    \item\label{prop:verydiff:2} In a signature with no constants of state type, $\big(\cbd{v}{A}{M}\big) [u:=N] \equiv \cbd{v}{A}{\big(M[u:=N]\big)}$.
\end{enumerate}
\end{prop}
\begin{proof}

For (\ref{prop:verydiff:1}), this proof is by induction on complexity of $\typd{M}{B}$, with a universal quantifier over $\typd{u}{D}$, $\typd{N}{D}$.

Suppose that $\typd{v}{A}$ does not appear in $\typd{M}{B}$. Since $\typd{v}{A}$ does not appear in  $\typd{N}{D}$, we have that $\typd{v}{A}$ does not appear in $M[u:=N]$. Then we have $(\cbd{v}{A}{M})[u:=N] \equiv \big(\ctK{B}{A}{M}{}\big) [u:=N] \equiv \ctK{B}{A}{\big(M[u:=N]\big)}{} \equiv  \cbd{v}{A}{\big(M[u:=N]\big)}$, wherein the first and the last are by two applications of Definition~\ref{defn:cbd}(\ref{defn:cbd:1}).

Suppose that $\typd{v}{A}$ does appear in $\typd{M}{B}$ and that $\typd{M}{B}$ is $\typd{v}{A}$. Then by Proposition~\ref{prop:clbetapre}(\ref{prop:clbetapre:2}), no variables appear in $\cbd{v}{A}{M}$, and so $\big(\cbd{v}{A}{M}\big) [u:=N]$ is $\cbd{v}{A}{M}$. Further, since $\typd{u}{D}$ is distinct from $\typd{v}{A}$, we have that $\cbd{v}{A}{\big(M[u:=N]\big)}$ is also $\cbd{v}{A}{M}$. By Definition~\ref{defn:cbd}(\ref{defn:cbd:2}), both terms are identical to $\ctI{A}{}$.

Suppose for the remainder of the proof that $\typd{v}{A}$ does appear in $\typd{M}{B}$ and $\typd{M}{B}$ is $\typd{M_0 M_1}{B}$ where $\typd{M_0}{C\rightarrow B}$ and $\typd{M_1}{C}$. Since $\typd{v}{A},\typd{u}{D}$ are distinct variables, $\typd{v}{A}$ appears in $M[u:=N]$; and the term $M[u:=N]$ is $M_0[u:=N] M_1[u:=N]$. By induction hypothesis, we have that the result holds for $M_0$ and that it holds for $M_1$ when $C$ is of regular type.

As a first case, suppose that $C$ is a state type. There are then four subcases. 

First suppose that $\typd{M_1}{C}$ is $\typd{v}{A}$. Since $\typd{v}{A},\typd{u}{D}$ are distinct variables, we also have that $M_1[u:=N]$ is $\typd{v}{A}$. Then we have the following:
\begin{align*}
&\big(\cbd{v}{A}{M_0 M_1}\big) [u:=N] \equiv \big(\ctW{A}{B}{\big(\cbd{v}{A}{M_0}\big)} \big)[u:=N]\\
& \equiv \ctW{A}{B}{\big(\big(\cbd{v}{A}{M_0}\big)[u:=N]\big)}  \equals{w}{\param} \ctW{A}{B}{\big(\cbd{v}{A}{\big(M_0[u:=N]\big)}\big)} \\
& \equiv \cbd{v}{A}{\big(M_0[u:=N]v\big)} \equiv \cbd{v}{A}{\big(M_0[u:=N]M_1[u:=N]\big)}\\
&\equiv \cbd{v}{A}{\big((M_0M_1)[u:=N]\big)}
\end{align*}
In this, the first and third-to-last are by two applications of the Warbler case of Definition~\ref{defn:cbd}(\ref{defn:cbd:3}\ref{defn:cbd:3.1}), and the one with $\equals{w}{\param}$ follows from induction hypothesis for $\typd{M_0}{C\rightarrow B}$.

Second suppose that $\typd{M_1}{C}$ is $\typd{u}{D}$. Then $D$ is $C$, and hence $D$ is a state type, and hence by Proposition~\ref{prop:state-free:cl} one has that $\typd{N}{D}$ is a variable or a constant. Since $\typd{v}{A}$ does not appear in $\typd{N}{D}$, we have that $\typd{N}{D}$ is a variable distinct from $\typd{v}{A}$ or $\typd{N}{D}$ is a constant $\typd{c}{C}$. 
\begin{itemize}[leftmargin=*]
\item First consider the subsubcase where $\typd{N}{D}$ is a variable but not $\typd{v}{A}$. Then:
\begin{align*}
& \big(\cbd{v}{A}{M_0 M_1}\big) [u:=N] \equiv \big(\ctC{A}{C}{B}{\big(\cbd{v}{A}{M_0} \big) }{M_1}{}{}  \big)[u:=N]\\
& \equiv  \ctC{A}{C}{B}{\big(\big(\cbd{v}{A}{M_0}\big)[u:=N]\big)}{M_1[u:=N]}{}{}   \equals{w}{\param}  \ctC{A}{C}{B}{\big(\cbd{v}{A}{M_0[u:=N]} \big) }{M_1[u:=N]}{}{}  \\
& \equiv  \ctC{A}{C}{B}{\big(\cbd{v}{A}{M_0[u:=N]} \big) }{N}{}{}   \equiv \cbd{v}{A}{\big(M_0[u:=N] N\big)} \\
& \equiv \cbd{v}{A}{\big(M_0[u:=N]M_1[u:=N]\big)}\equiv \cbd{v}{A}{\big((M_0M_1)[u:=N]\big)}
\end{align*}
In this, the first and third-to-last are by two applications of the Cardinal case of Definition~\ref{defn:cbd}(\ref{defn:cbd:3}\ref{defn:cbd:3.1}), and the one with $\equals{w}{\param}$ follows from induction hypothesis for $\typd{M_0}{C\rightarrow B}$.
\item Second consider the subsubcase where $\typd{N}{D}$ is a constant $\typd{c}{C}$. Then one has:
\begin{align}
& \big(\cbd{v}{A}{M_0 M_1}\big) [u:=N] \equiv \big(\ctC{A}{C}{B}{\big(\cbd{v}{A}{M_0} \big) }{M_1}{}{}  \big)[u:=N] \notag \\
& \equiv  \ctC{A}{C}{B}{\big(\big(\cbd{v}{A}{M_0}\big)[u:=N]\big)}{M_1[u:=N]}{}{}   \equals{w}{\param}  \ctC{A}{C}{B}{\big(\cbd{v}{A}{M_0[u:=N]} \big) }{M_1[u:=N]}{}{} \notag  \\
& \equiv  \ctC{A}{C}{B}{\big(\cbd{v}{A}{M_0[u:=N]} \big) }{N}{}{} \equals{w}{\param}  \ctD{c}{A}{C}{B}{\big(\cbd{v}{A}{M_0[u:=N]} \big) } \label{eqn:verydifferent} \\
& \equiv \cbd{v}{A}{\big(M_0[u:=N] N\big)} \equiv \cbd{v}{A}{\big(M_0[u:=N]M_1[u:=N]\big)}\notag \\
&\equiv \cbd{v}{A}{\big((M_0M_1)[u:=N]\big)}\notag 
\end{align}
In this, the first and third-to-last are respectively by applications of the Cardinal and Dardinal case Definition~\ref{defn:cbd}(\ref{defn:cbd:3}\ref{defn:cbd:3.1}). The first line with $\equals{w}{\param}$ follows from induction hypothesis for $\typd{M_0}{C\rightarrow B}$. The second line with $\equals{w}{\param}$, namely in (\ref{eqn:verydifferent}), follows from the Cardinal-to-Dardinal weak reduction (cf. Definition~\ref{defn:weakreduction} and Remark~\ref{rmk:cardinal2dardingal}).
\end{itemize}

Since their proofs are similar, we omit the last two subcases (namely: when $\typd{M_1}{C}$ is a variable but not $\typd{v}{A}$ or $\typd{u}{D}$; and when $\typd{M_1}{C}$ is a constant).

A second case is when $C$ is a regular type. But then one uses Starling similar to how one used Warbler above.

For (\ref{prop:verydiff:2}), the only place in the proof of (\ref{prop:verydiff:1}) where a weak equality $\equals{w}{\param}$ rather than a syntactic identity $\equiv$ occurred, outside of an induction hypothesis, was in (\ref{eqn:verydifferent}), namely the case of where $\typd{N}{D}$ is a constant $\typd{c}{C}$ of state type $C$. Hence in a signature with no constants of state type, we will have syntactic identity $\equiv$ throughout.
\end{proof}

The following proposition is the combinatory logic analogue of distanced $\beta$-equality from Definition~\ref{defn:betareduction} (cf. \cite[Theorem 2.21 p. 27]{Hindley2008-vw}). It one of the key components of the proof of the correspondence between $\cltup$ and $\mltup$ established in Theorem~\ref{thm:frommltuptoclt}.

\begin{prop}[The combinatory logic analogue of distanced $\beta$-equality]\label{prop:keybackandforth}
Suppose that $\typd{L}{C}, \typd{\vec{M}}{\vec{B}}, \typd{N}{A}$ are terms of $\cltup$, and suppose that $\typd{\vec{x}}{\vec{B}}$ and $\typd{v}{A}$ are variables of $\cltup$ with $\typd{\vec{x}}{\vec{B}}, \typd{\vec{M}}{\vec{B}}$ having the same length. Suppose that
\begin{enumerate}[leftmargin=*]
    \item\label{defn:betareduction:3cl} the variables in $\typd{\vec{x}}{\vec{B}}$ do not appear in $\typd{N}{A}$
    \item\label{defn:betareduction:2cl} the variables in $\typd{\vec{x}}{\vec{B}}, \typd{v}{A}$ are pairwise distinct
\end{enumerate}
Then $\big( \cbd{\vec{x}}{\vec{B}}{\cbd{v}{A}{L}}\big) \vec{M} N \equals{w}{\param} \big( \cbd{\vec{x}}{B}{L[v:=N]} \big) \vec{M}$.
\end{prop}
\begin{proof}
It suffices to show that both sides $\transreduces{w}{\param}$-reduce to $L[\vec{x}:=\vec{M}, v:=N]$. Note that by condition~(\ref{defn:betareduction:2cl}) the simultaneous substitution in $L[\vec{x}:=\vec{M}, v:=N]$ is well-defined.

First we work on the left-hand side. By iterated applications of Proposition~\ref{prop:clbetapre}(\ref{prop:clbetapre:1}), we have $(\cbd{\vec{x}}{\vec{B}}{\cbd{v}{A}{L}})\vec{x}v \transreduces{w}{\param} L$. By Lemma~\ref{lem:substitution}(\ref{lem:substitution:3}) we have $\big(\,(\cbd{\vec{x}}{\vec{B}}{\cbd{v}{A}{L}})\vec{x}v \,\big)[v:=N] \transreduces{w}{\param} L[v:=N]$. By Proposition~\ref{prop:clbetapre}(\ref{prop:clbetapre:2}), $\typd{v}{A}$ does not appear in $\cbd{\vec{x}}{\vec{B}}{\cbd{v}{A}{L}}$; and by (\ref{defn:betareduction:2cl}) each variable in $\typd{\vec{x}}{\vec{B}}$ is distinct from $\typd{v}{A}$. Hence the left-hand side of the previous~$\transreduces{w}{\param}$ can be simplified to $(\cbd{\vec{x}}{\vec{B}}{\cbd{v}{A}{L}})\vec{x}N \transreduces{w}{\param} L[v:=N]$. By Lemma~\ref{lem:substitution}(\ref{lem:substitution:3}) again, we have $\big(\,(\cbd{\vec{x}}{\vec{B}}{\cbd{v}{A}{L}})\vec{x}N\,\big)[\vec{x}:=\vec{M}] \transreduces{w}{\param} \big(L[v:=N]\big)[\vec{x}:=\vec{M}]$. By Proposition~\ref{prop:clbetapre}(\ref{prop:clbetapre:2}), the variables $\typd{\vec{x}}{\vec{B}}$ do not appear in $\cbd{\vec{x}}{\vec{B}}{\cbd{v}{A}{L}}$; and by (\ref{defn:betareduction:3cl}) the variables $\typd{\vec{x}}{\vec{B}}$ do not appear in $N$. Hence we can simplify on both the left and right sides of the previous~$\transreduces{w}{\param}$ as follows: $\big(\cbd{\vec{x}}{\vec{B}}{\cbd{v}{A}{L}}\big)\vec{M}N \transreduces{w}{\param} L[\vec{x}:=\vec{M}, v:=N]$.

Second we work on the right-hand side. By iterated applications of  Proposition~\ref{prop:clbetapre}(\ref{prop:clbetapre:1}), we have $ \big( \cbd{\vec{x}}{\vec{B}}{L[v:=N]} \big) \vec{x} \transreduces{w}{\param} L[v:=N]$.  By Lemma~\ref{lem:substitution}(\ref{lem:substitution:3}) $\big(\, \big( \cbd{\vec{x}}{\vec{B}}{L[v:=N]} \big) \vec{x} \,\big)[\vec{x}:=\vec{M}] \transreduces{w}{\param} \big( L[v:=N]\big)[\vec{x}:=\vec{M}]$. By Proposition~\ref{prop:clbetapre}(\ref{prop:clbetapre:2}), the variables in  $\typd{\vec{x}}{\vec{B}}$ do not appear in $\cbd{\vec{x}}{B}{L[v:=N]}$; by (\ref{defn:betareduction:3cl}) we have that the variables $\typd{\vec{x}}{\vec{B}}$ do not appear in $N$. Hence we can simplify on both the left and right of the previous~$\transreduces{w}{\param}$ to obtain: $\big( \cbd{\vec{x}}{B}{L[v:=N]} \big) \vec{M} \transreduces{w}{\param}  L[\vec{x}:=\vec{M}, v:=N]$.
\end{proof}

The following is an elementary observation but one which is important for later understanding the interaction of $\alpha$-conversion and combinatory logic.
\begin{prop}[Action of type-preserving permutations of variables on terms of combinatory logic]\label{prop:clperm}
Given a type-preserving permutation $\pi$ of the variables of $\cltup$, we extend it to a type-preserving permutation of all terms of $\cltup$ by letting it be the identity on the constants and combinatory terms, and by setting $(MN)^{\pi}$ to be $M^{\pi} N^{\pi}$. Then
\begin{enumerate}[leftmargin=*]
    \item\label{prop:clperm:1} If $\pi$ is the identity on all variables appearing in $M$, then $M^{\pi}$ is $M$.
    \item\label{prop:clperm:2} If $\pi(\typd{v}{A})$ is $\typd{u}{A}$, then $(\cbd{v}{A}{M})^{\pi}$ is $\cbd{u}{A}{M^{\pi}}$.
\end{enumerate}
\end{prop}
\begin{proof}
For (\ref{prop:clperm:1}) this follows from construction, since $M^{\pi}$ is already the identity on the constants and combinatory terms.

For (\ref{prop:clperm:2}) this is by induction on complexity of $\typd{M}{B}$, using Definition~\ref{defn:cbd}. We omit the proof since it is routine.
\end{proof}

\begin{proofdetail}
\begin{proof}
But if one wanted to see the detail, we record it here. Again, the proof of (\ref{prop:clperm:2}) is by induction on complexity of $\typd{M}{B}$, using Definition~\ref{defn:cbd}.

Suppose $\typd{v}{A}$ does not appear in $\typd{M}{B}$. Then $\cbd{v}{A}{M}$ is $\ctK{B}{A}{M}{}$ and so $(\cbd{v}{A}{M})^{\pi}$ is $\ctK{B}{A}{M^{\pi}}{}$. Further since $\pi$ is a permutation, $\typd{u}{A}$ does not appear in $\typd{M^{\pi}}{B}$, and so  $\cbd{u}{A}{M^{\pi}}$ is also $\ctK{B}{A}{M^{\pi}}{}$.

Suppose $\typd{v}{A}$ does appear in $\typd{M}{B}$ and $\typd{M}{B}$ is $\typd{v}{A}$. Then $\cbd{v}{A}{M}$ is $\ctI{A}{}$. Since $\ctI{A}{}$ is a closed term of $\cltup$ (by Proposition~\ref{prop:clrecoveryidentity}), by (\ref{prop:clperm:1}) one has that $(\cbd{v}{A}{M})^{\pi}$ is also $\ctI{A}{}$. Further $\typd{u}{A}$ does appear in $\typd{M^{\pi}}{B}$ and $\typd{M^{\pi}}{B}$ is $\typd{u}{A}$, and hence $\cbd{u}{A}{M^{\pi}}$ is also~$\ctI{A}{}$.

Suppose $\typd{v}{A}$ does appear in $\typd{M}{B}$ and $\typd{M}{B}$ is $\typd{M_0 M_1}{B}$ where $\typd{M_0}{C\rightarrow B}$ and $\typd{M_1}{C}$. Then $\typd{u}{A}$ does appear in $\typd{M^{\pi}}{B}$ and $\typd{M^{\pi}}{B}$ is $\typd{M_0^{\pi} M_1^{\pi}}{B}$. There are two subcases.

Suppose that $C$ is a state type. Further first suppose that $\typd{M_1}{C}$ is $\typd{v}{A}$. Then $\cbd{v}{A}{M_0 M_1}$ is $\ctW{A}{B}{\big(\cbd{v}{A}{M_0}\big)}$, and so  $(\cbd{v}{A}{M_0 M_1})^{\pi}$ is $\ctW{A}{B}{\big(\cbd{v}{A}{M_0}\big)^{\pi}}$. Further, $\typd{M_1^{\pi}}{C}$ is $\typd{u}{A}$. Then $\cbd{u}{A}{M_0^{\pi} M_1^{\pi}}$ is $\ctW{A}{B}{\big(\cbd{u}{A}{M_0^{\pi}}\big)}$. Then we are done by induction hypothesis on $M_0$. The other subcases, depending on the other possibilities for $\typd{M_1}{C}$, is similar and so we omit it.

The case when $C$ is a regular type is similar and so we omit it.
\end{proof}
\end{proofdetail}

%% file: 06-pure-combinatory/04-lambda-calculus-to-CL.tex
\subsection{From lambda calculus to combinatory logic}\label{sec:fromltocl}

We then define a translation from terms of $\mltup$ to terms of $\cltup$ in the natural way (cf. \cite[Definition 9.10 p. 95]{Hindley2008-vw}).

\begin{defi}[Translation from lambda calculus to combinatory logic]\label{defn:cltranslation}
If $\typd{M}{B}$ is a term of $\mltup$, then we define a term $\typd{M^{cl}}{B}$ of $\cltup$ inductively as follows:
\begin{enumerate}[leftmargin=*]
    \item\label{defn:cltranslation:1} If $\typd{M}{B}$ is a variable or constant, then $\typd{M^{cl}}{B}$ is $\typd{M}{B}$.
    \item\label{defn:cltranslation:2} If $\typd{M}{B}$ is $\typd{M_0M_1}{B}$ where $\typd{M_0}{C\rightarrow B}$ and $\typd{M_1}{C}$, then  $\typd{M^{cl}}{B}$ is $\typd{M_0^{cl} M_1^{cl}}{B}$.
    \item\label{defn:cltranslation:3} If $\typd{M}{B}$ is $\lbd{v}{A}{L}$ where $\typd{L}{C}$, then   $\typd{M^{cl}}{B}$ is $\cbd{v}{A}{L^{cl}}$
\end{enumerate}
\end{defi}

The following is obvious from the construction and Proposition~\ref{prop:clbetapre}(\ref{prop:clbetapre:2}) but we include it for ease of future reference:
\begin{prop}[Free variables and appearance of variables under the translation]\label{prop:variablebackandforth}
If $\typd{M}{B}$ is a term of $\mltup$ and $\typd{v}{A}$ is a variable of $\mltup$, then $\typd{v}{A}$ appears free in $\typd{M}{B}$ iff $\typd{v}{A}$ appears in $\typd{M^{cl}}{B}$.
\end{prop}

\begin{proofdetail}

\begin{proof}

This is by induction on complexity of $\typd{M}{B}$. The base case for variables and constants is trivial since in this case $\typd{M^{cl}}{B}$ is $\typd{M}{B}$. 

The inductive step for application holds because the free variables of $M_0M_1$ are the union of those of $M_0$ and $M_1$; while the variables appearing in $M_0^{cl} M_1^{cl}$ are the union of the those appearing in $M_0^{cl}$ and $M_1^{cl}$. 

For the induction step for lambda abstraction, first suppose that $\typd{v}{A}$ is free in $\lbd{u}{B}{L}$. Then $\typd{v}{A}$ is distinct from $\typd{u}{B}$, and $\typd{v}{A}$ appears free in $L$. Then by induction hypothesis for $L$, we have that $\typd{v}{A}$ appears in $L^{cl}$. Then by Proposition~\ref{prop:clbetapre}(\ref{prop:clbetapre:2}) and the distinctness of $\typd{v}{A}$ and $\typd{u}{B}$, we have that $\typd{v}{A}$ appears in $\cbd{u}{A}{L^{cl}}$. 

Conversely, suppose that $\typd{v}{A}$ appears in  $\cbd{u}{A}{L^{cl}}$. Then by Proposition~\ref{prop:clbetapre}(\ref{prop:clbetapre:2}), we have that $\typd{v}{A}$ is distinct from $\typd{u}{B}$ and that $\typd{v}{A}$ appears in $L^{cl}$. Then by induction hypothesis for $L$ we have that $\typd{v}{A}$ appears free in $L$. By the distinctness of $\typd{v}{A}$ and $\typd{u}{B}$, we have that $\typd{v}{A}$ appears free in $\lbd{u}{B}{L}$.
\end{proof}

\end{proofdetail}

In this next proposition, we work with type-preserving permutations of the variables of $\mltup$, which recall are extended to type-preserving permutations of the terms of $\mltup$, as part of our official definition of $\alpha$-conversion (cf. Definition~\ref{defn:alphaeq}).
\begin{prop}[Commutativity of type-preserving permutations and translation]\label{prop:commpermtran}
For all terms $\typd{M}{A}$ of $\mltup$, for all type-preserving permutations $\pi$ of the variables of $\mltup$, one has that $(M^{\pi})^{cl}$ is $(M^{cl})^{\pi}$.
\end{prop}
\begin{proof}
If $M$ is a variable or constant, then both of these are equal to $M^{\pi}$.

The induction step for application is trivial.

Suppose that $\pi(\typd{v}{A})$ is $\typd{u}{A}$. Then $((\lbd{v}{A}{M})^{\pi})^{cl}$ is by definition $\cbd{u}{A}{(M^{\pi})^{cl}}$, which by induction hypothesis is $\cbd{u}{A}{(M^{cl})^{\pi}}$, which by Proposition~\ref{prop:clperm}(\ref{prop:clperm:2}) is  $(\cbd{v}{A}{(M^{cl})})^{\pi}$, which by definition is $((\lbd{v}{A}{M})^{cl})^{\pi}$.
\end{proof}

Since, formally, terms of $\mltup$ are $\alpha$-equivalence classes of terms, and since the translation in Definition~\ref{defn:cltranslation} is defined on members of these equivalence classes, we need to check that the translation respects the equivalence. This is so in a very strong form: ``the analogue in $CL$ of the $\lambda$-calculus relation of $=_{\alpha}$ is simply identity'' (\cite[p. 29]{Hindley2008-vw}):
\begin{prop}[Under the translation, $\alpha$-equivalence is identity]\label{prop:alphaisidentity}\hfill
\begin{enumerate}[leftmargin=*]
    \item\label{prop:alphaisidentity:1} Suppose $\typd{M}{B}$ is a term of $\mltup$. Suppose $\pi$ is a type-preserving permutation of the variables of $\mltup$ which is the identity on the free variables of $\typd{M}{B}$. Then $M^{cl}$ is $(M^{\pi})^{cl}$.
    \item\label{prop:alphaisidentity:2} Suppose $\typd{M,N}{B}$ are terms of $\mltup$ which are $\alpha$-equivalent. Then $\typd{M^{cl},N^{cl}}{B}$ are terms of $\cltup$ which are identical.
\end{enumerate}
\end{prop}
\begin{proof}
For (\ref{prop:alphaisidentity:1}), this is by an induction on complexity of $\typd{M}{B}$, with a universal quantifier over type-preserving permutations in the statement of the induction hypothesis.

If $\typd{M}{B}$ is a variable, then $M^{\pi}$ is $M$ since by hypothesis $\pi$ is the identity on the free variables of $M$; and then both $M^{cl}, (M^{\pi})^{cl}$ are $M$ as well.

If $\typd{M}{B}$ is a constant, then both $M^{cl}, (M^{\pi})^{cl}$ are $M$.

The induction step for application is trivial. 

For the induction step for lambda abstraction $\lbd{v}{A}{M}$, suppose that $\pi(\typd{v}{A})$ is $\typd{u}{A}$. Then by definition, $((\lbd{v}{A}{M})^{\pi})^{cl}$ is $\cbd{u}{A}{(M^{\pi})^{cl}}$. By Proposition~\ref{prop:commpermtran} the latter is $\cbd{u}{A}{(M^{cl})^{\pi}}$. By Proposition~\ref{prop:clperm}(\ref{prop:clperm:2}), this is identical to $(\cbd{v}{A}{M^{cl}})^{\pi}$. By definition this is identical to $((\lbd{v}{A}{M})^{cl})^{\pi}$. Since $\pi$ is the identity on the free variables of $\lbd{v}{A}{M}$, by Proposition~\ref{prop:variablebackandforth} we have that $\pi$ is the identity on all the variables appearing in $(\lbd{v}{A}{M})^{cl}$. Then by Proposition~\ref{prop:clperm}(\ref{prop:clperm:1}), we have that $((\lbd{v}{A}{M})^{cl})^{\pi}$ is identical to $(\lbd{v}{A}{M})^{cl}$.

For (\ref{prop:alphaisidentity:2}), this follows from (\ref{prop:alphaisidentity:1}) and the definition of  $\alpha$-equivalence, in Definition~\ref{defn:alphaeq}, as the smallest equivalence relation containing the compatible closure of the relation defined in terms of permutations. The base case of the induction is handled by (\ref{prop:alphaisidentity:1}), and the inductive steps are trivial and so we omit them.
\begin{proofdetail}
    
Here are the details of (\ref{prop:alphaisidentity:2}).

For (\ref{prop:alphaisidentity:2}), recall that $\alpha$-equivalence $\equals{\alpha_A}{\param}$ is defined, in Definition~\ref{defn:alphaeq}, as the smallest equivalence relation on terms of $\mltup$ of type $A$ containing the compatible closure $\reduces{\alpha_A}{\param}$ of the binary relation $\alpha_A$ given by: the ordered pair $(M,N)$ of terms of $\mltup$ of type $A$ stands in the $\alpha_A$-relation iff $N$ is $M^{\pi}$ for some type-preserving permutation $\pi$ of the variables of $\mltup$ which is the identity on the free variables of $M$. 

First, we show by induction that if $\typd{M,N}{A}$ are terms of $\mltup$ in the compatible closure $M\reduces{\alpha_A}{\param} N$, then $M^{cl}$ and $N^{cl}$ are identical:
\begin{itemize}[leftmargin=*]
    \item For the base case when $(M,N)$ are terms of $\mltup$ of type $A$ which stand in the $\alpha_A$ relation, we are done by (\ref{prop:alphaisidentity:1}). 
    \item Suppose that the result holds for $\typd{M_0}{B\rightarrow C}$ and that $M_0\reduces{\alpha_{B\rightarrow C}}{\param} N_0$. Then by induction hypothesis, we have $M_0^{cl}, N_0^{cl}$ are terms of $\cltup$ which are identical. Then for any term $\typd{P}{B}$ of $\mltup$ we have that $(M_0 P)^{cl}$ and $(N_0 P)^{cl}$ are identical since they are respectively equal to $M_0^{cl} P^{cl}$ and $N_0^{cl}, P^{cl}$.
    \item Suppose that the result holds for $\typd{M_1}{B}$ and that $M_1\reduces{\alpha_{B}}{\param} N_1$. Then by induction hypothesis, we have $M_1^{cl}, N_1^{cl}$ are terms of $\cltup$ which are identical. Then for any term $\typd{P}{B\rightarrow C}$ of $\mltup$ we have that $(P M_0)^{cl}$ and $(P N_0)^{cl}$ are identical since they are respectively equal to $P^{cl} M_0^{cl}$ and $P^{cl} N_0^{cl}$.
    \item Suppose that the result holds for $\typd{M}{C}$ and that  $M\reduces{\alpha_{C}}{\param} N$. Then by induction hypothesis, we have $M^{cl}$ and $N^{cl}$ are terms of $\cltup$ which are identical. Since these are identical, we have that  $\cbd{v}{B}{M^{cl}}$ and $\cbd{v}{B}{N^{cl}}$ are identical. Then we are done since by definition we have $(\lbd{v}{B}{M})^{cl}$ is $\cbd{v}{B}{M^{cl}}$, while $(\lbd{v}{B}{N})^{cl}$ is $\cbd{v}{B}{N^{cl}}$.
\end{itemize}

Second, we show by induction that if $\typd{M,N}{A}$ are terms of $\mltup$ in the compatible closure $M\equals{\alpha_A}{\param} N$, then $M^{cl}$ and $N^{cl}$ are identical:
\begin{itemize}[leftmargin=*]
    \item One base case is when $M\reduces{\alpha_A}{\param} N$. But then we are done by the previous argument.
    \item A second base case, for reflexivity, is when $M$ and $N$ are already equal. But then trivially $M^{cl}$ and $N^{cl}$ are equal.
    \item One induction step, for symmetry, supposes that $M\equals{\alpha_A}{\param} N$ because $N\equals{\alpha_A}{\param} M$. By induction hypothesis,  $N^{cl}$ and $M^{cl}$ are identical, and so trivially  $M^{cl}$ and $N^{cl}$ are identical.
    \item A second induction step, for transitivity, supposes that $M\equals{\alpha_A}{\param} P$ because $M\equals{\alpha_A}{\param} N$ and $N\equals{\alpha_A}{\param} P$. Then by induction hypothesis one has that   $M^{cl}$ and $N^{cl}$ are identical, and that $N^{cl}$ and $P^{cl}$ are identical. Then $M^{cl}$ and $P^{cl}$ are identical.
\end{itemize}
\end{proofdetail}
\end{proof}

The following is another technical proposition (cf. \cite[Lemma 9.13(d) p. 97]{Hindley2008-vw}). We need it to establish our correspondence between $\mltup$ and $\cltup$ later in Theorem~\ref{thm:frommltuptoclt}:
\begin{prop}[Substitution in lambda calculus mirrored by substitution in combinatory logic, in a signature with no constants of state type]\label{prop:missinglink}
Suppose that the signature has no constants of state type.

Suppose that $\typd{L}{C}$ is a term of $\mltup$. Suppose that $\typd{v}{A}$ is a variable of $\mltup$ and suppose that $\typd{N}{A}$ is a term of $\mltup$. 

If $\typd{N}{A}$ is free for $\typd{v}{A}$ in $\typd{L}{C}$, then $(L[v:=N])^{cl} \equiv L^{cl}[v:=N^{cl}]$.
\end{prop}
\begin{proof}
This is by induction on complexity of term $\typd{L}{C}$ of $\mltup$. The base cases and induction step for application are trivial. The interesting induction step for lambda abstraction is when $\typd{L}{C}$ is $\lbd{u}{D}{M}$, where $\typd{u}{D}$ is distinct from $\typd{v}{A}$, and $\typd{M}{B}$ is a term of $\mltup$ with $B$ regular. Suppose that $\typd{N}{A}$ is free for $\typd{v}{A}$ in $\lbd{u}{D}{M}$. Then we have that $\typd{N}{A}$ is free for $\typd{v}{A}$ in $M$. The non-trivial case to consider is when $\typd{v}{A}$ occurs free in $\typd{M}{B}$. Then since $\typd{N}{A}$ is free for $\typd{v}{A}$ in $\lbd{u}{D}{M}$, we have that $\typd{u}{D}$ does not occur free in term $\typd{N}{A}$ of $\mltup$. Then by Proposition~\ref{prop:variablebackandforth} we have that $\typd{u}{D}$ does not appear in term $\typd{N^{cl}}{A}$ of $\cltup$. Then:
\begin{align*}
& \big((\lbd{u}{D}{M})[v:=N]\big)^{cl}  \equiv \big(\lbd{u}{D}{M[v:=N]}\big)^{cl}  \equiv \cbd{u}{D}{\big( M[v:=N]\big)^{cl}}\\
& \equiv \cbd{u}{D}{\big( M^{cl}[v:=N^{cl}]\big)} \equiv \big(\cbd{u}{D}{M^{cl}}\big)[v:=N^{cl}] \equiv \big(\lbd{u}{D}{M}\big)^{cl}[v:=N^{cl}]
\end{align*} 
On the second line, the first $\equiv$ follows by induction hypothesis on $M$; and the second $\equiv$ follows by Proposition~\ref{prop:verydiff}(\ref{prop:verydiff:2}) and the hypothesis that we are in a signature with no constants of state type.
\begin{proofdetail}
Here is some more detail on the base case and induction steps. Again, this is by induction on complexity of term $\typd{L}{C}$ of $\mltup$.

Suppose $\typd{L}{C}$ is the variable $\typd{v}{A}$. Then $L[v:=N]$ is $N$ and so $(L[v:=N])^{cl}$ is $N^{cl}$. Further, $L^{cl}$ is $v$ and so $L^{cl}[v:=N^{cl}]$ is $N^{cl}$ as well.

Suppose $\typd{L}{C}$ is a distinct variable $\typd{u}{C}$. Then $L[v:=N]$ is $u$, and so $(L[v:=N])^{cl}$ is $u$. Further, $L^{cl}$ is $u$ and so $L^{cl}[v:=N^{cl}]$ is $u$ as well.

Suppose $\typd{L}{C}$ is a constant $\typd{c}{C}$. Then $L[v:=N]$ is $c$, and so $(L[v:=N])^{cl}$ is $c$. Further, $L^{cl}$ is $c$ and so $L^{cl}[v:=N^{cl}]$ is $c$ as well.

Suppose $\typd{L}{C}$ is $\typd{L_0L_1}{C}$ where $\typd{L_0}{B\rightarrow C}$ and $\typd{L_1}{B}$. Then we have that $\typd{N}{A}$ is free for $\typd{v}{A}$ in both $\typd{L_0}{B\rightarrow C}$ and $\typd{L_1}{B}$. Then one has
\begin{align*}
& ((L_0 L_1)[v:=N])^{cl} \equiv (L_0[v:=N] L_1[v:=N])^{cl} \equiv (L_0[v:=N])^{cl}  (L_1[v:=N])^{cl} \\
& \equiv L_0^{cl}[v:=N^{cl}] L_1^{cl}[v:=N^{cl}]\equiv  (L_0^{cl}L_1^{cl})[v:=N^{cl}] \equiv (L_0L_1)^{cl}[v:=N^{cl}] 
\end{align*}
The first step in the second line follows from induction hypothesis for both $\typd{L_0}{B\rightarrow C}$ and $\typd{L_1}{B}$.

Suppose $\typd{L}{C}$ is $\lbd{v}{A}{M}$,  where $\typd{M}{B}$ is a term of $\mltup$ with $B$ regular. Then $\typd{v}{A}$ does not occur free in $\lbd{v}{A}{M}$. Hence by Proposition~\ref{prop:variablebackandforth} one has that $\typd{v}{A}$ does not appear in $(\lbd{v}{A}{M})^{cl}$. Hence we have
\begin{equation*}
\big((\lbd{v}{A}{M})[v:=N]\big)^{cl} \equiv \big(\lbd{v}{A}{M}\big)^{cl} \equiv \big(\lbd{v}{A}{M}\big)^{cl}[v:=N^{cl}]
\end{equation*}
\end{proofdetail}
\end{proof}

\begin{thm}[Weak beta equalities translate to combinatory weak equalities, in a signature with no constants of state type]\label{thm:frommltuptoclt}
Suppose that the signature has no constants of state type.

Suppose that $\typd{M,N}{A}$ are terms of $\mltup$.  If $\mltup\vdash_{w\beta} M=N$ then $\cltup\vdash_w M^{cl}=N^{cl}$.
\end{thm}
\begin{proof}
Since the ${\cdot}^{cl}$-translation commutes with application, it suffices to show that an instance of $\beta$-reduction at the top level in $\mltup$ results in a weak equality in $\cltup$ under the ${\cdot}^{cl}$-translation. Suppose that we are given an instance of $\beta$-reduction in $\mltup$:
\begin{equation}
\big( \lbd{\vec{x}}{\vec{B}}{\lbd{v}{A}{L}}\big) \vec{M} N \reduces{\beta}{\param} \big( \lbd{\vec{x}}{B}{L[v:=N]} \big) \vec{M}
\end{equation}
so that as in Definition~\ref{defn:betareduction} we have:
\begin{enumerate}[leftmargin=*]
    \item\label{defn:betareduction:1apply} $\typd{N}{A}$ is free for $\typd{v}{A}$ in $\typd{L}{C}$
    \item\label{defn:betareduction:3apply} the variables in $\typd{\vec{x}}{\vec{B}}$ are not free in $\typd{N}{A}$
    \item\label{defn:betareduction:2apply} the variables in $\typd{\vec{x}}{\vec{B}},\typd{v}{A}$ are pairwise distinct.
\end{enumerate}
By (\ref{defn:betareduction:1apply}) and  Proposition~\ref{prop:missinglink} and the definition of the $^{cl}$-translation (cf. Definition~\ref{defn:cltranslation}), it suffices to show: $\big( \cbd{\vec{x}}{\vec{B}}{\cbd{v}{A}{L^{cl}}}\big) \vec{M}^{cl} N^{cl} \equals{w}{\param} \big( \cbd{\vec{x}}{B}{L^{cl}[v:=N^{cl}]} \big) \vec{M}^{cl}$. By Proposition~\ref{prop:keybackandforth} it suffices to show
\begin{enumerate}[leftmargin=*, label=(\alph*), ref=\alph*]
    \item\label{defn:betareduction:3clapply} the variables in $\typd{\vec{x}}{\vec{B}}$ do not appear in $\typd{N^{cl}}{A}$
    \item\label{defn:betareduction:2clapply}the variables in $\typd{\vec{x}}{\vec{B}},\typd{v}{A}$ are pairwise distinct.
\end{enumerate}
But (\ref{defn:betareduction:3clapply}) follows from (\ref{defn:betareduction:3apply}) and Proposition~\ref{prop:variablebackandforth}, while (\ref{defn:betareduction:2clapply}) follows directly from (\ref{defn:betareduction:2apply}).
\end{proof}

%% file: 06-pure-combinatory/05-CL-to-lambda-calculus.tex
\subsection{From combinatory logic to lambda calculus}\label{sec:fromcombotolambda}

We define a second translation (cf. \cite[Definition 9.2 p. 93]{Hindley2008-vw}):
\begin{defi}[Translation from combinatory logic to lambda calculus]\label{defn:lambdatranslation}
If $\typd{M}{B}$ is a term of $\cltup$, then we define a term $\typd{M^{\lambda}}{B}$ of $\mltup$ inductively as follows:
\begin{enumerate}[leftmargin=*]
    \item\label{defn:lambdatranslation:1} If $\typd{M}{B}$ is a variable or constant, then $\typd{M^{\lambda}}{B}$ is $\typd{M}{B}$.
    \item\label{defn:lambdatranslation:2} If $\typd{M}{B}$ is a combinatory term from Definition~\ref{defn:typedcombo2}, then $\typd{M^{\lambda}}{B}$ is the corresponding combinatory term $\mltup$ from Definition~\ref{defn:typedcombo}.
    \item\label{defn:lambdatranslation:3} If $\typd{M}{B}$ is $\typd{M_0M_1}{B}$ where $\typd{M_0}{C\rightarrow B}$ and $\typd{M_1}{C}$, then  $\typd{M^{\lambda}}{B}$ is $\typd{M_0^{\lambda} M_1^{\lambda}}{B}$.
\end{enumerate}
\end{defi}

In (\ref{defn:lambdatranslation:2}), we do not need to choose a specific version since they are all $\alpha$-equivalent.

We have the following analogue of Proposition~\ref{prop:variablebackandforth}, which is a routine induction:
\begin{prop}[Appearance of variables and free variables under the translation]\label{prop:variablebackandforth2}
If $\typd{M}{B}$ is a term of $\cltup$ and $\typd{v}{A}$ is a variable of $\cltup$, then $\typd{v}{A}$ appears in $\typd{M}{B}$ iff $\typd{v}{A}$ appears free in $\typd{M^{\lambda}}{B}$.
\end{prop}

\begin{proofdetail}
\begin{proof}
In more detail, this is by induction on complexity of $\typd{M}{B}$. 

The base case for variables and constants is trivial since in this case $\typd{M^{\lambda}}{B}$ is $\typd{M}{B}$. 

The base case for combinatory terms is trivial since in this case $\typd{M}{B}$ has no variables appearing in it, while $\typd{M^{\lambda}}{B}$ is closed and has no free variables appearing in it.

The inductive step for application holds because the variables appearing in $M_0M_1$ are the union of those of $M_0$ and $M_1$; while the free variables of $M_0^{\lambda} M_1^{\lambda}$ are the union of the those appearing in $M_0^{\lambda}$ and $M_1^{\lambda}$. 
\end{proof}
\end{proofdetail}

In the next results, we further pay attention to regular $\beta$-reductions (cf. Definition~\ref{defn:betareduction}(\ref{defn:betareduction:5})).

As in Theorem~\ref{thm:frommltuptoclt}, we have a correspondence theorem for the translations (cf. \cite[Lemma 9.5(b) pp. 93-94]{Hindley2008-vw}):
\begin{thm}[Combinatory weak equalities translate to $\beta$-equalities]\label{thm:fromclttolambda}
Suppose that $\typd{M,N}{A}$ are terms of $\cltup$. 

If $\cltup\vdash_w M=N$ then $\mltup\vdash_{\gamma} M^{\lambda}=N^{\lambda}$ for each of $\gamma=\beta_0, \beta_r, \beta$.
\end{thm}
\begin{proof}
By Proposition~\ref{prop:oneway}, it suffices to prove that if $M\reduces{w}{\param} N$ then $M^{\lambda}\transreduces{\beta_r}{\param} N^{\lambda}$. 

The base case of a direct weak reduction follows from Definition~\ref{defn:weakreduction} and Theorem~\ref{prop:combinatoryeffect} and Proposition~\ref{prop:thefirstcardinaltodardinal} and Definition~\ref{defn:lambdatranslation}(\ref{defn:lambdatranslation:2})-(\ref{defn:lambdatranslation:3}).

For the inductive step for application on the left suppose that $M P\reduces{w}{\param} N P$ because $M\reduces{w}{\param} N$. Then by induction hypothesis $M^{\lambda}\transreduces{\beta_r}{\param} N^{\lambda}$. Then $M^{\lambda}P^{\lambda}\transreduces{\beta_r}{\param} N^{\lambda} P^{\lambda}$, and we are done by Definition~\ref{defn:lambdatranslation}(\ref{defn:lambdatranslation:3}). The case of application on the right is similar.
\end{proof}

The following is a proposition that we need in order to round out, in the subsequent results, the treatment of the correspondence. It is in this proposition that we use $\eta$-reductions (cf. \cite[Lemma 9.16 p. 98]{Hindley2008-vw}). We further note, in the proof, where we lambda abstract over both sides of an equality, so that this proposition is not obviously available in the weak version of the lambda calculus.
\begin{prop}[Result on abstracts of translations of combinatory terms]\label{prop:littlehelper}
Suppose that $A,B$ are types and $B$ is regular. For all terms $\typd{M}{B}$ of $\cltup$, one has $\mltup\vdash_{\gamma\eta} (\cbd{v}{A}{M})^{\lambda}  = \lbd{v}{A}{M^{\lambda}}$ for each of $\gamma=\beta_0, \beta_r, \beta$.
\end{prop}
\begin{proof}
By Proposition~\ref{prop:clbetapre}(\ref{prop:clbetapre:1}), we have $\big(\cbd{v}{A}{M}\big)v \equals{w}{\param} M$. By Theorem~\ref{thm:fromclttolambda} we have $\big(\big(\cbd{v}{A}{M}\big)v\big)^{\lambda} \equals{\beta_r}{\param} M^{\lambda}$. By abstracting over $\typd{v}{A}$, we have $\lbd{v}{A}{\big(\big(\cbd{v}{A}{M}\big)v\big)^{\lambda}} \equals{\beta_r}{\param} \lbd{v}{A}{M^{\lambda}}$. By Definition~\ref{defn:lambdatranslation}(\ref{defn:cltranslation:1}),(\ref{defn:cltranslation:3}) we have that the left-hand side can be simplified to $\lbd{v}{A}{\big(\big(\cbd{v}{A}{M}\big)^{\lambda}v\big)} \equals{\beta_r}{\param} \lbd{v}{A}{M^{\lambda}}$. By Proposition~\ref{prop:oneway}, we have $\lbd{v}{A}{\big(\big(\cbd{v}{A}{M}\big)^{\lambda}v\big)} \equals{\beta_0}{\param} \lbd{v}{A}{M^{\lambda}}$. Then by $\eta$, we can further simplify the left-hand side to obtain  $(\cbd{v}{A}{M})^{\lambda} \equals{\beta_0\eta}{\param} \lbd{v}{A}{M^{\lambda}}$. The application of $\eta$ is legal due to Proposition~\ref{prop:clbetapre}(\ref{prop:clbetapre:2}) and Proposition~\ref{prop:variablebackandforth2}, which imply that $\typd{v}{A}$ is not free in $(\cbd{v}{A}{M})^{\lambda}$.
\end{proof}

The second part of the following forms a converse to Theorem~\ref{thm:frommltuptoclt} (cf. \cite[Theorem 9.17(c)-(d) p. 98]{Hindley2008-vw}):
\begin{thm}[Translates of combinatory weak equalities are beta-eta equalities]\label{thm:bigone}\hfill
\begin{enumerate}[leftmargin=*]
    \item \label{thm:bigone:1} For all types $B$ and all terms $\typd{M}{B}$ of $\mltup$, one has $\mltup \vdash_{\gamma\eta} (M^{cl})^{\lambda} = M$ for each of $\gamma=\beta_0, \beta_r, \beta$.
    \item \label{thm:bigone:2} For all terms $\typd{M,N}{A}$ $\mltup$, if $\cltup\vdash_w M^{cl}=N^{cl}$ then $\mltup\vdash_{\gamma\eta} M=N$ for each of $\gamma=\beta_0, \beta_r, \beta$.
\end{enumerate}
\end{thm}
\begin{proof}
For (\ref{thm:bigone:1}), this is by an induction on complexity of $\typd{M}{B}$.

For variables and constants, from Definition~\ref{defn:cltranslation}(\ref{defn:cltranslation:1}) followed by Definition~\ref{defn:lambdatranslation}(\ref{defn:lambdatranslation:1}), one has that $ (M^{cl})^{\lambda}$ is $M$ itself. 

For application, from Definition~\ref{defn:cltranslation}(\ref{defn:cltranslation:2}) followed by Definition~\ref{defn:lambdatranslation}(\ref{defn:lambdatranslation:3}), one has that  $ ((M_0 M_1)^{cl})^{\lambda}$ is $ (M_0^{cl})^{\lambda} (M_1^{cl})^{\lambda}$, and then we are done by induction hypothesis.

For lambda abstraction,  from Definition~\ref{defn:cltranslation}(\ref{defn:cltranslation:3}), we have that $(\lbd{v}{A}{M})^{cl}$ is $\cbd{v}{A}{M^{cl}}$. By Proposition~\ref{prop:littlehelper}, since $M^{cl}$ is a term of $\cltup$, we have that $(\cbd{v}{A}{M^{cl}})^{\lambda}$ is $\gamma\eta$-equivalent to $\lbd{v}{A}{(M^{cl})^{\lambda}}$ in $\mltup$, and we are done by induction hypothesis. 

For (\ref{thm:bigone:2}), suppose $\cltup \vdash_w M^{cl}=N^{cl}$. By Theorem~\ref{thm:fromclttolambda} we have that $\mltup \vdash_{\gamma} (M^{cl})^{\lambda}=(N^{cl})^{\lambda}$. Then by (\ref{thm:bigone:1}) we have that $\mltup \vdash_{\gamma\eta} M=N$. 
\end{proof}

%% file: 06-pure-combinatory/06-completeconserve.tex
\subsection{Conservation and the weak lambda calculus}\label{sec:pureclcompleteconserve}

We are finally in a position to prove our partial conservation result:

\thmconservemltpureovermltup*

\begin{proof}
Suppose $\mltpure\vdash_{w\beta} M=N$. By Theorem~\ref{thm:frommltuptoclt} we have $\cltpure \vdash_w M^{cl}=N^{cl}$. Since $M^{cl}, N^{cl}$ are terms of $\cltup$, we have by Corollary~\ref{cor:cltcon} that $\cltup \vdash_w M^{cl}=N^{cl}$. By Theorem~\ref{thm:bigone}(\ref{thm:bigone:2}) we have that $\mltup \vdash_{\beta\eta} M=N$. 
\end{proof}